\documentclass[10pt]{iopart}

\usepackage{iopams}
\usepackage{amssymb}
\usepackage{amsfonts}
\usepackage{amstext}
\usepackage{graphicx}
\usepackage{setstack}
\usepackage{amscd}
\usepackage{mathrsfs}
\usepackage{stmaryrd}
\usepackage{bbm}
\usepackage{color}

\newcommand{\ihbar}{\imath \hbar}
\newcommand{\Aut}{\mathrm{Aut}}
\newcommand{\R}{\Re\mathrm{e}}
\newcommand{\I}{\Im\mathrm{m}}
\newcommand{\llangle}{\langle \hspace{-0.2em} \langle}
\newcommand{\rrangle}{\rangle \hspace{-0.2em} \rangle}

\newcommand{\Sp}{\mathrm{Sp}}
\newcommand{\mod}{\ \mathrm{mod}\ }
\newcommand{\Te}{\mathbb{T}e}

\newtheorem{defi}{Definition}
\newtheorem{theo}{Theorem}
\newtheorem{prop}{Property}
\newtheorem{propo}{Proposition}

\newtheorem{lem}{Lemma}

\newenvironment{proof}{\noindent \textit{Proof:}}{\hfill $\Box$ \\}

\begin{document}

\title[Quasienergy states of quantum systems driven by a classical flow]{Schr\"odinger-Koopman quasienergy states of quantum systems driven by a classical flow}

\author{David Viennot \& Lucile Aubourg}
\address{Institut UTINAM (CNRS UMR 6213, Universit\'e de Bourgogne-Franche-Comt\'e, Observatoire de Besan\c con), 41bis Avenue de l'Observatoire, BP1615, 25010 Besan\c con cedex, France.}

\begin{abstract}
We study the properties of the quasienergy states of a quantum system driven by a classical dynamical system. The quasienergies are defined in a same manner as in light-matter interaction but where the Floquet approach is generalized by the use of the Koopman approach of dynamical systems. We show how the properties of the classical flow (fixed and cyclic points, ergodicity, chaos) influence the driven quantum system. This approach of the Schr\"odinger-Koopman quasienergies can be applied to quantum control, quantum information in presence of noises, and dynamics of mixed classical-quantum systems. We treat the example of a kicked spin ensemble where the kick modulation is governed by discrete classical flows as the Arnold's cat map and the Chirikov standard map.
\end{abstract}


\section{Introduction}
The concept of quasienergy has been introduced in the semiclassical theory of light-matter interactions in the sequel of the introduction of the Floquet theory by Shirley \cite{Shirley}. The Floquet theory and the concept of quasienergy states has then been used in various works \cite{Sambe, Barone, Guerin, Drese, Guerin2, Viennot}. A quantum system interacting with a laser field is described by a Schr\"odinger equation governed in the Hilbert space $\mathcal H$ by a periodic quantum Hamiltonian, as for example $H(\omega t) = H_0 + \mu E\cos(\omega t + \theta_0)$ (where $H_0$ is the free Hamiltonian, $\mu$ is the dielectric dipole, $E$ the laser field amplitude and $\omega$ the laser field frequency). The Floquet theory consists to consider the Schr\"odinger-Floquet equation governed by the Floquet Hamiltonian $H_F = -\ihbar \omega \frac{\partial}{\partial \theta} + H(\theta)$ living in the enlarged Hilbert space $\mathcal F = L^2(\mathbb S^1,\frac{d\theta}{2\pi}) \otimes \mathcal H$ (where $L^2(\mathbb S^1,\frac{d\theta}{2\pi})$ is the space of square-integrable functions on the circle $\mathbb S^1$). The quasienergy spectrum is $\Sp(H_F)$ and the quasienergy states are the eigenvectors of $H_F$. In contrast with $H$, $H_F$ (and then the quasienergies) is invariant under Weyl gauge transformations; quasienergy states represent the states of the quantum system dressed by the photons of the field; and the Floquet theory is closely related to the pure quantum theory of light-matter interactions \cite{Guerin3}. This approach has also been used for periodically kicked systems \cite{Haake}, quasiperiodic laser control (quantum system driven by multifrequency laser fields) \cite{Guerin2}, and to define (non-adiabatic) perdiodic geometric phases \cite{Moore1, Moore2, Moore3, Moore4, Moore5}.\\
The light-matter interaction example consists to a quantum system driven by a classical flow onto the circle, $\varphi^t \in \Aut(\mathbb S^1)$, defined by $\varphi^t(\theta) = \omega t + \theta \mod 2\pi$ ($\Aut(\mathbb S^1)$ is the space of the automorphisms of $\mathbb S^1$). It is possible to generalize the approach to any classical flow (non-necessarily periodic and eventually chaotic) by replacing the Floquet theorem by the Koopman approach of dynamical systems \cite{Koopman1, Koopman, Budisic, RS1, Lasota, Eisner}. This method has been used to define quantum Lyapunov exponents \cite{Sapin} and to define entanglement of mixed classical-quantum systems \cite{Jauslin}. These works focus on some properties of the Schr\"odinger-Koopman equation; in the present paper we want to study the physics supported by the quasienergy states involved by this approach. In particular, we want to show how some properties of the classical flow are ``transmitted'' to the driven quantum system. Schr\"odinger-Koopman quasienergy states can be used to study mixed classical-quantum systems \cite{Jauslin,Gay}, quantum information of open quantum systems where the classical flow modelizes the environmental noise \cite{Viennot2, Aubourg}, and quantum control problems \cite{Aubourg2}.\\
This paper is organized as follows. Section 2 recalls the principle of the Schr\"odinger-Koopman (SK) approach. This section is a review of kown results needed to understand the present paper. From section 3, we present new considerations and results which are not been considered in previous works. Section 3 is dedicated to the SK quasienergy states, their fundamental properties and how compute them. The role for the controlled quantum system of the fixed points, cycles and ergodic components of the classical system is explored. Section 4 studies the dynamics starting from a quasienergy state. In particular, we introduce a new geometric phase occuring for quantum system driven by an ergodic flow and we study the density matrix resulting from the entanglement between the quantum and the classical systems. The effects of the ergodic and mixing properties of the classical system onto the controlled quantum system are studied. Finally section 5 exhibits quasienergy states for quantum kicked spin systems where the kicks are modulated by three representative classical flows, a cyclic continuous automorphism of the torus (CAT) map, the Arnold's CAT map and the Chirikov standard map. The main result of this paper is the extension of the notion of quasienergy state to any classical flow and the obtention of their properties, which are presented sections 3 and 4. Section 6 presents a discussion concerning how these quasienergy states can be used in quantum control and quantum information problems.

\section{The Schr\"odinger-Koopman approach}
In this section, we recall some results usefull to understand the sequel of this paper. Some usefull results of the Koopman theory can also be found in \ref{appendixA}. More complete expositions of the Koopman theory can be found in \cite{Budisic, Lasota, Eisner}, and of the Schr\"odinger-Koopman approach in \cite{Sapin, Jauslin}. In this paper, we use the terminology ``Koopman approach'' for the use of the Koopman operator to treat a single classical system, whereas the terminology ``Schr\"odinger-Koopman approach'' is used for a quantum system controlled by a classical flow.

\subsection{The Koopman approach of dynamical systems}
\begin{defi}[Continuous time classical dynamical system]
A continuous time (autonomous) classical dynamical system is the three kinds of data $(\Gamma,\varphi^t,\mu)$ where $\Gamma$ is a topological space called the phase space, $\mathbb R^{(+)} \ni t \mapsto \varphi^t \in \mathrm{Aut}(\Gamma)$ is a one parameter continuous group of automorphisms of $\Gamma$ called the flow, and $\mu$ is a measure on $\Gamma$ defined with a $\sigma$-algebra $\mathscr T$. The dynamical system is said conservative if for all open set $A \in \mathscr T$, $\mu(\varphi^t(A)) = \mu(A)$.
\end{defi}
For convenience reasons, in this paper we consider that $\Gamma = \mathbb T^m$ ($m$-torus) and $\mu(\Gamma)=1$ with $\mathscr T$ the Borel $\sigma$-algebra. A point of $\Gamma$ is denoted by $\theta = (\theta^1,...,\theta^m)$ with the abuse of notation consisting to denoting a point with their local coordinates. Moreover we restrict our attention only on conservative dynamical systems.\\

Let $\theta(t) = \varphi^t(\theta_0)$ be a phase trajectory ($\theta_0 \in \Gamma$). $F \in \mathrm{Aut}(\Gamma)$ defined by
\begin{equation}
\dot \theta = F(\theta)
\end{equation}
is called the generator of the flow $\varphi^t$.

\begin{defi}[Koopman operator]
The Koopman operator of a dynamical system $(\Gamma,\varphi^t,\mu)$ is the linear operator $\mathcal T^t \in \mathcal L(L^2(\Gamma,d\mu))$ defined by
\begin{equation}
\forall f \in L^2(\Gamma,d\mu), \quad \mathcal T^t f(\theta) = f(\varphi^t(\theta))
\end{equation}
\end{defi}
$L^2(\Gamma,d\mu)$ is the space of square-integrable observables of the dynamical system. The Koopman operator permits to treat the nonlinear dynamics $\dot \theta = F(\theta)$ of the phase space as a linear dynamics on the space of the observables.

\begin{prop}
The linear generator of the Koopman operator is $F^\mu(\theta) \frac{\partial}{\partial \theta^\mu} \in \mathcal L(L^2(\Gamma,d\mu))$, i.e. $\mathcal T^t = e^{t F^\mu \partial_\mu}$.
\end{prop}
This property results from a direct application of the Stone theorem \cite{RS1}.\\

The Koopman operator is unitary and its generator is anti-selfadjoint for a conservative flow. We can note the interesting case of an Hamiltonian system with $\theta = (q,p)$ where $p_i$ is the conjugate momentum of $q^i$ and with
\begin{eqnarray}
\dot q^i & = & \frac{\partial \mathscr H(q,p)}{\partial p_i} \\ \dot p_i & = & - \frac{\partial \mathscr H(q,p)}{\partial q^i} 
\end{eqnarray}
where $\mathscr H \in \mathcal C^1(\Gamma)$ is the classical Hamiltonian of the dynamical system. In that case we have
\begin{eqnarray}
F^\mu \partial_\mu & = & \frac{\partial \mathscr H(q,p)}{\partial p_i} \frac{\partial}{\partial q^i} - \frac{\partial \mathscr H(q,p)}{\partial q^i}  \frac{\partial}{\partial p_i} \\
& = & \{ \cdot, \mathscr H \}
\end{eqnarray}
where $\{\cdot,\cdot\}$ denotes the Poisson braket.\\

Throughout this paper we use the Koopman eigenvalues $\lambda \in \Sp(F^\mu \partial_\mu)$ and $f_\lambda \in \mathcal C^1(\Gamma)$ their associated eigenfunctions.
\begin{eqnarray}
F^\mu(\theta) \frac{\partial f_\lambda(\theta)}{\partial \theta^\mu} & = & \lambda f_\lambda(\theta) \\
\mathcal T^t f_\lambda(\theta) & = & e^{\lambda t} f_\lambda(\theta)
\end{eqnarray}
$f_\lambda$ is called a Koopman mode.\\

In a same manner, we can define for a \textbf{discrete time classical dynamical system} $(\Gamma,\varphi,\mu)$ ($\varphi \in \Aut(\Gamma)$, $\theta_{n+1} = \varphi(\theta_n)$) a Koopman operator such that $\forall f \in L^2(\Gamma,d\mu)$, $\mathcal T f(\theta) = f(\varphi(\theta))$.\\

\subsection{The Schr\"odinger-Koopman equation}
% 2.1
\begin{defi}[Driven continuous time quantum system]
A continuous time quantum system driven by a classical dynamical system is the five kinds of data $(\Gamma,\mu,\varphi^t,\mathcal H,H)$ where $(\Gamma,\mu,\varphi^t)$ is a classical dynamical system, $\mathcal H$ is a quantum state Hilbert space and $\Gamma \ni \theta \mapsto H(\theta) \in \mathcal L(\mathcal H)$ is a familly of self-adjoint Hamiltonians strongly continuous with respect to $\theta$. The dynamics of the driven quantum system is governed by the Schr\"odinger equation:
\begin{equation}
\ihbar \frac{d\tilde \psi(t)}{dt} = H(\varphi^t(\theta_0)) \tilde \psi(t)
\end{equation}
with $\tilde \psi(t=0)=\tilde \psi_0 \in \mathcal H$ the initial condition for the quantum system, $\theta_0 \in \Gamma$ being the initial condition for the classical system.
\end{defi}
Generally we have $H(\theta) = H_0 + H_{ctrl}(\theta)$ where $H_0$ independent of $\theta$ is the Hamiltonian of the isolated quantum system (spin, atom, molecule, ...) and $H_{ctrl}(\theta)$ is a control Hamiltonian representing the action on the quantum system of a classical control system (electromagnetic fields, STM, classical medium out of equilibrium, ...) obeying to the dynamics of $(\Gamma,\mu,\varphi^t)$.\\

The definition can be extended to a discrete time dynamical system:
\begin{defi}[Driven stroboscopic quantum system]
  A stroboscopic quantum system driven by a classical dynamical system is the five kinds of data $(\Gamma,\mu,\varphi,\mathcal H,U)$ where $(\Gamma,\mu,\varphi)$ is a discrete time classical dynamical system, $\mathcal H$ is a quantum state Hilbert space and $\Gamma \ni \theta \mapsto U(\theta) \in \mathcal U(\mathcal H)$ is a family of unitary evolution operators strongly continuous with respect to $\theta$. The stroboscopic dynamics of the driven quantum system is governed by the equation
  \begin{equation}
    \tilde \psi_{n+1} = U(\varphi^n(\theta_0)) \tilde \psi_n
  \end{equation}
  with $\tilde \psi_0 \in \mathcal H$ the initial condition for the quantum system, $\theta_0 \in \Gamma$ beging the initial condition for the classical system.
\end{defi}
Such a system results from a quantum system governed by a time-dependent Hamiltonian as
\begin{equation}
  H(t) = H_0 + \sum_{n \in \mathbb N} W(\varphi^n(\theta_0)) \delta(t-nT)
\end{equation}
which corresponds to a system with free Hamiltonian $H_0$ periodically kicked by ultra-fast pulses with $W(\theta) \in \mathcal L(\mathcal H)$ the kicking operator (depending from the value of $\theta$). The single period evolution operator (from $nT$ to $(n+1)T$) is $U(\varphi^n(\theta_0)) =  e^{-\ihbar^{-1} H_0 T} e^{-\imath W(\varphi^n(\theta_0))}$. $\tilde \psi_n = \tilde \psi(nT)$ with $t \mapsto \tilde \psi(t)$ solution of the Schr\"odinger equation $\ihbar \frac{d\tilde \psi}{dt} = H(t) \tilde \psi(t)$. The series $\tilde \psi_n$ is called the stroboscopic evolution of the quantum system.

\begin{defi}[Mixed state]
Let $(\Gamma,\mu,\varphi^t,\mathcal H,H)$ be a driven quantum system, $\tilde \psi_0 \in \mathcal H$ be a quantum state and $\rho_0 \in L^1_+(\Gamma,d\mu)$ be a density of $\Gamma$ ($\rho_0(\theta) \in \mathbb R^+$, $\int_\Gamma \rho_0(\theta) d\mu(\theta)=1$). The mixed state associated with the initial quantum and statistical states $(\tilde \psi_0,\rho_0)$ is the density matrix
\begin{equation} \label{defdensitymatrix}
\rho(t) = \int_\Gamma |\tilde \psi(t;\theta) \rangle \langle \tilde \psi(t;\theta)| \rho_0(\theta) d\mu(\theta)
\end{equation}
where $\tilde \psi(t;\theta) \in \mathcal H$ is solution of the Schr\"odinger equation $\ihbar \frac{d\tilde \psi(t;\theta)}{dt} = H(\varphi^t(\theta)) \tilde \psi(t;\theta)$ with $\tilde \psi(0;\theta) = \tilde \psi_0$.
\end{defi}
$\rho \in \mathcal L(\mathcal H)$, $\rho^\dagger = \rho$, $\rho \geq 0$ and $\tr \rho = 1$. For a stroboscopic driven quantum system we have $\rho_n = \int_\Gamma |\tilde \psi_n(\theta) \rangle\langle \tilde \psi_n(\theta)|\rho_0(\theta)d\mu(\theta)$.

\begin{defi}[Schr\"odinger-Koopman Hamiltonian]
Let $(\Gamma,\mu,\varphi^t,\mathcal H,H)$ be a driven quantum system and let $\mathcal K = L^2(\Gamma,d\mu) \otimes \mathcal H$ be the ``enlarged'' Hilbert space. We call Schr\"odinger-Koopman Hamiltonian of the driven quantum system the operator $H_K \in \mathcal L(\mathcal K)$ defined by
\begin{equation}
H_K = -\ihbar F^\mu(\theta) \frac{\partial}{\partial \theta^\mu} \otimes 1_{\mathcal H} + H(\theta)
\end{equation}
where $F^\mu \partial_\mu$ is the Koopman generator.
\end{defi}
The enlarged Hilbert space $\mathcal K$ is endowed with the inner product $\llangle \psi|\phi \rrangle = \int_{\Gamma} \langle \psi(\theta)|\phi(\theta) \rangle d\mu(\theta)$ (where $\langle \cdot|\cdot \rangle$ denotes the inner product of $\mathcal H$). For example, consider a kicked diatomic molecule in a plane, where the vibration is treated as a quantum system and the rotation is treated as classical (classical kicked rotator). The enlarged Hilbert space is then $\mathcal K = L^2(\mathbb T^2,\frac{d\theta^1d\theta^2}{4\pi^2}) \otimes L^2(\mathbb R^+,dr)$ where $\theta^1$ is the angular position of the rotator, $\theta^2$ is the reduced momentum of the rotation, and $r$ is the internuclear distance. The classical flow can be a nonautonomous continuous time flow associated with the Hamilton equations for the classical Hamiltonian $\mathscr H(\theta^1,\theta^2,t) = \frac{(\theta^2)^2}{2} + K \cos(\theta^1) \sum_n \delta(t-n\tau)$ ($\tau$ being the kick period and $K$ being the kick strength), or a discrete time flow defined by the Chirikov standard map $\varphi(\theta^1,\theta^2) = (\theta^1 + \theta^2+K\sin(\theta^2), \theta^2+K\sin(\theta^2))$ (which is the stroboscopic evolution of the rotator).

\begin{theo} \label{SKE}
Let $(\Gamma,\mu,\varphi^t,\mathcal H,H)$ be a driven quantum system. Let $t \mapsto \psi(t) \in \mathcal K = L^2(\Gamma,d\mu) \otimes \mathcal H$ be a solution of the Schr\"odinger-Koopman equation:
\begin{equation}
\ihbar \frac{\partial \psi(\theta,t)}{\partial t} = H_K \psi(\theta,t)
\end{equation}
where $H_K$ is the Schr\"odinger-Koopman Hamiltonian. Then $\tilde \psi(t;\theta_0) = \psi(\varphi^t(\theta_0),t)$ is a solution of the usual Schr\"odinger equation.
\end{theo}

\begin{proof}
\begin{equation}
\ihbar \frac{\partial \psi(\theta,t)}{\partial t} = - \ihbar F^\mu(\theta) \frac{\partial \psi(\theta,t)}{\partial \theta^\mu} + H(\theta) \psi(\theta,t)
\end{equation}
It follows that
\begin{eqnarray}
\ihbar \frac{d\tilde \psi(t;\theta_0)}{dt} & = & \ihbar \frac{d\psi(\varphi^t(\theta_0),t)}{dt} \\
& = & \ihbar \left. \frac{\partial \psi(\theta,t)}{\partial t} \right|_{\theta = \varphi^t(\theta_0)} + \ihbar \left. F^\mu(\theta) \frac{\partial \psi(\theta,t)}{\partial \theta^\mu} \right|_{\theta = \varphi^t(\theta_0)} \\
& = & \left. H(\theta) \psi(\theta,t) \right|_{\theta=\varphi^t(\theta_0)} \\
& = & H(\varphi^t(\theta_0)) \tilde \psi(t;\theta_0)
\end{eqnarray}
\end{proof}

\begin{defi}[Schr\"odinger-Koopman evolution operator]
Let $(\Gamma,\mu,\varphi^t,\mathcal H,H)$ be a driven quantum system and $H_K \in \mathcal L(\mathcal K)$ be its Schr\"odinger-Koopman Hamiltonian. The Schr\"odinger-Koopman evolution operator of the driven quantum system is
\begin{equation}
U_K(t,0) = e^{-\ihbar^{-1} H_K t}
\end{equation}
\end{defi}
By construction, $\psi(t) = U_K(t,0) \psi(0) \iff \ihbar \frac{\partial \psi}{\partial t} = H_K \psi(t)$ with $\psi \in \mathcal K$. If the classical dynamical system is conservative then $H_K$ is self-adjoint and $U_K$ is unitary.

\begin{prop}
 Let $(\Gamma,\mu,\varphi^t,\mathcal H,H)$ be a conservative driven quantum system, $U_K(t,0) \in \mathcal U(\mathcal K)$ be its Schr\"odinger-Koopman evolution operator, $\mathcal T^t \in \mathcal U(L^2(\Gamma,d\mu))$ be the Koopman operator of the classical system and $U(t,0;\theta) \in \mathcal U(\mathcal H)$ be the evolution operator of the quantum system, i.e. the strongly continuous solution of the equation:
\begin{equation}
\ihbar \frac{dU(t,0;\theta)}{dt} = H(\varphi^t(\theta)) U(t,0;\theta) \qquad U(0,0;\theta) = 1_{\mathcal H}
\end{equation}
The three operators are related by
\begin{equation}
U_K(t,0) = \mathcal T^{-t} U(t,0;\theta) = U(t,0;\varphi^{-t}(\theta)) \mathcal T^{-t}
\end{equation}
\end{prop}

\begin{proof}
See \cite{Sapin}
\end{proof}

Remark: $U(t,t_1;\theta) = U(t-t_1,0;\varphi^{t_1}(\theta))$ (by a variable change $s=t-t_1$ in the Schr\"odinger equation).\\
For a stroboscopic quantum system we define directly the Koopman evolution operator as $U_K = \mathcal T^{-1} U(\theta) = U(\varphi^{-1}(\theta)) \mathcal T^{-1}$.
\begin{theo}
  Let $(\Gamma,\mu,\varphi,\mathcal H,U)$ be a stroboscopic driven quantum system. Let $\psi_n = U_K^n \tilde \psi_0$ be the stroboscopic Schr\"odinger-Koopman state. Then $\tilde \psi_n = \psi_n(\varphi^n(\theta_0))$ with $\tilde \psi_{n+1} = U(\varphi^n(\theta_0)) \tilde \psi_n$.
\end{theo}
\begin{proof}
  \begin{eqnarray}
    \tilde \psi_{n+1} & = & \left. U_K^{n+1}(\theta) \psi_0(\theta) \right|_{\theta = \varphi^{n+1}(\theta_0)} \\
    & = & \left. U_K(\theta) U_K^n(\theta) \psi_0(\theta) \right|_{\theta=\varphi^{n+1}(\theta_0)} \\
    & = & \left. U(\varphi^{-1}(\theta)) \mathcal T^{-1} \psi_n(\theta) \right|_{\theta=\varphi^{n+1}(\theta_0)} \\
    & = & \left. U(\varphi^{-1}(\theta)) \psi_n(\varphi^{-1}(\theta)) \right|_{\theta=\varphi^{n+1}(\theta_0)} \\
    & = & U(\varphi^n(\theta_0)) \psi_n(\varphi^n(\theta_0)) \\
    & = & U(\varphi^n(\theta_0)) \tilde \psi_n
  \end{eqnarray}
\end{proof}

\begin{prop}
Let $(\Gamma,\mu,\varphi^t,\mathcal H,H)$ be a conservative driven quantum system and $\rho(t)$ be its mixed state for the initial conditions $(\tilde \psi_0,\rho_0) \in \mathcal H \times L^1_+(\Gamma,d\mu)$ (as defined by equation \ref{defdensitymatrix}). Let $\Psi(\theta,t) = U_K(t,0) \sqrt{\rho_0(\theta)} \tilde \psi_0 \in \mathcal K$. We have
\begin{equation}
\rho(t) = \tr_{L^2(\Gamma,d\mu)} |\Psi(t) \rrangle \llangle \Psi(t)|
\end{equation}
\end{prop}

\begin{proof}
\begin{eqnarray}
\Psi(\theta,t) & = & \mathcal T^{-t} \sqrt{\rho_0(\theta)} \underbrace{U(t,0;\theta) \tilde \psi_0}_{\tilde \psi(t;\theta)} \\
& = & \sqrt{\rho_0(\varphi^{-t}(\theta))} \tilde \psi(t;\varphi^{-t}(\theta))
\end{eqnarray}
where $\tilde \psi(t;\theta)$ is the solution of the usual Schr\"odinger equation.
\begin{eqnarray}
& & \tr_{L^2(\Gamma,d\mu)} |\Psi(t) \rrangle \llangle \Psi(t)| \nonumber \\
& &  =\int_{\Gamma} \langle \theta|\Psi(t) \rrangle \llangle \Psi(t)|\theta\rangle d\mu(\theta) \\
& & =\int_\Gamma |\Psi(\theta,t)\rangle \langle \Psi(\theta,t)| d\mu(\theta) \\
& & =\int_\Gamma |\tilde \psi(t;\varphi^{-t}(\theta)) \rangle \langle \psi(t;\varphi^{-t}(\theta))| \rho_0(\varphi^{-t}(\theta)) d\mu(\theta) \\
& & =\int_{\varphi^{-t}(\Gamma)} |\tilde \psi(t;\theta) \rangle \langle \tilde \psi(t;\theta)| \rho_0(\theta) |\det(\partial \varphi^t_\theta)| \rho_\mu(\varphi^t(\theta)) d\theta^1...d\theta^m
\end{eqnarray}
where $\rho_\mu$ is the density of the measure $\mu$ and $\partial \varphi^t_\theta$ is the Jacobian matrix of $\varphi^t$ at $\theta$. But since the flow is measure-preserving, $|\det(\partial \varphi^t_\theta)| \rho_\mu(\varphi^t(\theta)) d\theta^1...d\theta^m = d\mu(\theta)$ and $\overline{\phi^{-t}(\Gamma)} = \Gamma$. It follows that
\begin{equation}
\tr_{L^2(\Gamma,d\mu)} |\Psi(t) \rrangle \llangle \Psi(t)| = \int_\Gamma |\tilde \psi(t;\theta)\rangle\langle \tilde \psi(t;\theta) \rho_0(\theta)d\mu(\theta) = \rho(t)
\end{equation}
\end{proof}
We have the same result for a stroboscopic driven quantum system with $\rho_n = \tr_{L^2(\Gamma,d\mu)} |\Psi_n\rrangle\llangle \Psi_n|$ where $\Psi_n(\theta) = U_K^n \sqrt{\rho_0(\theta)} \tilde \psi_0$.\\
It follows that the mixed state appears as the partial trace of a pure state of the enlarged Hilbert space $\mathcal K$. The set of the observables of the classical dynamical system $L^2(\Gamma,d\mu)$ plays the role of an environment inducing decoherence and relaxation on the quantum system, characterizing a kind of entanglement between the classical and the quantum systems. A discussion concerning entanglement between classical and quantum systems can be found in \cite{Jauslin}. In the present context, let $\Psi \in \mathcal K$ be the solution of the Schr\"odinger-Koopman equation. Let $(\zeta_i)_i$ be an orthonormal basis of $\mathcal H$, for example the eigenbasis of the isolated quantum system without control by the classical flow (eigenvectors of $H_0$ with $H(\theta) = 1_{L^2(\Gamma,d\mu)} \otimes H_0 + H_{ctrl}(\theta)$). Let a set of orthonormalized Koopman modes $(f_\lambda)_\lambda$ generating a subspace in $L^2(\Gamma,d\mu)$ in which the dynamics associated with $\Psi$ takes place. We can decompose the state in the enlarged Hilbert space onto the tensorial basis: $\Psi(\theta) = \sum_{i,\lambda} c_{i,\lambda}  f_\lambda(\theta) \otimes \zeta_i$. Since the evolution of $\Psi$ is governed by a Hamiltonian non-separable as a sum of an operator of $\mathcal H$ and of an operator of $L^2(\Gamma,d\mu)$, $\Psi$ is not a separable state for $t>0$ even if it is the case for $t=0$ (we cannot write $\Psi = g \otimes \phi$ for some $g \in L^2(\Gamma,d\mu)$ and $\phi \in \mathcal H$). The state $\Psi$ is then entangled for the mathematical viewpoint of the tensor Hilbert space $L^2(\Gamma,d\mu) \otimes \mathcal H = \mathcal K$. It is for this reason that $\rho = \tr_{L^2(\Gamma,d\mu)} |\Psi \rrangle \llangle \Psi| = \sum_{i,j} \sum_\lambda c_{i,\lambda} \overline{c_{j,\lambda}} |\zeta_i \rangle \langle \zeta_j|$ is a mixed state (and not a pure state), even if we start at $t=0$ with a seperable state $\Psi$ (a pure state $\rho$), implying a decoherence phenomenon due to the classical flow. Physically, the pure quantum entanglement corresponds to nonlocal correlations between two quantum systems, it is the entanglement between wave functions of the two systems. In the Schr\"odinger-Koopman picture, we have an entanglement between wave functions of a quantum system and Koopman modes of a classical system. Physically the Koopman modes define observables which are dynamically coherent onto the phase space (for example for a jet in crossflow, one Koopman mode corresponds to the shear-layer structures and another one to the wall structures of the vortex induced by the turbulences), see \cite{Budisic}. Koopman modes are then a kind of generalisation for any classical dynamical system of the notion of normal modes for the classical wave systems. The entanglement in the Schr\"odinger-Koopman approach then corresponds to correlations between the quantum system and the classical system viewed as a collection of observables generated by the Koopman modes. The correlation is nonlocal in the sense that it is associated with the whole of the classical phase space (the Koopman modes are functions with support extended (in general) on the whole of $\Gamma$). But this nonlocality results from a statistical uncertainty (a lack of information concerning the initial condition of the classical flow modelled by the classical statistical distribution $\rho_0 \in L^1_+(\Gamma,d\mu)$ in the previous property), and not from an intrisic uncertainty (induced by the fundamental quantum laws) as in the pure quantum case.

\section{Quasienergies}
Now we introduce the quasienergy states associated with the Koopman Hamiltonian. For the sake of simplicity, we consider that $\mathcal H$ is finite dimensional, the results can be adaptated to infinite dimensional Hilbert spaces (with some topological precautions).

\subsection{The quasienergy spectrum}
% 2.2.1 jusqu'a apres def 14
\begin{defi}[Quasienergie]
Let $(\Gamma,\mu,\varphi^t,\mathcal H,H)$ be a driven quantum system. We call quasienergies the eigenvalues of the Schr\"odinger-Koopman Hamiltonian:
\begin{equation}
(-\ihbar F^\mu(\theta) \frac{\partial}{\partial \theta^\mu} + H(\theta)) |a,\theta \rangle = \chi_a |a,\theta \rangle
\end{equation}
\end{defi}

It is interesting to consider the case of an Hamiltonian dynamical system with $\theta = (q,p)$ and $\mathscr H \in \mathcal C^1(\Gamma)$ the classical Hamiltonian. In that case we have
\begin{equation}
H_K = \ihbar \{\mathscr H, \cdot\} + H(\theta)
\end{equation}
The Schr\"odinger-Koopman Hamiltonian is then the sum of the quantum Hamitonian $H$ and a non-canonical quantized version of the classical Hamiltonian $\ihbar \{\mathscr H, \cdot\}$. It follows that the quasienergies represent the energies of the quantum system plus energies of the classical dynamical system. This remark is in accordance with the case of the Floquet theory of light-matter interaction. In that case $\Gamma = \mathbb S^1$ with $\varphi^t(\theta) = \theta + \omega t \mod 2\pi$, and $H(\theta) = H_0 + \mu E \sin(\theta)$. $H_K = -\ihbar \omega \frac{\partial}{\partial \theta} + H(\theta)$ and it is proved \cite{Guerin} that it is in a certain topology the limit of the pure quantum Hamiltonian $\hbar \omega a^+a + H_0 + \mu \sqrt{\frac{\hbar \omega}{2\epsilon_0 V}}\imath (a - a^+)$ when the average number of photons tends to $+\infty$ and the volume of the cavity $V$ tends to $+\infty$ ($a$ and $a^+$ are the photon annihilation and creation operators, $\epsilon_0$ is the vacuum permittivity and $\omega$ is the frequency of the photons). In this limit the photon Hamiltonian $\hbar \omega a^+a$ becomes the Koopman generator $- \ihbar \omega \frac{\partial}{\partial \theta}$.\\
As example, consider a quantum system of Hamiltonian $H_0$ perturbed by a classical harmonic oscillator system $\mathscr H(q,p) = \frac{p^2}{2m} + \frac{kq^2}{2}$ with $m$ the inertial parameter and $k$ the stiffness, with a perturbation $\epsilon V(\theta)$ dependent from the oscillator phase $\theta = \arctan \frac{p}{\sqrt{km} q}$. In that case, $\ihbar \{\mathscr H,\cdot\} = \ihbar \left(kq\frac{\partial}{\partial p} - \frac{p}{m} \frac{\partial}{\partial q}\right) = \ihbar \omega \frac{d}{d\theta}$ (with $\omega = \sqrt{\frac{k}{m}}$). The Koopman modes are then $f_m(\theta) = e^{\imath m \theta}$ ($m \in \mathbb Z$) with the associated eigenvalues $m\hbar \omega$. Finally the Schr\"odinger-Koopman Hamiltonian is $H_K = \sum_m m\hbar \omega |f_m\rangle\langle f_m| + H_0 + \epsilon V(\theta)$ and the quasienergy are (at the first order of perturbation) $\chi_{m,i} = m \hbar \omega + \nu_i + \epsilon \int_0^{2\pi} \langle i|V(\theta)|i\rangle \frac{d\theta}{2\pi} + \mathcal O(\epsilon^2)$ (where $H_0|i\rangle = \nu_i|i \rangle$).

\begin{prop}
Let $(\Gamma,\mu,\varphi^t,\mathcal H,H)$ be a conservative driven quantum system and $\chi_a \in \Sp(H_K)$ be a quasienergy associated with the eigenstate $|a,\theta\rangle \in \mathcal K$. Let $U(t,0;\theta) \in \mathcal U(\mathcal H)$ be the evolution operator of the quantum system. We have
\begin{equation}
\label{eigenU}
U(t,0;\theta) |a,\theta \rangle = e^{-\ihbar^{-1} \chi_a t} |a,\varphi^t(\theta)\rangle
\end{equation}
\end{prop}

\begin{proof}
$U_K(t,0) = e^{-\ihbar^{-1} H_K t} \Rightarrow U_K(t,0)|a,\theta\rangle = e^{-\ihbar^{-1} \chi_a t} |a,\theta \rangle$. But $U_K(t,0)= \mathcal T^{-t} U(t,0;\theta)$, it follows that $U(t,0;\theta)|a,\theta\rangle = e^{-\ihbar^{-1} \chi_a t} \mathcal T^t|a,\theta \rangle$.
\end{proof}

For a stroboscopic driven quantum system we define directly the quasienergy states by $U_K |a \rrangle = e^{-\imath \chi_a} |a \rrangle$ where $\chi_a$ is dimensionless. The last property takes then the form $U(\theta) |a,\theta \rangle = e^{-\imath \chi_a} |a,\varphi(\theta) \rangle$.

\begin{lem}[Orbital stability of the quasienergy spectrum] \label{orbstab1}
  Let $\theta_1,\theta_2 \in \Gamma$ be two distinct points of $\Gamma$. For $i=1,2$, let $\{\chi_{ai}\}_a$ be the set of the quasienergies such that $\exists (t \mapsto |ai,\varphi^t(\theta_i) \rangle \in \mathcal H \setminus\{0\})$ with $U(t,0;\theta_i)|ai,\theta_i \rangle = e^{-\ihbar^{-1} \chi_{ai} t}|ai,\varphi^t(\theta_i)\rangle$. If $\exists t_*$ such that $\varphi^{t_*}(\theta_1) = \theta_2$, then $\forall a$, $\exists b$, such that $\chi_{a1} = \chi_{b2}$.
\end{lem}

\begin{proof}
\begin{eqnarray}
  U(t,0;\theta_1) & = & U(t,t_*;\theta_1) U(t_*,0;\theta_1) \\
  & = & U(t-t_*,0;\theta_2) U(t_*,0;\theta_1)
\end{eqnarray}
\begin{eqnarray}
  U(t,0;\theta_1) |a1,\theta_1 \rangle & = & e^{-\ihbar^{-1}\chi_{a1}t} |a1,\varphi^t(\theta_1) \rangle \\
  U(t-t_*,0;\theta_2) U(t_*,0;\theta_1) |a1,\theta_1 \rangle & = & e^{-\ihbar^{-1}\chi_{a1}t} |a1,\varphi^t(\theta_1) \rangle \\
  e^{-\ihbar^{-1}\chi_{a1}t_*} U(t-t_*,0;\theta_2)|a1,\theta_2 \rangle & = & e^{-\ihbar^{-1}\chi_{a1}t} |a1,\varphi^t(\theta_1) \rangle  \\
  U(t-t_*,0;\theta_2)|a1,\theta_2 \rangle & = & e^{-\ihbar^{-1}\chi_{a1}(t-t_*)} |a1,\varphi^{t-t_*}(\theta_2) \rangle
\end{eqnarray}
 It follows that $\chi_{a1}$ is a quasienergy associated with $\theta_2$, except if $|a1,\theta_2\rangle = |a1,\varphi^{t-t_*}(\theta_2) \rangle=0$ ($\forall t$). But this last alternative is impossible since $|a1,\varphi^{t-t_*}(\theta_2) \rangle = |a1,\varphi^t(\theta_1)\rangle \not=0$.

\end{proof}

The lemma says that if we restrict the phase space $\Gamma$ to a particular orbit or to any submanifold of this particular orbit, we find the same quasienergy spectrum. Nothing ensures that the quasienergy spectra of the system restricted to two distinct orbits are the same. We can generalize this result:
\begin{theo}[Orbital stability of the quasienergy spectrum]\label{orbstab2}
  Let $\theta_1,\theta_2 \in \Gamma$ two distinct points of $\Gamma$. For $i=1,2$, let $\{\chi_{ai}\}_a$ be the set of the quasienergies such that $\exists (t \mapsto |ai,\varphi^t(\theta_i) \rangle \in \mathcal H \setminus\{0\})$ with $U(t,0;\theta_i)|ai,\theta_i \rangle = e^{-\ihbar^{-1} \chi_{ai} t}|ai,\varphi^t(\theta_i)\rangle$. If $\theta_2 \in \overline{\{\varphi^t(\theta_1)\}_{t \in \mathbb R}}$, then $\forall a$, $\exists b$, such that $\chi_{a1} = \chi_{b2}$.
\end{theo}

\begin{proof}
  Since $\theta_2$ belongs to the topological closure of the orbit of $\theta_1$, $\exists (t_n)_{n\in \mathbb N}$ such that $\lim_{n \to + \infty} \varphi^{t_n}(\theta_1) = \theta_2$. Let $\theta_2^{(n)} \equiv \varphi^{t_n}(\theta_1)$. By application of the previous lemma:
  $$ U(t-t_n,0;\theta_2^{(n)})|a1,\theta_2^{(n)} \rangle  =  e^{-\ihbar^{-1}\chi_{a1}(t-t_n)} |a1,\varphi^{t-t_n}(\theta_2^{(n)}) \rangle $$
  proving that $\chi_{a1}$ is a quasienergy associated with $\theta_2^{(n)}$. This relation being true $\forall t$, it is true for $t=\tau+t_n$:
  $$ U(\tau,0;\theta_2^{(n)})|a1,\theta_2^{(n)} \rangle  =  e^{-\ihbar^{-1}\chi_{a1}\tau} |a1,\varphi^{\tau}(\theta_2^{(n)}) \rangle $$
  $\lim_{n \to +\infty} U(\tau,0;\theta_2^{(n)}) = U(\tau,0;\theta_2)$. $\lim_{n \to+\infty} |a1,\theta_2^{(n)}\rangle = |a1,\theta_2 \rangle$ since $|a1,\theta\rangle$ can be extended to the whole of $\Gamma$ as solution of $(H(\theta)-\ihbar F^\mu \partial_\mu)|a1,\theta \rangle = \chi_{a1} |a1,\theta \rangle$ and $|a1,\theta \rangle$ is  $\mathcal C^1(C)$ for all $C$ integral curve of $F^\mu \partial_\mu$ as $\{\varphi^t(\theta_1)\}_{t \in \mathbb R}$. It follows that $\chi_{a1}$ is a quasienergy associated with $\theta_2$.
\end{proof}

\begin{propo}
If $\chi_a \in \Sp(H_K)$ is a quasienergy associated with the eigenstate $|a,\theta\rangle \in \mathcal K$, and $\lambda \in \Sp(F^\mu \partial_\mu)$ is a Koopman eigenvalue associated with the eigenfunction $f_\lambda \in L^2(\Gamma,d\mu)$ then $\chi_a -\ihbar \lambda \in \Sp(H_K)$ is another quasienergy associated with the eigenstate $f_\lambda(\theta)|a,\theta \rangle$.
\end{propo}
This follows directly from the fact that the Koopman generator is a first order derivative. Reciprocally:
\begin{prop}
  If $\chi_b - \chi_a \in -\ihbar \Sp(F^\mu\partial_\mu)$ (with $\chi_a,\chi_b \in \Sp(H_K)$) then for all $|b,\theta\rangle$ quasienergy states associated with $\chi_b$, it exists $|a,\theta \rangle$ quasienergy state associated with $\chi_a$ such that $|b,\theta\rangle = f_{\chi_b-\chi_a}(\theta)|a,\theta \rangle$ (where $f_{\chi_b-\chi_a}$ is a Koopman mode associated with $\chi_b-\chi_a$).
\end{prop}

\begin{proof}
  Let $\lambda \in \Sp(F^\mu\partial_\mu)$ be such that $\chi_b = \chi_a -\ihbar \lambda$. We suppose that it exists $|b,\theta \rangle$ with $U(t,0;\theta)|b,\theta \rangle = e^{-\ihbar^{-1}\chi_bt}|b,\varphi^t(\theta)\rangle$ such that $|b,\theta \rangle \not= f_\lambda(\theta)|a,\theta \rangle$ for all Koopman modes $f_\lambda$ associated with $\lambda$ and for all quasienergy states $|a,\theta\rangle$ associated with $\chi_a$. But by multiplying the equation defining $|b,\theta\rangle$ by $f_{-\lambda}(\theta)$ we have
  \begin{eqnarray}
    U(t,0;\theta)f_{-\lambda}(\theta)|b,\theta\rangle & = & e^{-\ihbar^{-1} \chi_b t} f_{-\lambda}(\theta)|b,\varphi^t(\theta) \rangle \\
    & = & e^{-\ihbar^{-1} \chi_b t} e^{\lambda t} f_{-\lambda}(\varphi^t(\theta))|b,\varphi^t(\theta)\rangle \\
    & = & e^{-\ihbar^{-1} \chi_a t} f_{-\lambda}(\varphi^t(\theta))|b,\varphi^t(\theta)\rangle
  \end{eqnarray}
  It follows that $|a,\theta \rangle \equiv f_{-\lambda}(\theta)|b,\theta\rangle$ is a quasienergy state associated with $\chi_a$ and then $|b,\theta \rangle = f_\lambda(\theta) |a,\theta \rangle$ in contradiction with the hypothesis.
\end{proof}
It follows from these propositions the following decomposition of the quasienergy spectrum:
\begin{defi}[Fundamental quasienergies]
  We call fundamental quasienergies a minimal set $\{\tilde \chi_i\}_i \subset \Sp(H_K)$ such that $\forall \chi_a \in \Sp(H_K)$, $\exists \lambda \in \Sp(F^\mu \partial_\mu)$, $\exists i$, $\chi_a = \tilde \chi_i -\ihbar \lambda$ and such that there is for all $\theta$ a subset of the associated fundamental quasienergy states $\{|Z\mu_i,  \theta \rangle\}_i$ ($H_K |Z\mu_i, \theta \rangle = \tilde \chi_i |Z\mu_i, \theta \rangle$) which is a basis of $\mathcal H$.
\end{defi}
We can write $\chi_{\lambda,i} = -\ihbar \lambda + \tilde \chi_i$ and
\begin{equation}
|(\lambda,i),\theta \rangle = f_\lambda(\theta) |Z\mu_i, \theta \rangle
\end{equation}
Note that the choice of $\{\tilde \chi_i\}_{i}$ is not necessarily unique. Moreover the number of fundamental quasienergies can be larger than $\dim \mathcal H$ (because some states $|Z\mu_i,\theta \rangle$ can be equal to $0$ for some $\theta$).\\
This notion of fundamental quasienergy states is a generalization of a result of the Floquet theory implying that all Schr\"odinger-Floquet quasienergy states can be decomposed as $|a,\theta \rangle = e^{\imath n \theta} Z(\theta) |\mu_i \rangle$ with $\chi_a = \tilde \chi_i + n \hbar \omega$, where $\theta \mapsto Z(\theta)$ is a $2\pi$-periodic unitary operator and $|\mu_i \rangle$ is a $\theta$-independent eigenvector associated with the eigenvalue $\tilde \chi_i$ of the monodromy matrix of the evolution \cite{Viennot}. The notation used here $|Z\mu_i, \theta \rangle$ is in accordance with the decomposition $Z(\theta) |\mu_i \rangle$ (only possible for periodically driven systems, since it is a consequence of the Floquet theorem). The notion of fundamental quasienergies in the general case is complicated by the possible complicated structure of the Koopman spectrum, in contrast with simplicity of the Floquet spectrum $\mathbb Z \hbar \omega$.\\

For a stroboscopic driven system, we have $\chi_a = \tilde \chi_i - \imath \lambda$ since $U_K f_\lambda(\theta)|a,\theta \rangle = e^{-\imath \chi_a + \lambda} f_\lambda(\theta)|a,\theta \rangle$.

\subsection{Normalisation of the quasienergy states}
% 2.2.1 suite jusqu'a prop 13
We choose the quasienergy states $\{|a,\theta \rangle\}_a$ normalized in $\mathcal K$: $\llangle a|a \rrangle = 1$.
\begin{equation}
\llangle Z\mu_i|Z\mu_i \rrangle = 1 \iff \int_\Gamma \langle Z\mu_i,\theta|Z\mu_i,\theta\rangle d\mu(\theta) = 1
\end{equation}
\begin{equation}
\llangle (\lambda,i)|(\lambda,i)\rrangle=1 \iff \int_\Gamma \langle Z\mu_i,\theta|Z\mu_i,\theta\rangle |f_\lambda(\theta)|^2d\mu(\theta) = 1
\end{equation}
$|Z\mu_i,\theta \rangle$ cannot be normalized in $\mathcal H$ for all $\theta$, since $N(\theta) |Z\mu_i,\theta \rangle$ is not an eigenvector of $-\ihbar F^\mu \partial_\mu + H(\theta)$ in $\mathcal K$. $f_\lambda$ is then not normalized in $L^2(\Gamma,d\mu)$ but in $L^2(\Gamma,d\zeta_i)$ with the ``quantum eigenmeasure'' $d\zeta_i(\theta) = \langle Z\mu_i,\theta|Z\mu_i,\theta\rangle d\mu(\theta)$.
\begin{eqnarray}
\chi_{\lambda,i} & = & \llangle Z\mu_i|\overline{f_\lambda} H_K f_\lambda |Z\mu_i \rrangle \\
& = & \int_{\Gamma} \langle Z\mu_i,\theta|H(\theta)|Z\mu_i,\theta\rangle d\xi_\lambda(\theta) \nonumber \\
& & \quad - \ihbar \int_\Gamma F^\mu(\theta) \langle Z\mu_i,\theta|\partial_\mu|Z\mu_i,\theta\rangle d\xi_\lambda(\theta) \nonumber \\
& & \quad - \ihbar \int_\Gamma F^\mu(\theta) \overline{f_\lambda(\theta)} \frac{\partial f_\lambda(\theta)}{\partial \theta^\mu} d\zeta_i(\theta)
\end{eqnarray}
with the ``classical eigenmeasure'' $d\xi_\lambda(\theta) = |f_\lambda(\theta)|^2 d\mu(\theta)$.

\begin{prop} \label{unimod}
If $f_\lambda$ are all unimodular then the fundamental quasienergy states can be orthonormalized in $\mathcal H$: $\langle Z\mu_i,\theta|Z\mu_j,\theta \rangle = \delta_{ij}$ for $\mu$-almost all $\theta$.
\end{prop}

\begin{proof}
If $f_\lambda$ are all unimodular, $d\xi_\lambda(\theta)= d\mu(\theta)$. It follows that $\int_{\Gamma} \langle Z\mu_i,\theta|H(\theta)|Z\mu_i,\theta\rangle d\xi_\lambda(\theta) - \ihbar \int_\Gamma F^\mu(\theta) \langle Z\mu_i,\theta|\partial_\mu|Z\mu_i,\theta\rangle d\xi_\lambda(\theta)  = \llangle Z\mu_i|H_K|Z\mu_i \rrangle = \tilde \chi_i$. We have then $- \ihbar \int_\Gamma F^\mu(\theta) \overline{f_\lambda(\theta)} \frac{\partial f_\lambda(\theta)}{\partial \theta^\mu} d\zeta_i(\theta) = \chi_{\lambda,i}-\tilde \chi_i = -\ihbar \lambda = - \ihbar \int_\Gamma F^\mu(\theta) \overline{f_\lambda(\theta)} \frac{\partial f_\lambda(\theta)}{\partial \theta^\mu} d\mu(\theta)$. It follows that $d\zeta_i(\theta) = d\mu(\theta) \iff \langle Z\mu_i,\theta|Z\mu_i,\theta \rangle = 1$ for $\mu$-almost all $\theta$.\\
For $j\not=i$, $\llangle Z\mu_i|\overline f_\lambda H_K f_\lambda|Z\mu_j\rrangle=0$. It follows that $\int_\Gamma \langle Z\mu_i,\theta|H(\theta)|Z\mu_j \rangle d\mu(\theta) - \ihbar \int_\Gamma F^\mu(\theta) \langle Z\mu_i,\theta|\partial_\mu|Z\mu_j,\theta\rangle d\mu(\theta) - \ihbar \int_\Gamma F^\mu(\theta) \overline{f_\lambda(\theta)} \frac{\partial f_\lambda(\theta)}{\partial \theta^\mu} \langle Z\mu_i,\theta|Z\mu_j,\theta\rangle d\mu(\theta) =0$. But $\int_\Gamma \langle Z\mu_i,\theta|H(\theta)|Z\mu_j \rangle d\mu(\theta) - \ihbar \int_\Gamma F^\mu(\theta) \langle Z\mu_i,\theta|\partial_\mu|Z\mu_j,\theta\rangle d\mu(\theta) = \llangle Z\mu_i|H_K|Z\mu_j \rrangle = 0$. We have then $\int_\Gamma F^\mu(\theta) \overline{f_\lambda(\theta)} \frac{\partial f_\lambda(\theta)}{\partial \theta^\mu} \langle Z\mu_i,\theta|Z\mu_j,\theta\rangle d\mu(\theta) =0 \Rightarrow \langle Z\mu_i,\theta|Z\mu_j,\theta\rangle=0$ for $\mu$-almost all $\theta$.
\end{proof}

\subsection{Choice of fundamental quasienergy states}
% 2.2.1 prop 19
In this paragraph, we see how to find the fundamental quasienergy states, the other ones being obtained by composition with the Koopman modes. We start by showing that it is easy to exhibit the fundamental quasienergies on the fixed and cyclic points.
\begin{prop} \label{fundquasifixed}
Let $(\Gamma,\mu,\varphi^t,\mathcal H,H)$ be a conservative driven quantum system.
\begin{itemize}
\item Let $\theta_* \in \Gamma$ be a fixed point of $\varphi^t$, we can choose as fundamental quasienergies the set $\{\tilde \chi_i\}_{i=1,...,\dim \mathcal H} = \Sp(H(\theta_*))$. Moreover $|Z\mu_i, \theta_*\rangle$ is the eigenvector of $H(\theta_*)$ associated with $\tilde \chi_i$.
\item Let $\theta_* \in \Gamma$ be a cyclic point of $\varphi^t$, we can choose as fundamental quasienergies the set $\{\tilde \chi_i\}_{i=1,...,\dim \mathcal H} = \Sp(M_{\theta_*})$ where $M_{\theta_*}$ is the monodromy matrix of the Schr\"odinger equation. Moreover $|Z\mu_i, \theta_*\rangle$ is the eigenvector of $M_{\theta_*}$ associated with $\tilde \chi_i$.
\end{itemize}
Let $(\Gamma, \mu,\varphi,\mathcal H,U)$ be a conservative stroboscopic quantum system.
\begin{itemize}
\item Let $\theta_* \in \Gamma$ be a fixed point of $\varphi$, we can choose as fundamental quasienergies the set $\{e^{-\imath \tilde \chi_i}\}_{i=1,...,\dim \mathcal H} = \Sp(U(\theta_*))$. Moreover $|Z\mu_i,\theta_* \rangle$ is the eigenvector of $U(\theta_*)$ associated with $e^{-\imath \tilde \chi_i}$.
\item Let $\theta_* \in \Gamma$ be a $p$-cyclic point of $\varphi$, we can choose as the quasienergies the set $\{e^{-\imath p \tilde \chi_i} \}_{i=1,...,\dim \mathcal H} = \Sp(U(\varphi^{p-1}(\theta_*))...U(\theta_*))$. Moreover $|Z\mu_i,\theta_*\rangle$ is the eigenvector of $U(\varphi^{p-1}(\theta_*))...U(\theta_*)$ associated with $e^{-\imath p \tilde \chi_i}$.
\end{itemize}
\end{prop}

\begin{proof}
By using the decomposition of the quasienergies in equation \ref{eigenU} we have
\begin{eqnarray}
U(t,0;\theta) f_\lambda(\theta)|Z\mu_i,\theta \rangle & = & e^{-\ihbar^{-1} \chi_{\lambda,i}t} f_\lambda(\varphi^t(\theta)) |Z\mu_i,\varphi^t(\theta) \rangle \\
& = & e^{-\lambda t - \ihbar^{-1} \tilde \chi_i t} e^{\lambda t} f_\lambda(\theta) |Z\mu_i,\varphi^t(\theta) \rangle \\
& = & e^{-\ihbar^{-1} \tilde \chi_i t}f_\lambda(\theta) |Z\mu_i,\varphi^t(\theta) \rangle
\end{eqnarray}
Let $\theta_*$ be a fixed point: $\varphi^t(\theta_*) = \theta_*$. Since it exists at least one Koopman mode such that $f_\lambda(\theta_*) \not=0$ (with $\lambda=0$ for example), it follows
\begin{equation}
U(t,0;\theta_*) |Z\mu_i,\theta_* \rangle = e^{-\ihbar^{-1} \tilde \chi_i t} |Z\mu_i,\theta_* \rangle
\end{equation}
But $U(t,0;\theta) = \Te^{-\ihbar \int_0^t H(\varphi^t(\theta)) dt} \Rightarrow U(t,0;\theta_*) = e^{-\ihbar^{-1} H(\theta_*) t}$ ($\Te$ denoting the time-ordered exponential, i.e. the Dyson series), and then
\begin{eqnarray}
& & e^{-\ihbar^{-1} H(\theta_*) t} |Z\mu_i,\theta_* \rangle = e^{-\ihbar^{-1} \tilde \chi_i t} |Z\mu_i,\theta_* \rangle \\
& & \iff H(\theta_*) |Z\mu_i,\theta_* \rangle = \tilde \chi_i |Z\mu_i,\theta_* \rangle
\end{eqnarray}
Let $\theta_*$ be a cyclic point: $\varphi^T(\theta_*) = \theta_*$. Because it exists at least one Koopman mode such that $f_\lambda(\theta_*) \not=0$, it follows
\begin{equation}
U(T,0;\theta_*) |Z\mu_i,\theta_* \rangle = e^{-\ihbar^{-1} \tilde \chi_i T} |Z\mu_i,\theta_* \rangle
\end{equation}
By the Floquet theorem we have $U(t,0;\theta_*) = Z(t,0;\theta_*) e^{-\ihbar^{-1} M_{\theta_*} t}$ with $Z(t+T,0,\theta_*)=Z(t,0;\theta_*)$ ($Z(T,0,\theta_*) = 1_{\mathcal H}$). It follows
\begin{eqnarray}
& & e^{-\ihbar^{-1} M_{\theta_*} T} |Z\mu_i,\theta_* \rangle = e^{-\ihbar^{-1} \tilde \chi_i T} |Z\mu_i,\theta_* \rangle \\
& & \iff M_{\theta_*} |Z\mu_i,\theta_* \rangle = \tilde \chi_i |Z\mu_i,\theta_* \rangle
\end{eqnarray}
Let $\theta_*$ be a fixed point of a discrete time dynamical system: $\varphi(\theta_*)=\theta_*$.
\begin{equation}
  U(\theta_*)|Z\mu_i,\theta_*\rangle = e^{-\imath \tilde \chi_i} |Z\mu_i,\theta_* \rangle
\end{equation}
Let $\theta_*$ be a $p$-cyclic point: $\varphi^p(\theta_*) = \theta_*$. We have
\begin{eqnarray}
  U(\theta_*)|Z\mu_i,\theta_*\rangle & = & e^{-\imath \tilde \chi_i} |Z\mu_i,\varphi(\theta_*)\rangle \\
  U(\varphi(\theta_*))U(\theta_*)|Z\mu_i,\theta_*\rangle & = & e^{-\imath 2 \tilde \chi_i} |Z\mu_i,\varphi^2(\theta_*)\rangle \\
  & \vdots & \\
  U(\varphi^{p-1}(\theta_*))...U(\theta_*)|Z\mu_i,\theta_*\rangle & = & e^{-\imath p \tilde \chi_i} |Z\mu_i,\varphi^p(\theta_*)\rangle \\
  U(\varphi^{p-1}(\theta_*))...U(\theta_*)|Z\mu_i,\theta_*\rangle & = & e^{-\imath p \tilde \chi_i} |Z\mu_i,\theta_*\rangle
\end{eqnarray}
\end{proof}
For the fixed point and for periodic orbit, we know then the fundamental quasienergy spectrum (which is the same on the whole of a periodic orbit because of the orbital stability lemma \ref{orbstab1}). But between two fixed points or between two cycles, does a relation exist for the quasienergies?
\begin{prop}
Let $(\Gamma,\mu,\varphi^t,\mathcal H,H)$ be a conservative driven quantum system, and $\theta_{*1}$ and $\theta_{*2}$ be two fixed points of $\varphi^t$. Let $\{\tilde \chi_{i1}\}_i$ be the eigenvalues of $H(\theta_{*1})$ and $\{\tilde \chi_{i2}\}_i$ be the eigenvalues of $H(\theta_{*2})$. Let $\{|Z\mu_{i1},\theta\rangle\}_i$ be the fundamental quasienergy states associated with the fixed point $\theta_{*1}$ and $\{|Z\mu_{i2},\theta\rangle\}_i$ be the fundamental quasienergy states associated with the fixed point $\theta_{*2}$. If $\tilde \chi_{i2} \not\in \Sp(H(\theta_{1*}))$ then $|Z\mu_{i2},\theta_{1*} \rangle= 0$.
\end{prop}

\begin{proof}
$(-\ihbar F^\mu(\theta) \partial_\mu + H(\theta))|Z\mu_{i2},\theta \rangle = \tilde \chi_{i2} |Z\mu_{i2},\theta \rangle \Rightarrow H(\theta_{1*}) |Z\mu_{i2},\theta_{1*} \rangle =\tilde \chi_{i2} |Z\mu_{i2},\theta_{1*} \rangle$ because $F^\mu(\theta_{1*}) = 0$ since $\theta_{1*}$ is a fixed point. If follows that $|Z\mu_{i2},\theta_{1*} \rangle$ is an eigenvector of $H(\theta_{1*})$ with eigenvalue $\tilde \chi_{i2}$ except if $|Z\mu_{i2},\theta_{1*} \rangle=0$.
\end{proof}

We have the same thing for a stroboscopic driven quantum system with $e^{-\imath \tilde \chi_{i_2}}$ for two fixed points $\theta_{*1}$ and $\theta_{*2}$.

\begin{prop} \label{cycltocycl}
  Let $(\Gamma,\mu,\varphi^t,\mathcal H,H)$ be a conservative driven quantum system, and $\theta_{*1}$ and $\theta_{*2}$ be two cyclic points of $\varphi^t$. Let $\{\tilde \chi_{i1}\}_i$ be the eigenvalues of $M_{\theta_{1*}}$ and $\{\tilde \chi_{i2}\}_i$ be the eigenvalues of $M_{\theta_{2*}}$ ($M_{\theta_{i*}}$ are the monodromy matrices). Let $\{|Z\mu_{i1},\theta\rangle\}_i$ be the fundamental quasienergy states associated with the cyclic point $\theta_{*1}$ and $\{|Z\mu_{i2},\theta\rangle\}_i$ be the fundamental quasienergy states associated with the cyclic point $\theta_{*2}$. If $\theta_{*1}$ and $\theta_{*2}$ belong to two different cycles of periods $T_1$ and $T_2$ such that $\frac{T_2}{T_1} = \frac{p}{q}\in \mathbb Q$ ($\frac{p}{q}$ being an irreductible fraction), then if $\tilde \chi_{i2}\not\in \Sp(M_{\theta_{1*}}) + \frac{2\pi \hbar}{pT_1} \mathbb Z$ then $|Z\mu_{i2},\theta_{1*} \rangle= 0$. If $\theta_{*1}$ and $\theta_{*2}$ belong to the same cycle of period $T$, then $\tilde \chi_{i2} \in \Sp(M_{\theta_{1*}}) + \frac{2\pi \hbar}{T-t_2} \mathbb Z$ (where $\varphi^{t_2}(\theta_{1*})=\theta_{2*}$) and $\exists j \in\{1,...,\dim \mathcal H\}$ such that $|Z\mu_{j1},\theta\rangle = |Z\mu_{i2},\theta \rangle$.\\

  Let $(\Gamma,\mu,\varphi,\mathcal H,H)$ be a conservative strobocopic driven quantum system, and $\theta_{*1}$ and $\theta_{*2}$ be two respectively $n_1$-cyclic and $n_2$-cyclic points of $\varphi$. $\frac{n_2}{n_1} = \frac{p}{q}$ ($\frac{p}{q}$ being an irreductible fraction). If $e^{-\imath qn_2 \tilde \chi_2} \not\in \Sp(U(\varphi^{pn_1-1}(\theta_{*1}))...U(\theta_{*1}))$ then $|Z\mu_{i2},\theta_{*1} \rangle = 0$. If $\theta_{*1}$ and $\theta_{*2}$ belong to the same cycle of period $n$, then $e^{-\imath n \tilde \chi_2} \in \Sp(U(\varphi^{n-1}(\theta_{*1}))...U(\theta_{*1}))$ and $\exists j \in \{1,...,\dim \mathcal H\}$ such that $|Z\mu_{j1},\theta \rangle = |Z\mu_{i2},\theta \rangle$.
\end{prop}

\begin{proof}
Let $\theta_{*1}$ and $\theta_{*2}$ belonging to two different cycles with $\frac{T_2}{T_1} = \frac{p}{q} \in \mathbb Q$. $U(qT_2,0;\theta)|Z\mu_{i2},\theta \rangle = e^{-\ihbar^{-1} \tilde \chi_{2i} qT_2} |Z\mu_{i2},\varphi^{qT_2}(\theta) \rangle$\\ $\Rightarrow Z(qT_2,0;\theta_{*1}) e^{-\ihbar^{-1} M_{\theta_{*1}} qT_2} |Z\mu_{i2},\theta_{*1} \rangle = e^{-\ihbar^{-1} \tilde \chi_{2i}qT_2} |Z\mu_{i2},\varphi^{qT_2}(\theta_{*1}) \rangle$. Since $qT_2=pT_1$ we have $Z(qT_2,0;\theta_{*1}) = Z(pT_1,0;\theta_{*1})=1_{\mathcal H}$ and $\varphi^{qT_2}(\theta_{*1}) = \varphi^{pT_1}(\theta_{*1}) = \theta_{*1}$. We have then $e^{-\ihbar^{-1} M_{\theta_{*1}} pT_1} |Z\mu_{i2},\theta_{*1} \rangle = e^{-\ihbar^{-1} \tilde \chi_{2i}pT_1} |Z\mu_{i2},\theta_{*1} \rangle$. $|Z\mu_{i2},\theta_{*1} \rangle $ is then an eigenvector of $M_{\theta_{*1}}$ with eigenvalue $\tilde \chi_{2i} \mod \frac{2\pi \hbar}{T_1}$ except if $|Z\mu_{i2},\theta_{*1} \rangle = 0$.\\
If $\theta_{*1}$ and $\theta_{*2}$ belong to the same cycle, the same arguments occur with $p=q=1$. Moreover $\exists t_2<T$ such that $\varphi^{t_2}(\theta_{*1}) = \theta_{*2}$. It follows
\begin{eqnarray}
U(T,0;\theta_{*1}) & = & U(T,t_2;\theta_{*1}) U(t_2,0;\theta_{*1}) \\
e^{-\ihbar^{-1}M_{\theta_{*1}} T} & = & Z(T,t_2;\theta_{*1}) e^{-\ihbar^{-1} M_{\theta_{*2}} (T-t_2)} \nonumber \\
& & \quad \times Z(t_2,0;\theta_{*1}) e^{-\ihbar^{-1} M_{\theta_{*1}} t_2} \\
e^{-\ihbar^{-1}M_{\theta_{*1}} (T-t_2)} & = & Z(T,t_2;\theta_{*1}) e^{-\ihbar^{-1} M_{\theta_{*2}} (T-t_2)} Z(t_2,0;\theta_{*1})
\end{eqnarray}
But $Z(T,t_2;\theta_{*1}) Z(t_2,0;\theta_{*1}) = Z(T,0;\theta_{*1}) = 1_{\mathcal H} \Rightarrow Z(T,t_2;\theta_{*1}) = Z(t_2,0;\theta_{*1})^{-1}$. It follows that $e^{-\ihbar^{-1}M_{\theta_{*1}} (T-t_2)}$ and $e^{-\ihbar^{-1} M_{\theta_{*2}} (T-t_2)}$ are then similar and $\Sp(M_{\theta_{*1}}) = \Sp(M_{\theta_{*2}}) \mod \frac{2\pi \hbar}{T-t_2}$.\\

For a stroboscopic system we have
\begin{equation}
  U(\varphi^{qn_2-1}(\theta))...U(\theta)|Z\mu_{i2},\theta \rangle = e^{-\imath qn_2 \tilde \chi_{i2}} |Z\mu_{i2},\varphi^{qn_2}(\theta)\rangle
\end{equation}
and then
\begin{eqnarray}
  U(\varphi^{qn_2-1}(\theta_{*1}))...U(\theta_{*1})|Z\mu_{i2},\theta_{*1} \rangle & = & e^{-\imath qn_2 \tilde \chi_{i2}} |Z\mu_{i2},\varphi^{qn_2}(\theta_{*1})\rangle \\
  & = & e^{-\imath qn_2 \tilde \chi_{i2}} |Z\mu_{i2},\varphi^{pn_1}(\theta_{*1})\rangle \\
  & = & e^{-\imath qn_2 \tilde \chi_{i2}} |Z\mu_{i2},\theta_{*1}\rangle
\end{eqnarray}
It follows that $e^{-\imath qn_2 \tilde \chi_{i2}} \in \Sp(U(\varphi^{qn_2-1}(\theta_{*1}))...U(\theta_{*1}))$ or $|Z\mu_{i2},\theta_{*1}\rangle=0$.\\
If $\theta_{*1}$ and $\theta_{*2}$ belong to the same cycle, we have the same thing with $n_1=n_2=1$. Moreover $\exists p<n$ such that $\varphi^p(\theta_{*1}) = \theta_{*2}$. We have then
\begin{eqnarray}
  U(\varphi^{n-1}(\theta_{*1}))...U(\theta_{*1}) & = & U(\varphi^{n-p-1}(\theta_{*2}))...U(\theta_{*2}) \nonumber \\
  & & \quad \times U(\varphi^{p-1}(\theta_{*1}))...U(\theta_{*1}) \\
  & = & U(\varphi^{n-p}(\theta_{*2}))^\dagger...U(\varphi^{n-1}(\theta_{*2}))^\dagger \nonumber \\
  & & \quad \times U(\varphi^{n-1}(\theta_{*2})...U(\theta_{*2}) \nonumber \\
  & & \qquad \times U(\varphi^{p-1}(\theta_{*1}))...U(\theta_{*1})
\end{eqnarray}
But $\varphi^{n-p}(\theta_{*2}) = \varphi^{n-p}(\varphi^p(\theta_{*1})) = \varphi^n(\theta_{*1}) = \theta_{*1}$, and then $U(\varphi^{n-p}(\theta_{*2}))^\dagger...U(\varphi^{n-1}(\theta_{*2})^\dagger = \left(U(\varphi^{p-1}(\theta_{*1}))...U(\theta_{*1})\right)^\dagger$. It follows that $U(\varphi^{n-1}(\theta_{*1}))...U(\theta_{*1})$ and $U(\varphi^{n-1}(\theta_{*2})...U(\theta_{*2})$ are similar and have then the same spectrum.
\end{proof}
With these two properties, we see that the condition for which a fundamental quasienergy state has a non-zero continuation from a cyclic orbit to another one, is that some spectral properties of $U$ are the same in the two orbits. Such a case is not generic since the cyclic orbits depend from the structure of the classical flow whereas the structure of $U$ depends on the quantum system. It follows that in general, each fixed point and each cyclic orbit involve specific fundamental quasienergies.\\

We know now how to find the fundamental quasienergies on the fixed points and on the cyclic orbits. Now, we want consider a ``non-regular'' orbit. Let $\Gamma_e$ be an ergodic component of $\Gamma$ for the flow $\varphi^t$ with respect to the measure $\mu$. For $\mu$-almost all $\theta \in \Gamma_e$, $\overline{\{\varphi^t(\theta)\}_{t \in \mathbb R^+}} = \Gamma_e$. Possibly, $\Gamma_e = \Gamma$ if the flow is ergodic on the whole of $\Gamma$. Since by definition $U(t,0;\theta) |a,\theta \rangle  = e^{-\ihbar^{-1} \chi_a t} |a,\varphi^t(\theta) \rangle$, we have by applying the ergodic theorem for $\mu$-almost all $\theta \in \Gamma_e$
\begin{eqnarray}
\lim_{T \to +\infty} \frac{1}{T} \int_0^T e^{\ihbar^{-1} \chi_a t} U(t,0;\theta) dt |a,\theta \rangle & = & \lim_{T \to +\infty} \frac{1}{T} \int_0^T |a,\varphi^t(\theta)\rangle dt \\
  & = & \int_{\Gamma_e} |a,\theta \rangle d\mu(\theta) \\
  & = & |\bar a \rangle
\end{eqnarray}
Let $V_a(\theta) = \lim_{T \to +\infty} \frac{1}{T} \int_0^T e^{\ihbar^{-1} \chi_a t} U(t,0;\theta) dt$. This result can be used to compute the fundamental quasienergy states for all $\theta \in \Gamma_e$. Indeed
\begin{equation}
  |Z\mu_i,\theta \rangle = V_i(\theta)^{-1} |\overline{Z\mu_i} \rangle
\end{equation}
It needs then to find $\tilde \chi_i$ and $|\overline{Z\mu_i} \rangle$ to solve the problem. Let $\theta_* \in \Gamma_e$ be a fixed point embedded into the ergodic component (the discussion can be easily adaptated to a cyclic point). The existence of such a point is not incompatible with the ergodic hypothesis, this one states only that the set of all fixed and cyclic points embedded into $\Gamma_e$ has a zero measure by $\mu$. Since $\theta_* \in \overline{\{\varphi^t(\theta)\}_{t \in \mathbb R^+}}$ (for $\mu$-almost all $\theta \in \Gamma_e$), by the orbital stability theorem \ref{orbstab2}, we can choose for the fundamental quasienergy spectrum associated with $\Gamma_e$, the fundamental quasienergy spectrum associated with $\theta_*$ by the property \ref{fundquasifixed}. We know then $\{\tilde \chi_i\}_i$ and  we can compute $V_i(\theta)$. Now it needs to find $|\overline{Z\mu_i} \rangle$. $V_i(\theta_*) |Z\mu_i,\theta_* \rangle \not= |\overline{Z\mu_i} \rangle$ since the orbit from $\theta_*$ is trivially not dense into $\Gamma_e$ (the ``$\mu$-almost'' in the ergodic properties excludes precisely such points). Let $\epsilon$ be a deviation into $\Gamma_e$ in the neighbourhood of zero, we can write that
\begin{equation}
  |\overline{Z\mu_i} \rangle = V_i(\theta_*+\epsilon) |Z\mu_i,\theta_*+\epsilon\rangle
\end{equation}
It follows that
\begin{equation}
  |Z\mu_i,\theta \rangle = V_i(\theta)^{-1} V_i(\theta_*+\epsilon) |Z\mu_i,\theta_*+\epsilon \rangle
\end{equation}
The last operation consists to find $|Z\mu_i,\theta_*+\epsilon \rangle$ which can be computed from $\{|Z\mu_i,\theta_*\rangle\}_i$ by using a local expansion around $\theta_*$ (see \ref{appendixB1}).\\

For a stroboscopic system, we have
\begin{equation}
  V_a(\theta) = \lim_{N \to +\infty} \frac{1}{N} \sum_{n=0}^{N-1} e^{\imath n \chi_a} U(\varphi^{n-1}(\theta))...U(\theta)
\end{equation}

\section{Schr\"odinger-Koopman dynamics}
Now we want to study the quantum dynamics in the enlarged Hilbert space $\mathcal K$ when we start from a quasienergy state. In a first time, we consider the case where the initial condition for the classical flow is a single point of $\Gamma$, and in a second time we consider the case where it is a probability distribution on $\Gamma$.
\subsection{Ergodic geometric phase}
% 2.3 debut
In this section, we denote by $T\Gamma$ the tangent bundle of $\Gamma$, $T_\theta \Gamma$ the tangent vector space of $\Gamma$ at $\theta$, by $\Omega^1 \Gamma$ the space of differential 1-forms on $\Gamma$, and $i_X : \Omega^1 \Gamma \to \mathbb R$ (with $X \in T\Gamma$) the inner product of the manifold $\Gamma$.
\begin{theo}
Let $(\Gamma,\mu,\varphi^t,\mathcal H,H)$ be a conservative driven quantum system and $|a,\theta\rangle \in \mathcal K$ be a normalized quasienergy state. Let $A_a = \langle a,\theta|d|a,\theta \rangle \in \Omega^1 \Gamma$ be the Berry potential associated with the quasienergy state and $X(t) = F^\mu(\varphi^t(\theta_0)) \partial_\mu \in T_{\varphi^t(\theta_0)} \Gamma$ be the tangent vector field of the flow.
If $\varphi^t$ is ergodic then
\begin{equation}
\tilde \psi(t) \sim e^{-\ihbar^{-1} \int_0^t \langle a,\varphi^t(\theta_0)|H(\varphi^t(\theta_0))|a,\varphi^t(\theta_0)\rangle dt - \int_0^t i_{X(t)}A_a dt} |a,\varphi^t(\theta_0)\rangle
\end{equation}
with $t$ in the neighbourhood of $+\infty$ and $\tilde \psi$ solution of the Schr\"odinger equation for $\mu$-almost all initial conditions $(|a,\theta_0\rangle,\theta_0) \in \mathcal H \times \Gamma$.\\
$e^{-\ihbar^{-1} \int_0^t \langle a,\varphi^t(\theta_0)|H(\varphi^t(\theta_0))|a,\varphi^t(\theta_0)\rangle dt}$ is the dynamical phase and $e^{- \int_0^t i_{X(t)}A_a dt}$ is the (non-adiabatic) geometric phase of the driven dynamics.
\end{theo}

\begin{proof}
Let $\psi \in \mathcal K$ be the solution of the Schr\"odinger-Koopman equation $\ihbar \frac{\partial \psi}{\partial t} = H_K \psi$ with $\psi(\theta,t=0) = |a,\theta \rangle$. We have
\begin{eqnarray}
\psi(\theta,t) & = & U_K(t,0) |a,\theta \rangle \\
& = & e^{-\ihbar^{-1} H_K t} |a,\theta \rangle \\
& = & e^{-\ihbar^{-1} \chi_a t} |a,\theta \rangle
\end{eqnarray}
By applying theorem \ref{SKE} we have
\begin{equation}
\tilde \psi(t) = e^{-\ihbar^{-1} \chi_a t} |a,\varphi^t(\theta_0) \rangle
\end{equation}
But since $\chi_a = \llangle a|H_K|a\rrangle$ we have (with $X = F^\mu(\theta) \partial_\mu \in T\Gamma$)
\begin{eqnarray}
\chi_a & = & \int_\Gamma \langle a,\theta|H(\theta)|a,\theta\rangle d\mu(\theta) - \ihbar \int_\Gamma F^\mu(\theta) \langle a,\theta|\partial_\mu|a,\theta \rangle d\mu(\theta) \\
& = & \int_\Gamma \langle a,\theta|H(\theta)|a,\theta\rangle d\mu(\theta) - \ihbar \int_\Gamma i_X A_a d\mu(\theta)
\end{eqnarray}
By applying the Birkhoff ergodic theorem \cite{Lasota} we have for $\mu$-almost all $\theta_0$
\begin{eqnarray}
& & \lim_{t \to +\infty} \frac{1}{t} \int_0^t \langle a,\varphi^t(\theta_0)|H(\varphi^t(\theta_0))|a,\varphi^t(\theta_0)\rangle dt \nonumber \\
& & \qquad  = \int_\Gamma \langle a,\theta|H(\theta)|a,\theta\rangle d\mu(\theta) 
\end{eqnarray}
\begin{equation}
\lim_{t \to +\infty} \frac{1}{t} \int_0^t i_{X(t)} A_a dt = \int_\Gamma i_X A_a d\mu(\theta)
\end{equation}
It follows
\begin{equation}
\chi_a \sim \frac{1}{t} \int_0^t \langle a,\varphi^t(\theta_0)|H(\varphi^t(\theta_0))|a,\varphi^t(\theta_0)\rangle dt -\frac{\ihbar}{t} \int_0^t i_{X(t)} A_a dt
\end{equation}
\end{proof}

For a stroboscopic driven quantum system, we have
\begin{equation}
  U(\theta)|a,\theta\rangle = e^{-\imath \chi_a} |a,\varphi(\theta)\rangle \Rightarrow \langle a,\theta|U(\theta)|a,\theta\rangle = e^{-\imath\chi_a} \langle a,\theta|a,\varphi(\theta) \rangle
\end{equation}
and then
\begin{eqnarray}
  \chi_a & = & \imath \ln \langle a,\theta|U(\theta)|a,\theta\rangle - \imath \ln \langle a,\theta|a,\varphi(\theta) \rangle \\
  & = & \imath \int_\Gamma \ln \langle a,\theta|U(\theta)|a,\theta\rangle d\mu(\theta) - \imath \int_\Gamma \ln \langle a,\theta|a,\varphi(\theta) \rangle d\mu(\theta)
\end{eqnarray}
The last equation following from $\int_\Gamma \chi_a d\mu(\theta) = \chi_a$ since $\chi_a$ is independent from $\theta$. But we have also
\begin{equation}
  \langle a,\varphi^n(\theta)|U(\varphi^n(\theta))|a,\varphi^n(\theta)\rangle = e^{-\imath \chi_a} \langle a,\varphi^n(\theta)|a,\varphi^{n+1}(\theta) \rangle
\end{equation}
It follows that
\begin{eqnarray}
  & & \langle a,\theta|U(\theta)|a,\theta \rangle \langle a,\varphi(\theta)|U(\varphi(\theta))|a,\varphi(\theta)\rangle...\langle a,\varphi^n(\theta)|U(\varphi^n(\theta))|a,\varphi^n(\theta)\rangle \nonumber \\
  & & = e^{-\imath \chi_a n} \langle a,\theta|a,\varphi(\theta)\rangle \langle a,\varphi(\theta)|a,\varphi^2(\theta)\rangle...\langle a,\varphi^n(\theta)|a,\varphi^{n+1}(\theta) \rangle
\end{eqnarray}
$\langle a,\theta|a,\varphi(\theta)\rangle...\langle a,\varphi^n(\theta)|a,\varphi^{n+1}(\theta) \rangle$ is very similar to the Bargmann invariant $\langle a,\theta|a,\varphi(\theta)\rangle... \langle a,\varphi^{n-1}(\theta)|a,\varphi^{n}(\theta) \rangle \langle a,\varphi^{n}(\theta)|a,\theta \rangle$ for a cyclic dynamics such that $\varphi^{n+1}(\theta) = \theta$ \cite{Rabei}. We can then consider $\langle a,\theta|a,\varphi(\theta)\rangle... \langle a,\varphi^n(\theta)|a,\varphi^{n+1}(\theta) \rangle$ for $n \to + \infty$ as a kind of Bargmann invariant for a non-cyclic dynamics. And since the Bargmann invariant is closely related to the geometric phases \cite{Rabei}, we can consider $\langle a,\theta|a,\varphi(\theta)\rangle ...\langle a,\varphi^n(\theta)|a,\varphi^{n+1}(\theta) \rangle$ for $n \to + \infty$ as the stroboscopic geometric phase. We have then
\begin{eqnarray}
  \chi_a & = & \lim_{n \to +\infty} \frac{\imath}{n} \sum_{k=0}^{n} \ln \langle a,\varphi^k(\theta)|U(\varphi^k(\theta))|a,\varphi^k(\theta) \rangle \nonumber \\
  & & \quad - \lim_{n \to+\infty} \frac{\imath}{n} \sum_{k=0}^n \ln \langle a,\varphi^k(\theta)|a,\varphi^{k+1}(\theta) \rangle
\end{eqnarray}
By using the ergodic theorem we recover:\\ $\lim_{n\to +\infty} \frac{1}{n} \sum_{k=0}^{n} \ln \langle a,\varphi^k(\theta)|U(\varphi^k(\theta))|a,\varphi^k(\theta) \rangle  = \int_\Gamma \ln \langle a,\theta|U(\theta)|a,\theta \rangle d\mu(\theta)$ and $\lim_{n\to+\infty} \frac{1}{n} \sum_{k=0}^n \ln \langle a,\varphi^k(\theta)|a,\varphi^{k+1}(\theta) \rangle = \int_\Gamma \ln \langle a,\theta|a,\varphi(\theta)\rangle d\mu(\theta)$ which is the argument of the ergodic geometric phase of the stroboscopic quantum system.\\
Remark: for an usual cyclic geometric phase, we have: \\ $e^{\oint_{\mathcal C} A_a} = \lim_{n \to + \infty} \langle a,\theta|a,\varphi^{\Delta t}(\theta) \rangle...\langle a,\varphi^{(n-1) \Delta t}|a,\varphi^{n \Delta t} \theta \rangle \langle a,\varphi^{n \Delta t}|a,\theta \rangle$ for $\varphi^t$ a continuous time dynamical system, $T$-cyclic by starting from $\theta$ ($\mathcal C$ is the cycle into $\Gamma$), with $\Delta t = \frac{T}{n+1}$ (we consider a partition $0<\Delta t<...<(n+1)\Delta t=T$ of $[0,T]$) and with $A_a = \langle a,\theta|d|a,\theta\rangle$. This formula results from the fact that $\langle a,\varphi^{p\Delta t}(\theta) |a,\varphi^{(p+1)\Delta t}(\theta) \rangle = 1 + \left. \langle a,\varphi^t(\theta)|\partial_t|a,\varphi^t(\theta)\rangle\right|_{t=p \Delta t} \Delta t + \mathcal O(\Delta t^2) = e^{\left. \langle a,\varphi^t(\theta)|\partial_t|a,\varphi^t(\theta)\rangle\right|_{t=p \Delta t} \Delta t} + \mathcal O(\Delta t^2)$ and by the definition of the integral as a Riemann sum. We see then that the discrete time ergodic geometric phase has a similar expression of the cyclic geometric phase (viewed as an infinite number of steps in the Bargmann invariant in the finite time range $T$).

\subsection{Density matrix}
% 2.3 fin
\begin{theo}
  Let $(\Gamma,\mu,\varphi^t,\mathcal H,H)$ be a mixing conservative driven quantum system with $\mu$ a preserved measure. Let $\rho_i = \tr_{L^2(\Gamma,d\mu)}|Z\mu_i \rrangle \llangle Z\mu_i|$ be the density matrices associated with the fundamental quasienergy states (by construction, $\rho_i$ are stationnary and then are kinds of steady states). Let $\rho(t) = \tr_{L^2(\Gamma,d\mu)}|\psi(t) \rrangle \llangle \psi(t)|$ be the density matrix associated with the solution of the Schr\"odinger-Koopman equation for a initial condition $\psi_0(\theta) = \tilde \psi_0 \sqrt{\rho_0(\theta)}$ (with $\rho_0 \in L^1_+(\Gamma,d\mu)$, $\rho_0\geq 0$ and $\int_\Gamma \rho_0(\theta) d\mu(\theta)=1$; and $\tilde \psi_0 \in \mathcal H$, $\|\tilde \psi_0\|=1$).\\
  If the limit $\lim_{t\to +\infty} \rho(t) = \rho_\infty$ exists then $\rho_\infty = \sum_i p_i \rho_i$ ($\sum_i p_i = 1$, $p_i \in [0,1]$) is a combination of steady states.
\end{theo}

\begin{proof}
  By using the fact that by definition the set of fundamental quasienergy states generates $\mathcal H$, and that the set of the Koopman modes is an eigenbasis of $L^2(\Gamma,d\mu)$, we have $\psi_0(\theta) = \sum_i \sum_\lambda c_{i\lambda} f_\lambda(\theta) |Z\mu_i,\theta \rangle$ and then $\psi(t,\theta) = \sum_i \sum_\lambda c_{i\lambda} e^{-\ihbar^{-1} \chi_{i\lambda} t}f_\lambda(\theta) |Z\mu_i,\theta \rangle$. For the sake of simplicity, we write the sum on $\lambda$ as a discrete sum without degeneracy index, the formulae can be easily adapted with degeneracies and with a continuous Koopman spectrum. We have then
  \begin{eqnarray}
    \rho(t) & = & \sum_{ij}\sum_{\lambda \nu} c_{i\lambda}\overline{c_{j\nu}} e^{\ihbar^{-1}(\tilde \chi_j - \tilde \chi_i)t} \nonumber \\
    & & \quad \times \int_\Gamma f_{\lambda-\nu}(\theta)|Z\mu_i,\varphi^t(\theta)\rangle\langle Z\mu_j,\varphi^t(\theta)|d\mu(\theta)
  \end{eqnarray}
  Since the $\varphi^t$ is mixing, we have by the property \ref{mixergo}:
 \begin{eqnarray}
    & & \lim_{t \to +\infty} \int_\Gamma f_{\lambda-\nu}(\theta)|Z\mu_i,\varphi^t(\theta)\rangle\langle Z\mu_j,\varphi^t(\theta)|d\mu(\theta) \nonumber\\
    & & \quad = \int_\Gamma f_{\lambda-\nu}(\theta)d\mu(\theta) \int_\Gamma |Z\mu_i,\theta \rangle \langle Z\mu_j,\theta|d\mu(\theta)
  \end{eqnarray}
  But $\int_\Gamma f_{\lambda-\nu}(\theta)d\mu(\theta) = \delta_{\lambda \nu}$ because of property \ref{nullFmix}. It follows that $\forall \epsilon >0$, $\exists t_\epsilon>0$, such that $\forall t\geq t_\epsilon$ we have
  \begin{equation}
    \left|\rho(t) - \sum_{ij}\sum_\lambda c_{i\lambda}\overline c_{j\lambda} e^{\ihbar^{-1}(\tilde \chi_j-\tilde \chi_i)t} \int_\Gamma|Z\mu_i,\theta\rangle \langle Z\mu_j,\theta|d\mu(\theta) \right| < \epsilon
  \end{equation}
  It follows that
  \begin{eqnarray}
    & & \sum_{ij} D_{ij} \frac{1}{T} \int_{t_\epsilon}^{t_\epsilon+T} e^{\ihbar^{-1}(\tilde \chi_j-\tilde \chi_i)t} dt - \epsilon \nonumber \\
    & & \qquad < \frac{1}{T} \int_{t_\epsilon}^{t_\epsilon+T} \rho(t)dt < \nonumber\\
    & & \qquad \qquad \sum_{ij} D_{ij} \frac{1}{T} \int_{t_\epsilon}^{t_\epsilon+T} e^{\ihbar^{-1}(\tilde \chi_j-\tilde \chi_i)t} dt + \epsilon
  \end{eqnarray}
with $D_{ij} = \sum_\lambda c_{i\lambda}\overline c_{j\lambda} \int_\Gamma|Z\mu_i,\theta\rangle \langle Z\mu_j,\theta|d\mu(\theta)$. And then
  \begin{eqnarray}
    & & \sum_{i \not=j} D_{ij} \frac{e^{\ihbar^{-1}(\tilde \chi_j-\tilde \chi_i)(t_\epsilon+T)}-e^{\ihbar^{-1}(\tilde \chi_j-\tilde \chi_i)t_\epsilon}}{\ihbar^{-1}(\tilde \chi_j-\tilde \chi_i)T} +\sum_i D_{ii} - \epsilon \nonumber \\
    & & \qquad < \frac{1}{T} \int_{t_\epsilon}^{t_\epsilon+T} \rho(t)dt < \nonumber\\
    & & \quad \qquad \sum_{i\not=j} D_{ij}  \frac{e^{\ihbar^{-1}(\tilde \chi_j-\tilde \chi_i)(t_\epsilon+T)}-e^{\ihbar^{-1}(\tilde \chi_j-\tilde \chi_i)t_\epsilon}}{\ihbar^{-1}(\tilde \chi_j-\tilde \chi_i)T}+ \sum_i D_{ii} + \epsilon
  \end{eqnarray}
  But $\lim_{T \to +\infty} \frac{e^{\ihbar^{-1}(\tilde \chi_j-\tilde \chi_i)(t_\epsilon+T)}-e^{\ihbar^{-1}(\tilde \chi_j-\tilde \chi_i)t_\epsilon}}{\ihbar^{-1}(\tilde \chi_j-\tilde \chi_i)T} = 0$. Finally we have then
  \begin{equation}
    \lim_{T \to +\infty} \lim_{t \to +\infty} \frac{1}{T} \int_t^{t+T} \rho(t')dt' = \sum_{i,\lambda} |c_{i\lambda}|^2 \int_\Gamma |Z\mu_i,\theta \rangle \langle Z\mu_i,\theta|d\mu(\theta)
  \end{equation}
  If $\lim_{t \to +\infty} \rho(t) = \rho_\infty$, then\\ $\lim_{T \to +\infty} \lim_{t \to +\infty} \frac{1}{T} \int_t^{t+T} \rho(t')dt' = \lim_{T \to +\infty}  \frac{1}{T} \int_t^{t+T} \rho_\infty dt' = \rho_\infty$. It follows that $\rho_\infty = \sum_i p_i \rho_i$ with $p_i = \sum_\lambda |c_{i\lambda}|^2$.
\end{proof}

\section{Example: kicked spin systems kicked controlled by classical flows}
\subsection{The model}
We consider an ensemble of $N$ spins without spin-spin interaction. A constant and uniform magnetic field $\vec B$ is applied on the spin ensemble inducing an energy level splitting by the Zeeman effect. Let $H_0 = \frac{\hbar \omega_1}{2} |\downarrow \rangle \langle \downarrow|$ be the Hamiltonian of a single spin with the Zeeman effect. The spin ensemble is submitted to trains of ultrashort pulses kicking the spins. Let $\omega_0$ be the frequency of the kick trains. These trains of pulses can be modulated following three variables: $\theta^1$ the kick strength, $\theta^2$ the kick delay and $\theta^3$ the kick direction. The modulation follows a discrete time classical flow $\varphi \in \Aut(\mathbb T^3)$. The stroboscopic dynamics of a spin is governed by the evolution operator (see \cite{Viennot2}):
\begin{equation} \label{Uexemple}
  U(\theta) = e^{-\imath \frac{H_0}{\hbar \omega_0} (2\pi - \theta^2)} \left[1+(e^{-\imath \theta^1}-1) W(\theta^3) \right] e^{-\imath \frac{H_0}{\hbar \omega_0} \theta^2}
\end{equation}
where $W(\theta^3) = |w(\theta^3)\rangle \langle w(\theta^3)|$ is the kick operator, $|w(\theta^3)\rangle = \cos \theta^3 |\uparrow \rangle  + \sin \theta^3 |\downarrow \rangle$. So, the stroboscopic dynamics of the $i$-th spin is $\tilde \psi_{n+1}^{(i)} = U(\varphi^n(\theta_{0i})) \tilde \psi_n^{(i)}$ where $\theta_{0i} \in \mathbb T^3$ is the initial condition of the train of pulses kicking the $i$-th spin. For $N$ large, $\sum_{i=1}^N \delta(\theta-\theta_{0i}) \tilde \psi_0^{(i)} \simeq \psi_0(\theta)$ is a state $\psi_0 \in \mathcal K$ of the enlarged Hilbert space $\mathcal K = L^2(\mathbb T^3,d\mu) \otimes \mathbb C^2$.\\

In this section, we study the driven stroboscopic quantum system $(\mathbb T^2,\mu,\varphi,\mathbb C^2,U)$ where the phase space is reduced to $\mathbb T^2$ by setting $\theta^3 = \frac{\pi}{4}$ or by setting $\theta^2 = 0$. $\mu$ is the Haar probability measure on $\mathbb T^2$: $d\mu(\theta^1,\theta^2) = \frac{d\theta^1 d\theta^2}{4 \pi^2}$. We will consider three different classical flows:
\begin{itemize}
\item The cyclic continuous automorphism of the torus (CAT) map defined by $\varphi(\theta) = \left(\begin{array}{cc} -1 & 1 \\ -1 & 0 \end{array} \right) \left( \begin{array}{c} \theta^1 \\ \theta^2 \end{array} \right) \mod \left(\begin{array}{c} 2\pi \\ 2 \pi \end{array} \right)$. Since $\left(\begin{array}{cc} -1 & 1 \\ -1 & 0 \end{array} \right)^3 = \left(\begin{array}{cc} 1 & 0 \\ 0 & 1 \end{array} \right)$, all points $\theta \in \mathbb T^2 \setminus\{0\}$ are 3-cyclic by this flow. $0 \in \mathbb T^2$ is the single fixed point. Due to this cyclicity, the Koopman spectrum of this flow is $\Sp(\mathcal T) = \{1, e^{\frac{2 \imath \pi}{3}}, e^{\frac{4 \imath \pi}{3}} \}$.
\item The Arnold's CAT map defined by $\varphi(\theta) = \left(\begin{array}{cc} 1 & 1 \\ 1 & 2 \end{array} \right) \left( \begin{array}{c} \theta^1 \\ \theta^2 \end{array} \right) \mod \left(\begin{array}{c} 2\pi \\ 2 \pi \end{array} \right)$. The Arnold's CAT map is a chaotic flow, mixing (and then ergodic) on the whole of $\mathbb T^2$. $0 \in \mathbb T^2$ is a fixed point, and we have an infinity but countable number of cyclic points (forming then a set of zero measure by $\mu$). Due to its chaotic behaviour, its discrete Koopman spectrum is reduced to $\Sp_d(\mathcal T) = \{1\}$ and its continuous Koopman spectrum is $\Sp_{cont} (\mathcal T) = U(1) \setminus \{1\}$ ($U(1)$ being the unit circle in $\mathbb C$).
  \item The Chirikov standard map defined by $\varphi(\theta) = \left( \begin{array}{c} \theta^1 + K \sin(\theta^2) \mod 2\pi \\ \theta^2 + \varphi^1(\theta) \mod 2\pi \end{array} \right)$, where $K$ is an adjustable parameter. This flow presents a chaotic sea (mixing and then ergodic component) with islands of stability (regions of periodic orbits). The respective sizes of the chaotic sea and of the islands of stability depend on $K$, more $K$ is large more the flow is chaotic. 
\end{itemize}
The dynamics of the different flows is represented figure \ref{classicalflow}.
\begin{figure}
  \begin{center}
    \includegraphics[width=10cm]{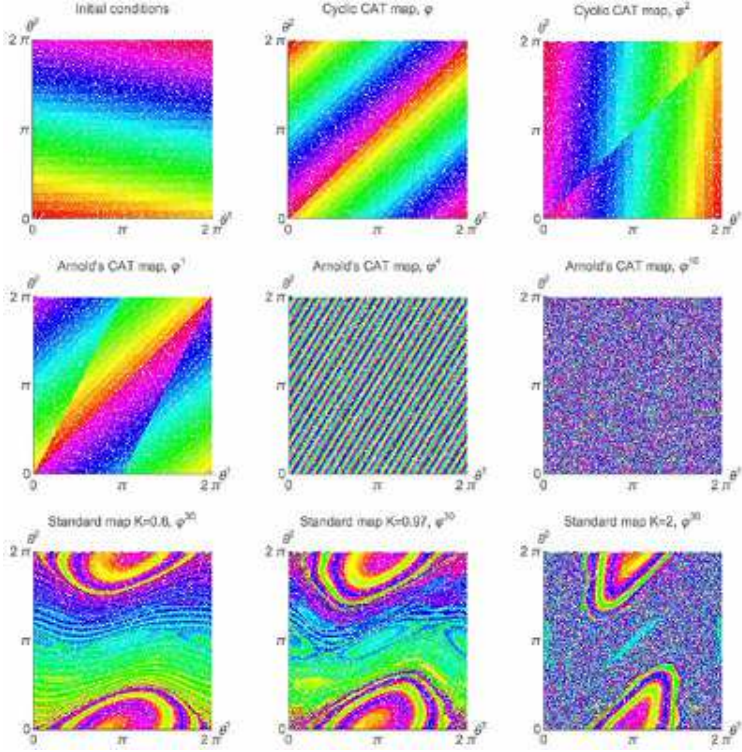}
    \caption{\label{classicalflow} Evolution of the points of $\mathbb T^2$ under the three classical flows, colored with respect to their initial positions. Up: the three configurations of the 3-cyclic CAT map. Middle: first, fourth and tenth iterations of the Arnold's CAT map. The chaotic character of this flow is illustrated by the noisy aspect of the last graph. Bottom: Thirtieth iteration of the standard map for $K=0.6$ (barely chaotic), $K=0.97$ (small chaotic sea with a lot of islands of stability) and $K=2$ (highly chaotic, large chaotic sea with some islands of stability).}
    \end{center}
\end{figure}
Remark: all figures presented in this section are computed with numerical simulations based on the semi-analytical formulea presented in this paper which involve repetition of the action of the evolution operator eq. \ref{Uexemple}. They are realized by using the \textit{Mathematica} software. The number of spins is $N=40 000$ in the simulations.

\subsection{A cyclic CAT map}
\subsubsection{Quasienergy states and SK modes:}
We consider first the case of the 3-cyclic CAT map. The elements of the orbifold $\mathbb T^2/\Phi$ are the 3-cyclic orbits covering $\mathbb T^2$ ($\Phi = \left\{\small \left(\begin{array}{ccc} 1 & 0 \\ 0 & 1 \end{array}\right), \left(\begin{array}{ccc} -1 & 1 \\ -1 & 0 \end{array}\right), \left(\begin{array}{ccc} -1 & 1 \\ -1 & 0 \end{array}\right)^2 \normalsize \right\}$ is the cyclic group acting on $\mathbb T^2$ as $\varphi$). Due to properties \ref{fundquasifixed} and \ref{cycltocycl}, each element of $\mathbb T^2/\Phi$ involves a different couple of fundamental quasienergies:
\begin{equation}
  \Sp_{fqe} = \bigcup_{\theta \in \mathbb T^2/\Phi} \{\tilde \chi_{\uparrow,\theta},\tilde \chi_{\downarrow,\theta} \}
\end{equation}
where
\begin{equation}
  \{e^{-3 \imath \tilde \chi_{\uparrow,\theta}}, e^{-3 \imath \tilde \chi_{\downarrow,\theta}} \} = \Sp\left(U(\varphi^2(\theta))U(\varphi(\theta)) U(\theta)\right)
\end{equation}
We can choose the labels $\uparrow/\downarrow$ and $\theta$ in order to $\theta \mapsto \tilde \chi_{\uparrow/\downarrow}(\theta) \equiv \tilde \chi_{\uparrow/\downarrow,\theta}$ be continuous functions. Note that $\tilde \chi_{\uparrow,\theta}$ is not a $\theta$-dependent fundamental quasienergy, $\tilde \chi_{\uparrow,\theta}$ and $\tilde \chi_{\uparrow,\theta'}$ ($\theta' \not= \theta$) are two distinct fundamental quasienergies; $\theta$ are just continuous indices due to the continuous character of the fundamental quasienergy spectrum. The choice of grouping the fundamental quasienergies into two continuous functions is just a convenience convention. The fundamental quasienergy spectrum is represented figure \ref{quasienergCAT}.
\begin{figure}
  \begin{center}
    \includegraphics[width=5.5cm]{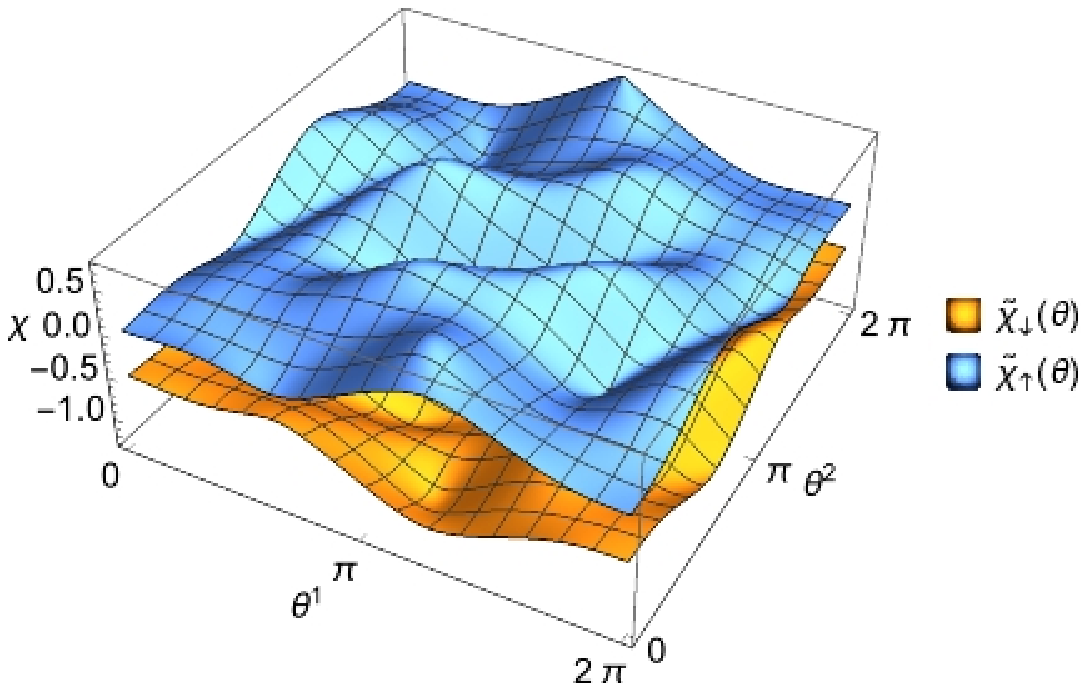} \includegraphics[width=5.5cm]{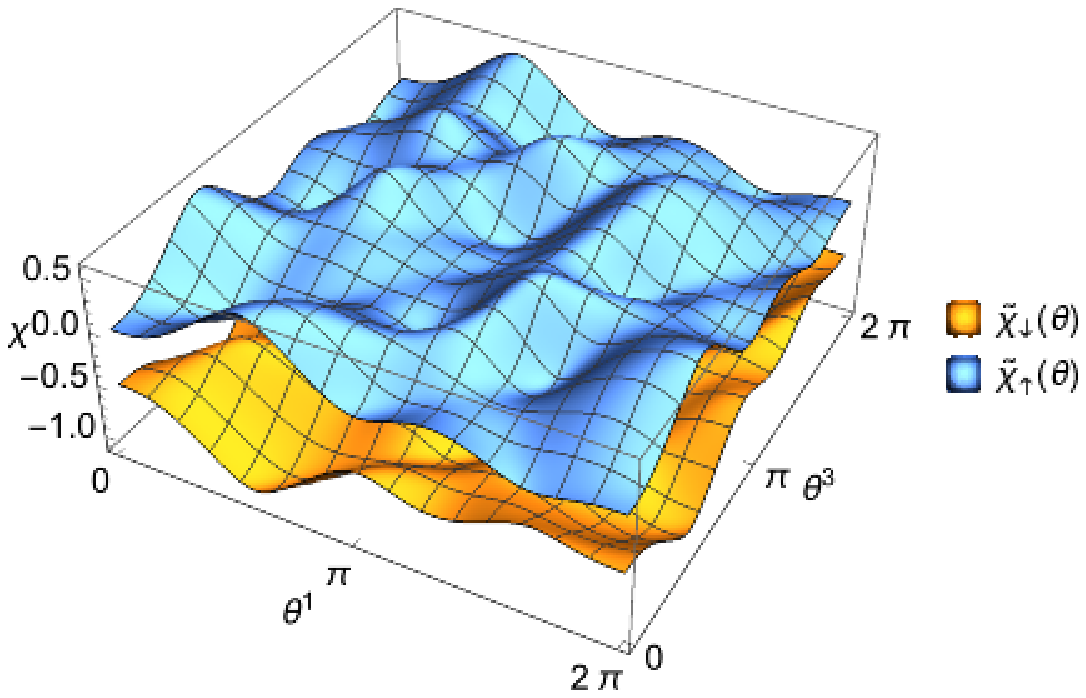} 
    \caption{\label{quasienergCAT} $\Sp_{fqe}$ organized as two continuous functions $\theta \mapsto \tilde \chi_{\uparrow/\downarrow}(\theta)$ for the cyclic CAT map, with $\frac{\omega_1}{\omega_0} = 2.5$ and $\theta^3 = \frac{\pi}{4}$ (left) or $\theta^2=0$ (right).}
  \end{center}
\end{figure}
The fundamental quasienergy states have the form: 
\begin{equation}
  |Z\mu_{\uparrow/\downarrow,\theta_0},\theta \rangle = \sum_{n=0}^2 |Z\mu_{\uparrow/\downarrow,\varphi^n(\theta_0)} \rangle \otimes \delta(\theta-\varphi^n(\theta_0))
\end{equation}
where $\theta_0 \in \mathbb T^2/\Phi$ and $|Z\mu_{\uparrow/\downarrow,\theta_0} \rangle$ is eigenvector of $U(\varphi^2(\theta_0))U(\varphi(\theta_0)) U(\theta_0)$. We can superpose the fundamental quasienergy states of the different orbits to obtain continuous states on $\mathbb T^2$:
\begin{equation}
  |Z \tilde \mu_{\uparrow/\downarrow},\theta \rangle = \int_{\mathbb T^2/\Phi}^\oplus |Z\mu_{\uparrow/\downarrow,\theta_0},\theta \rangle d\theta^1_0 d\theta^2_0
\end{equation}
Figure \ref{quasistateCAT} represents this state.
\begin{figure}
  \begin{center}
    \includegraphics[width=5.5cm]{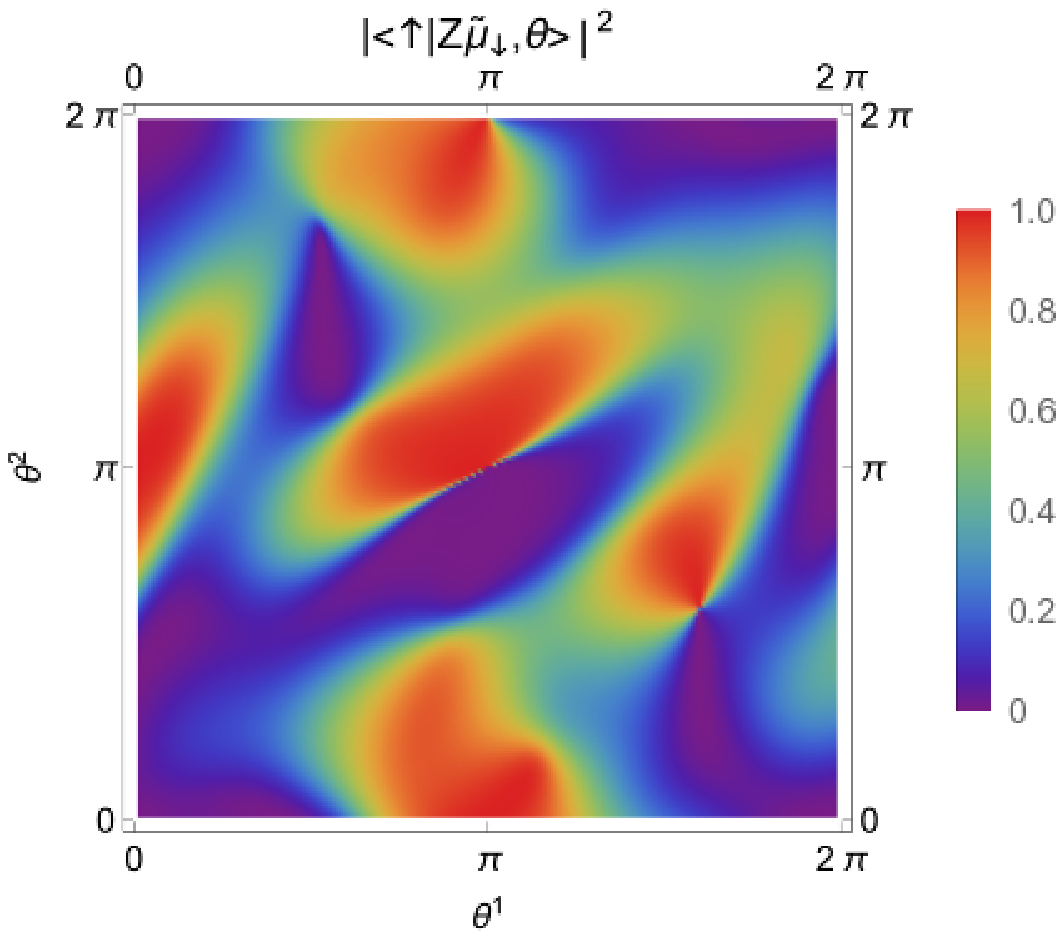} \includegraphics[width=5.5cm]{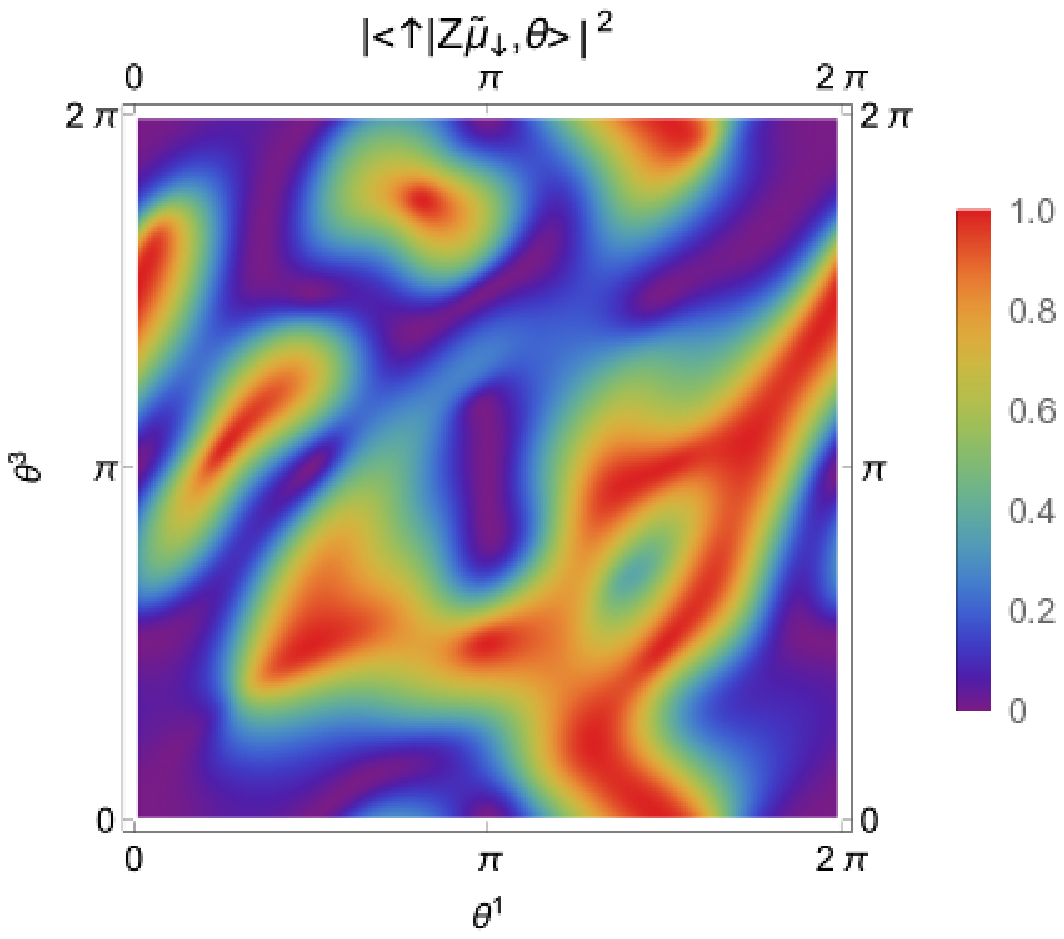}
    \caption{\label{quasistateCAT} Occupation probability of the state $|\uparrow\rangle$ with respect to $\theta$ for the fundamental quasienergy state $|Z\tilde \mu_{\downarrow},\theta \rangle$ of the cyclic CAT map, with $\frac{\omega_1}{\omega_0} = 2.5$ and $\theta^3 = \frac{\pi}{4}$ (left) or $\theta^2=0$ (right).}
  \end{center}
\end{figure}
The structures appearing in $\mathbb T^2$ are related to the structure of $\theta \mapsto U(\theta)$ (as we see it by comparing the two choices -- modulation of the kick delay or of the kick direction --).\\

In section 3 we have introduced the operator $V_i(\theta) = \lim_{N \to +\infty} V_i^{(N)}(\theta)$ with
\begin{equation}
  V_i^{(N)}(\theta) = \frac{1}{N} \sum_{n=0}^{N-1} e^{\imath n \tilde \chi_i} U(\varphi^{n-1}(\theta))...U(\theta)
\end{equation}
For an ergodic orbit, $V_i(\theta)^{-1}$ permits to compute the fundamental quasienergy states. Here, the limit does not exist since the orbits are not ergodic. But $V_i^{(N)}(\theta)$ could be interpreted as a kind of generalization of Fourier modes of the dynamics, in a same manner that the Koopman modes \cite{Budisic}. We consider then the state:
\begin{equation}
  |SK_\uparrow^{(N)},\theta \rangle = V_\uparrow^{(N)}(\theta)^{-1}V_\uparrow^{(N)}(0)|Z\mu_{\uparrow,0},0 \rangle
\end{equation}
that we call a Schr\"odinger-Koopman (SK) mode. It is represented figure \ref{SKmodeCAT}.
\begin{figure}
  \begin{center}
    \includegraphics[width=5.5cm]{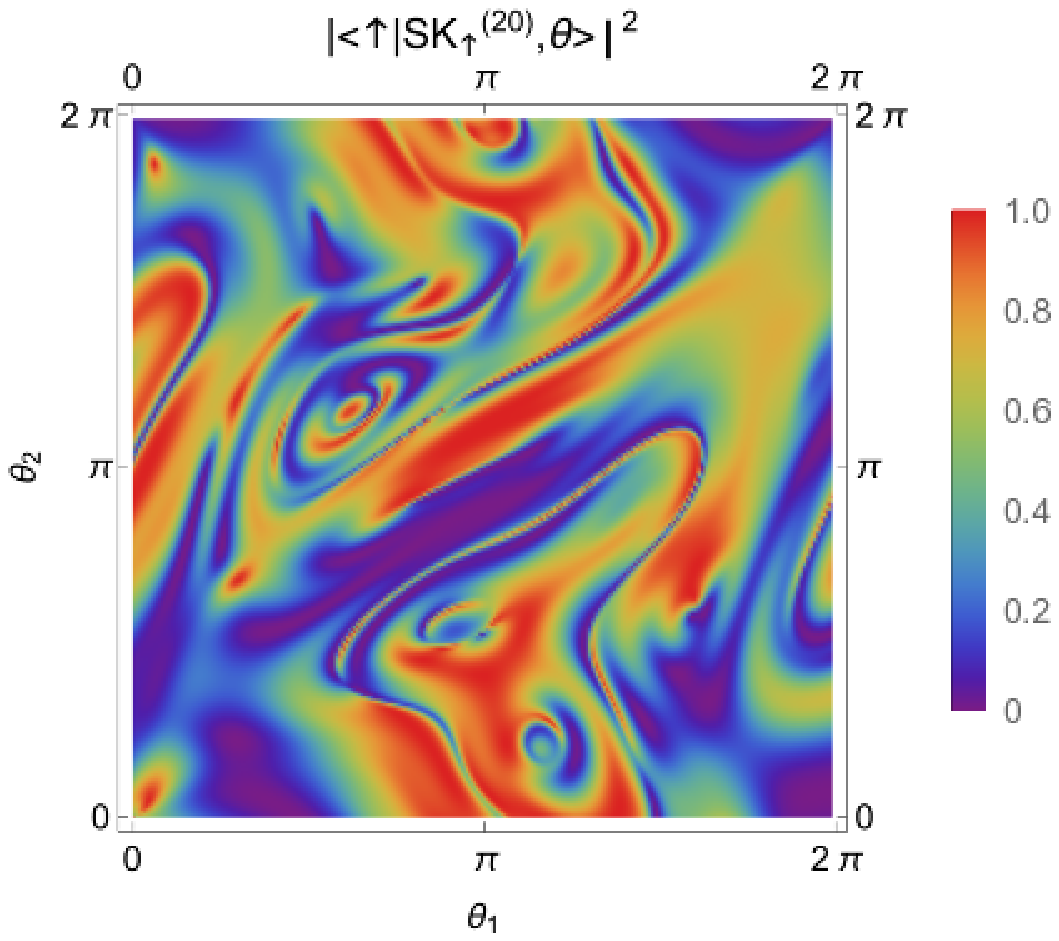} \includegraphics[width=5.5cm]{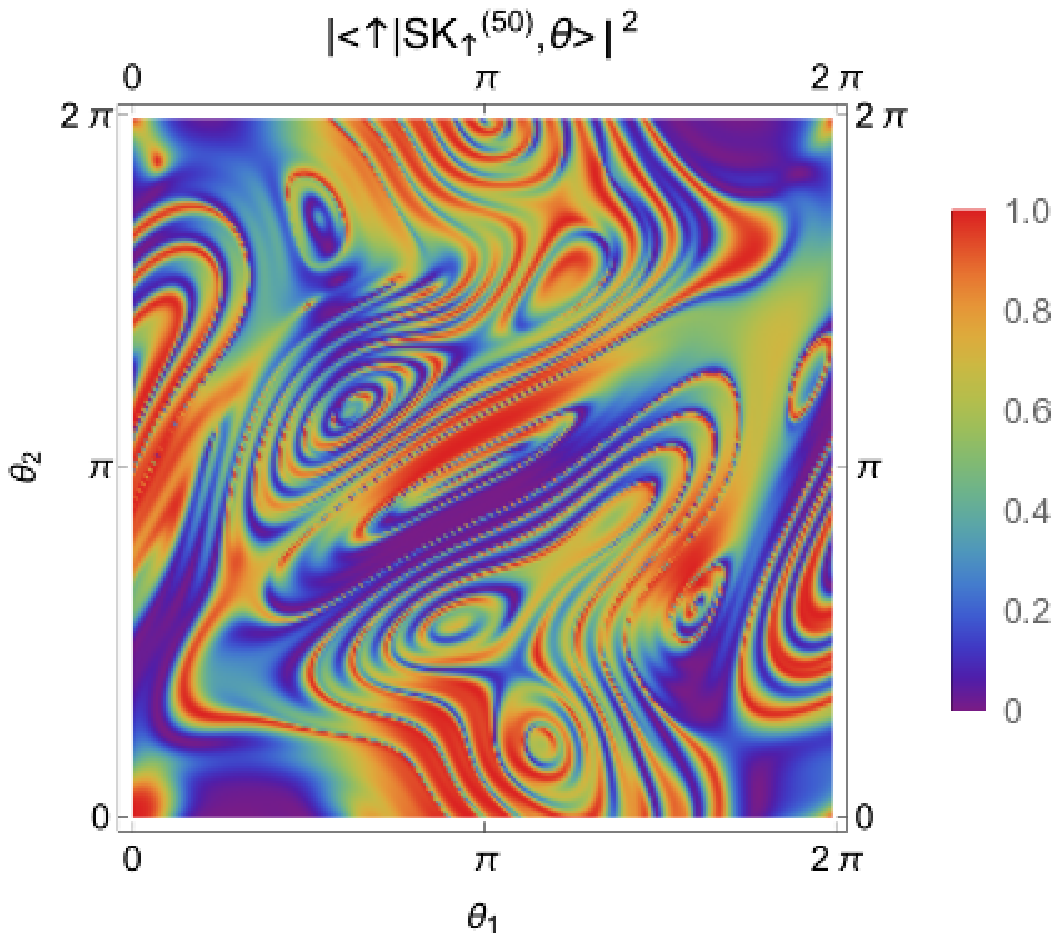}\\
    \includegraphics[width=5.5cm]{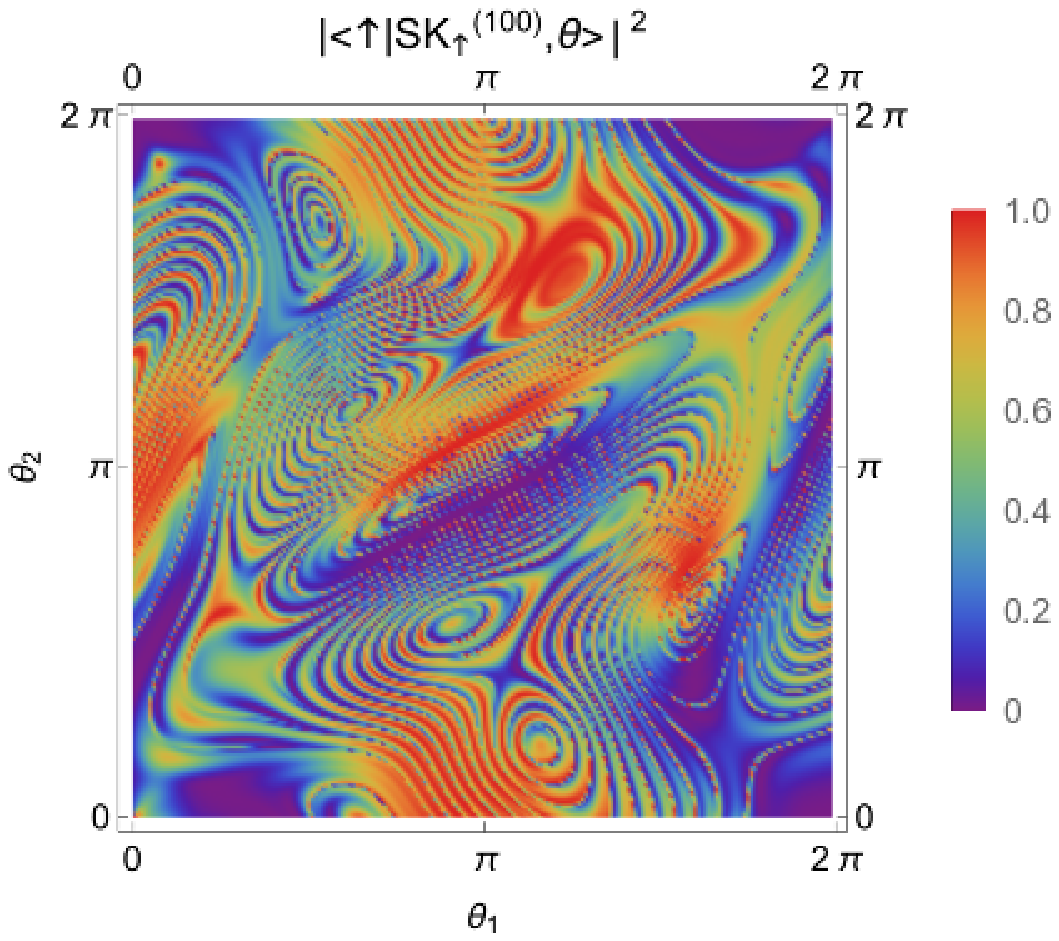} \includegraphics[width=5.5cm]{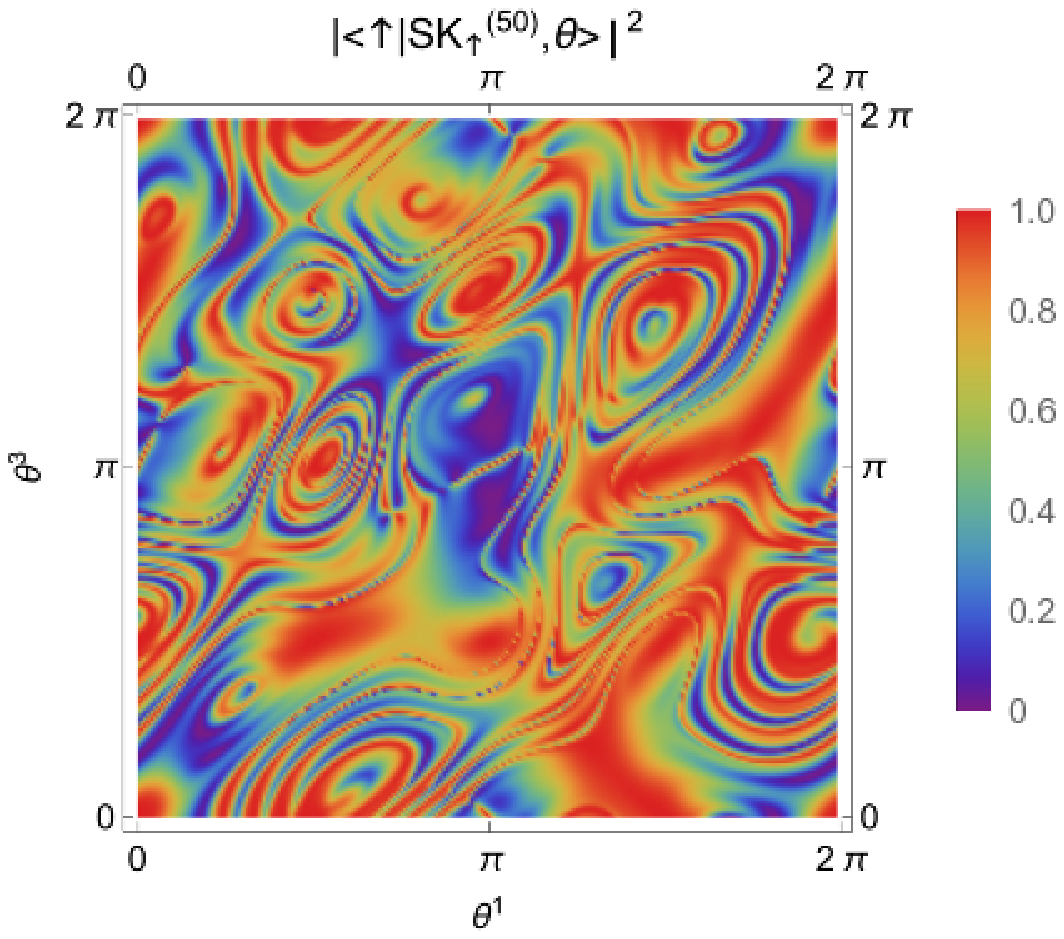}
    \caption{\label{SKmodeCAT} Occupation probability of the state $|\uparrow\rangle$ with respect to $\theta$ for the SK mode $|SK_{\uparrow}^{(N)},\theta \rangle$ of the cyclic CAT map, with $\frac{\omega_1}{\omega_0} = 2.5$ and $\theta^3 = \frac{\pi}{4}$ (top and bottom left) or $\theta^2=0$ (bottom right).}
  \end{center}
\end{figure}
We remark that the structures appearing in the fundamental quasienergy states (fig. \ref{quasistateCAT}) can be refound in the SK modes added with ``interferences''.

\subsubsection{Dynamics:}
We consider the dynamics for four initial conditions:
\begin{itemize}
\item $\psi_0(\theta) = \frac{\mathbb{I}_{D_0}(\theta)}{\mu(D_0)} \frac{1}{\sqrt 2}(|\uparrow \rangle + |\downarrow \rangle)$ where $\mathbb{I}_{D_0}$ is the characteristic function on $D_0$ a small square of side length equal to $10^{-3}$. This state corresponds to a highly coherent ensemble of spins with a small uniform dispersion of the first kicks.
\item $\psi_0(\theta) =  \frac{1}{\sqrt 2}(|\uparrow \rangle + |\downarrow \rangle)$ which corresponds to a large uniform dispersion (on the whole of $\mathbb T^2$) of the first kicks.
\item $\psi_0(\theta) = |Z\tilde \mu_\uparrow,\theta \rangle$ the superposition of fundamental quasienergy states.
\item $\psi_0(\theta) = |Z\mu_{\uparrow,\theta_0},\theta \rangle$ a fundamental quasienergy state.
\end{itemize}
The dynamics $\psi_n = U_K^n \psi_0$ corresponds to the dynamics $\psi_n^{(i)} = U(\varphi^n(\theta_{0i})) \psi_0(\theta_{0i})$ of a large number of spins with $\{\theta_{0i}\}_i$ randomly chosen following the probability distribution of density function $\tr_{\mathbb C^2} |\psi_0(\theta) \rrangle \llangle \psi_0(\theta)|$ (the numerical simulations are realized with such a spin ensemble). We consider then the density matrix $\rho_n = \tr_{L^2(\mathbb T^2,d\mu)} |\psi_n \rrangle \llangle \psi_n| = \lim_{N \to +\infty} \frac{1}{N} \sum_{i=1}^N |\psi_n^{(i)} \rangle \langle \psi_n^{(i)}|$ corresponding to the mixed state of the spin ensemble. The results of the different dynamics is represented figure \ref{dynamicsCAT}.
\begin{figure}
  \begin{center}
    \includegraphics[width=4.2cm]{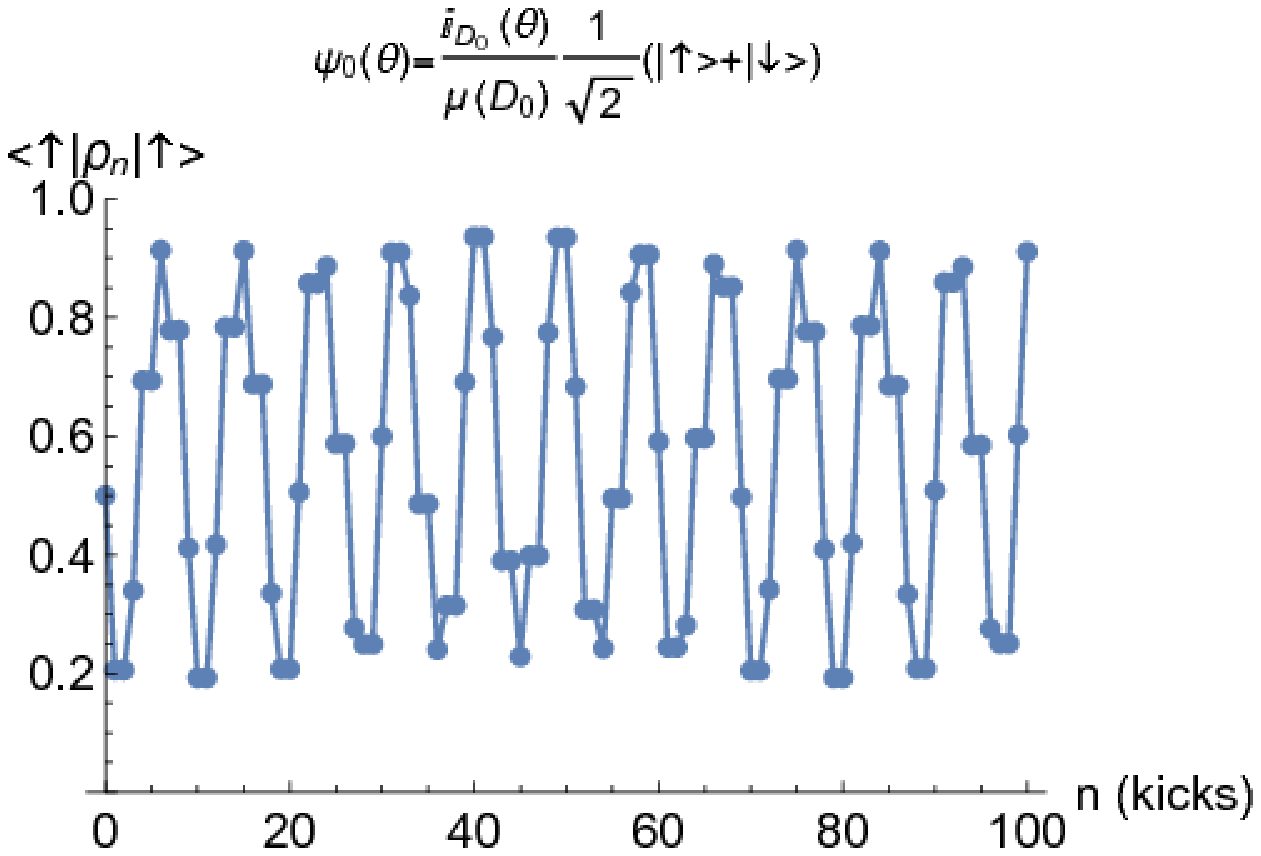} \includegraphics[width=4.2cm]{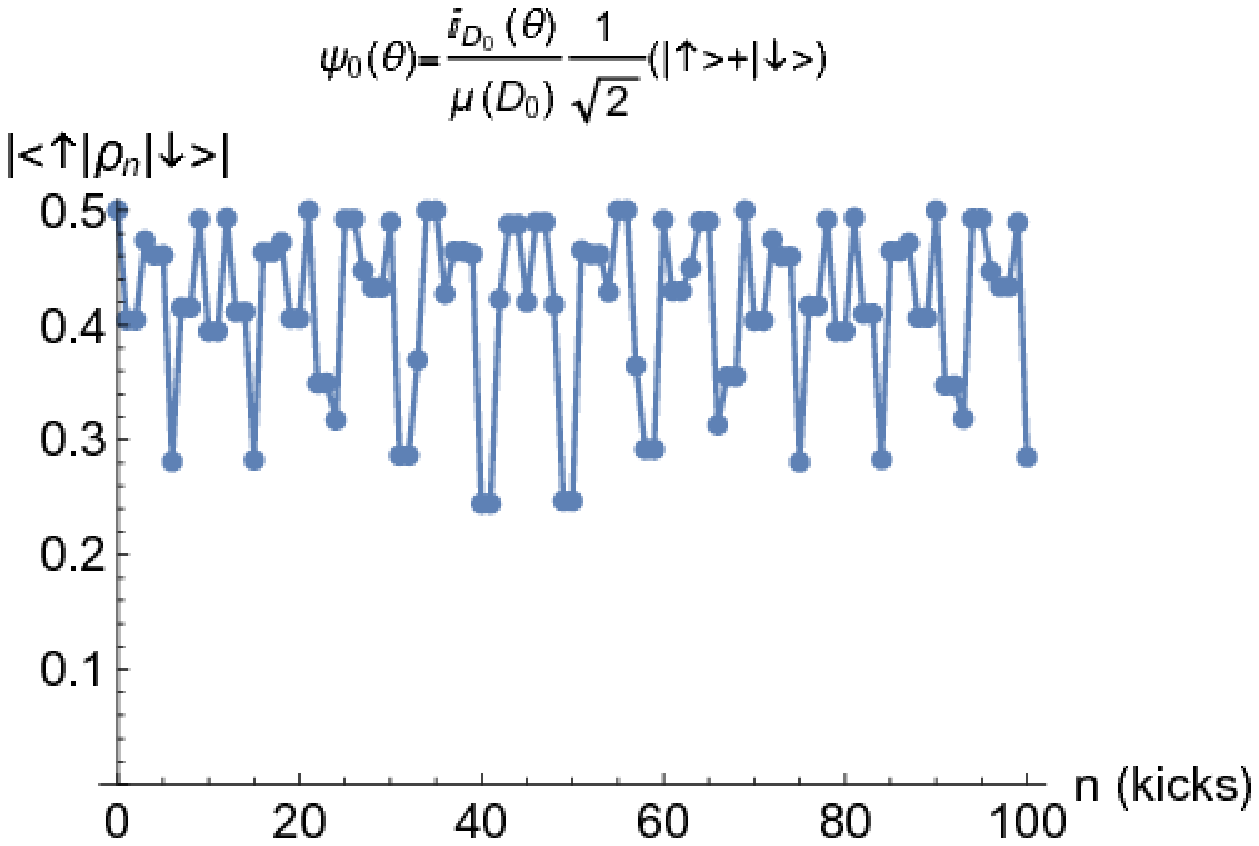} \includegraphics[width=4.2cm]{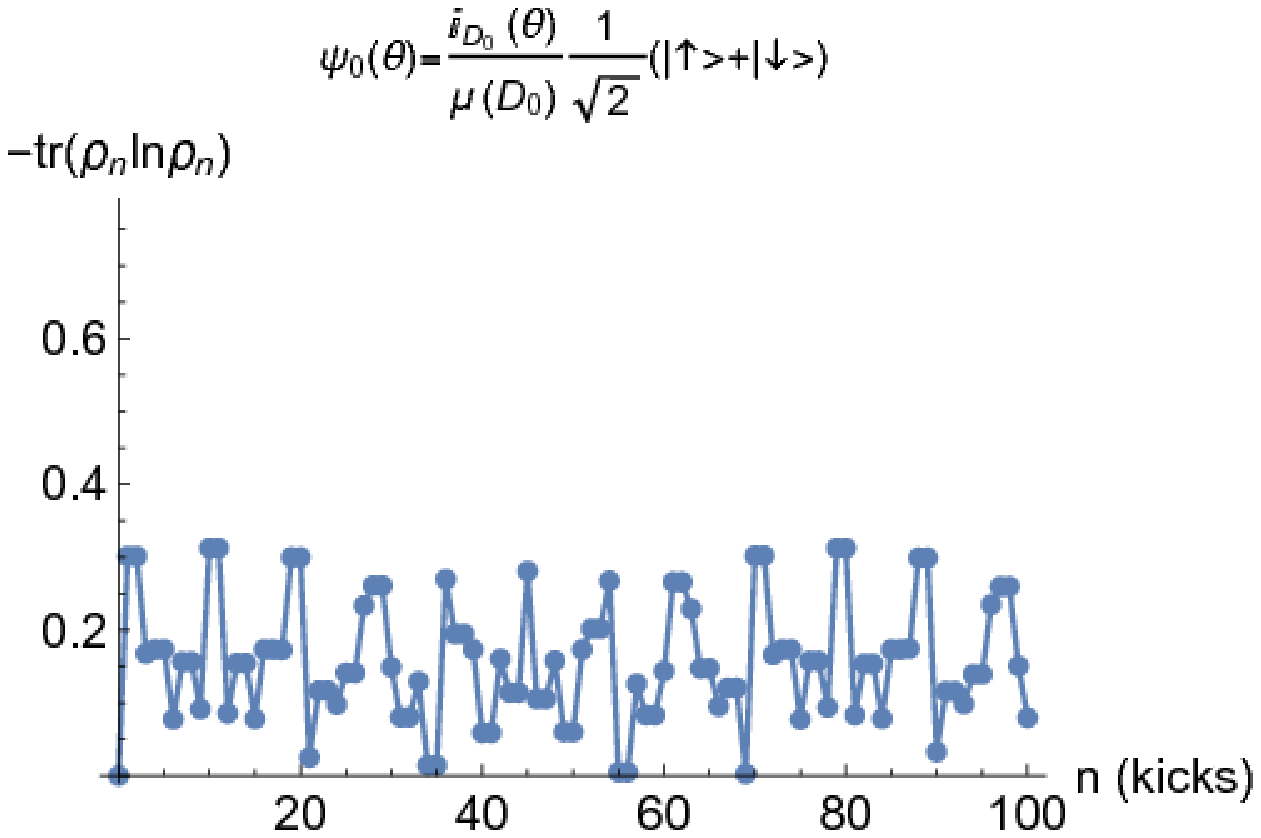}\\
    \includegraphics[width=4.2cm]{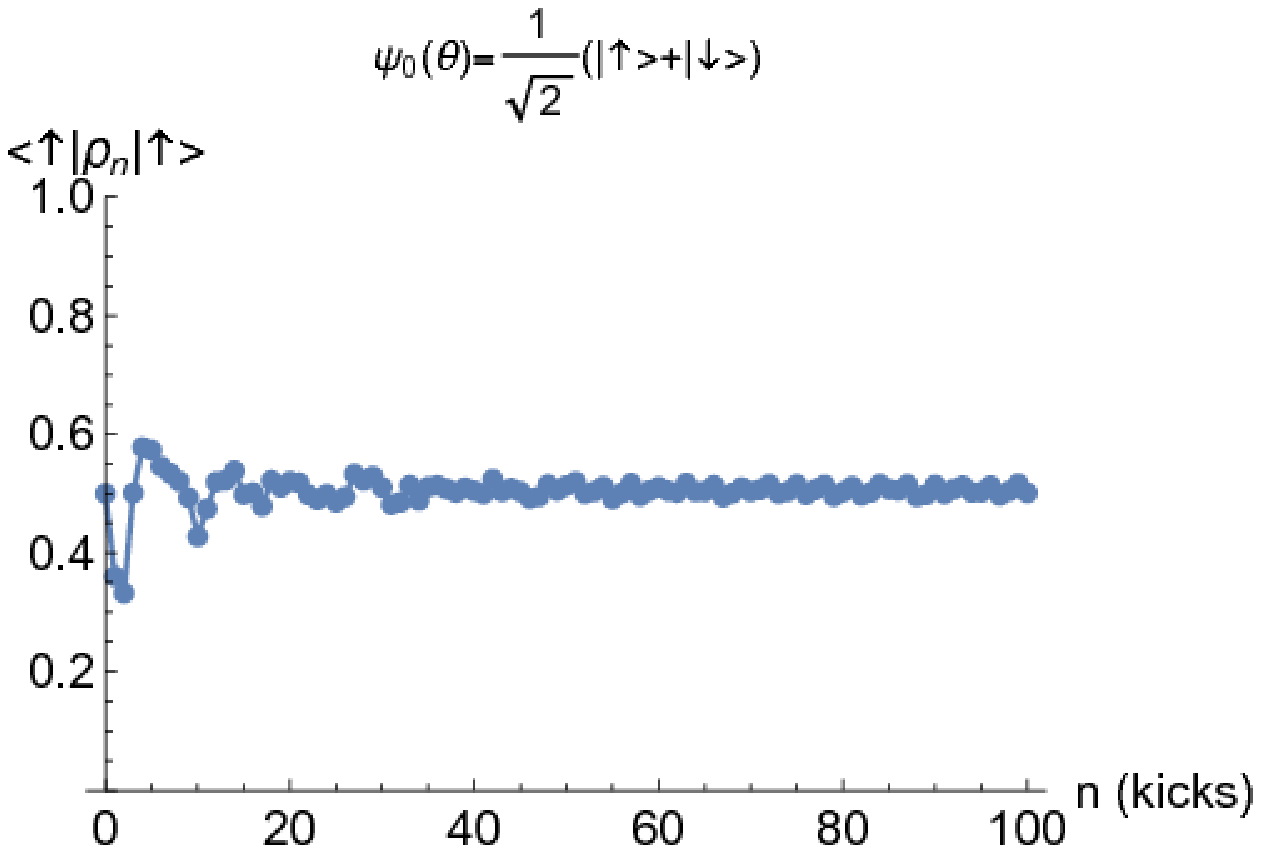} \includegraphics[width=4.2cm]{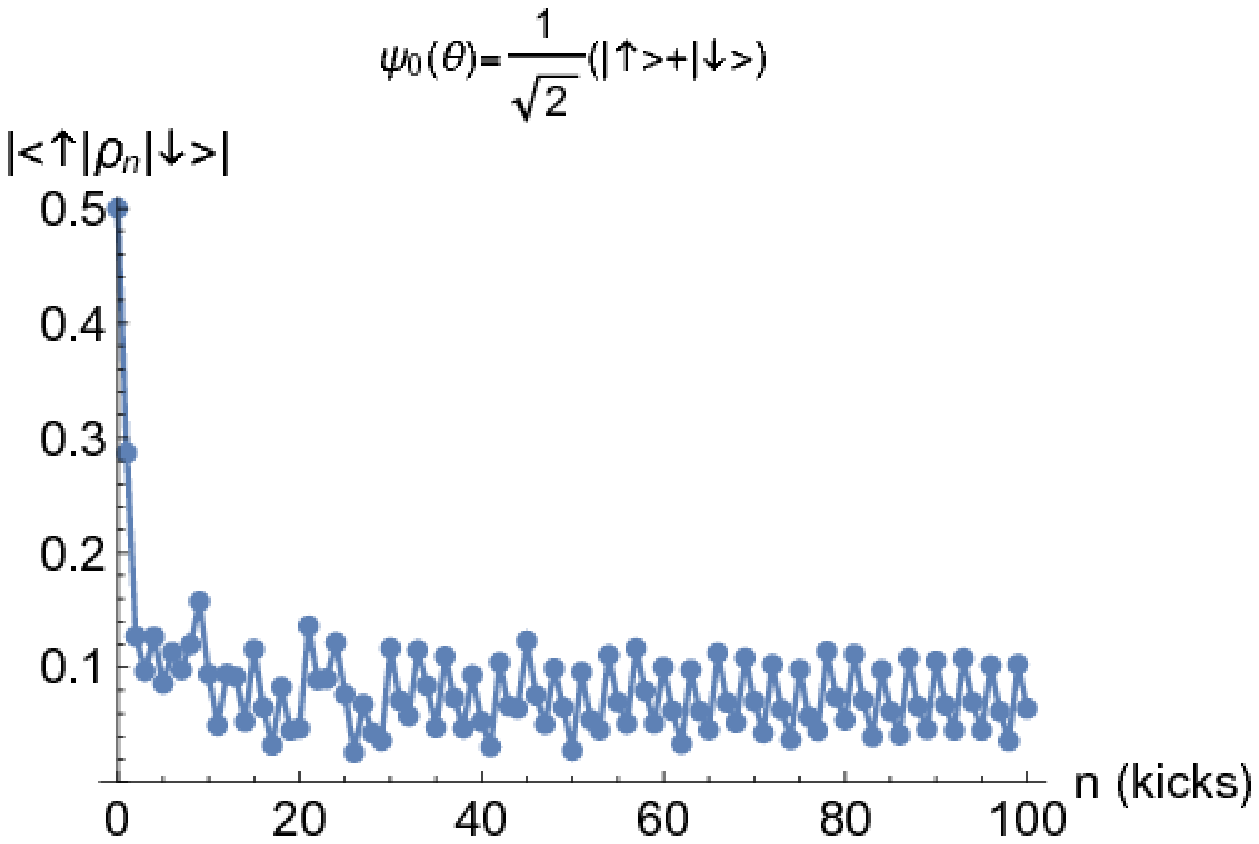} \includegraphics[width=4.2cm]{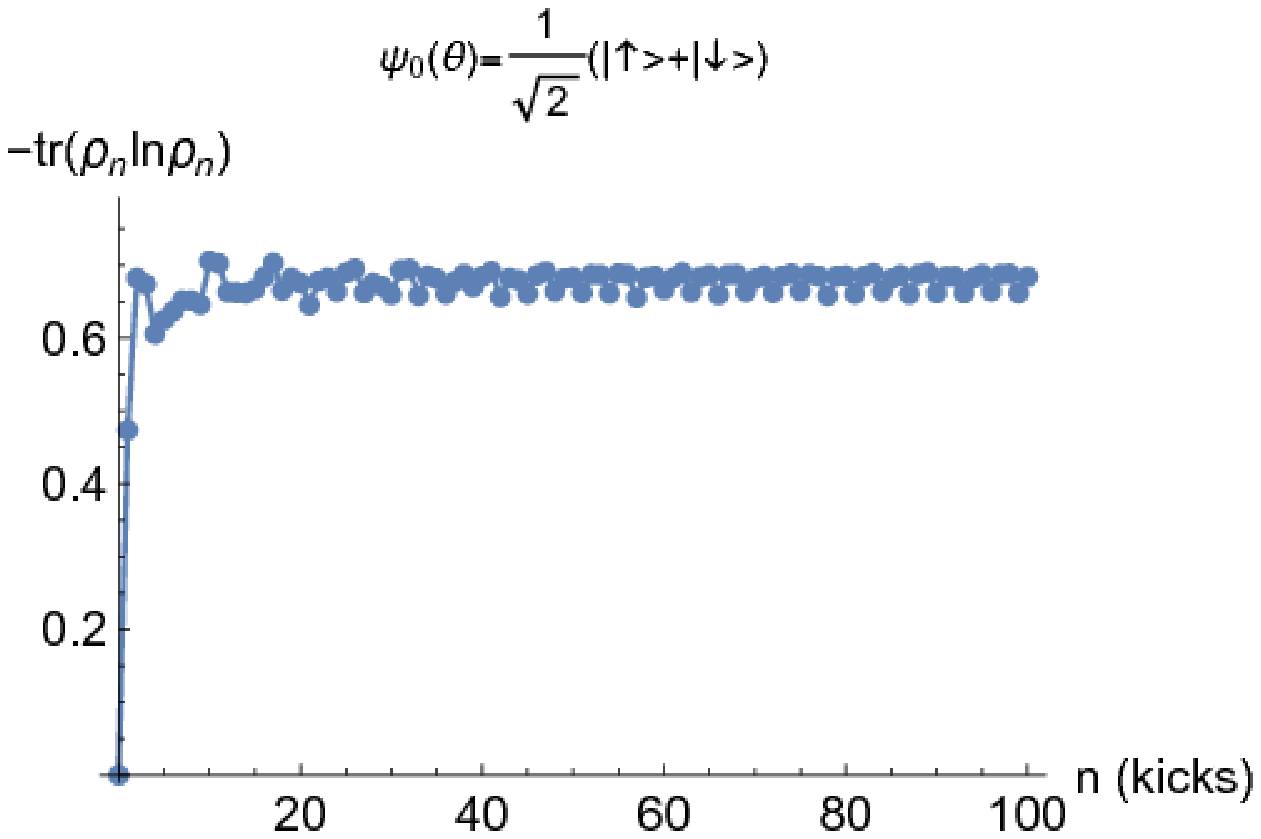}\\
    \includegraphics[width=4.2cm]{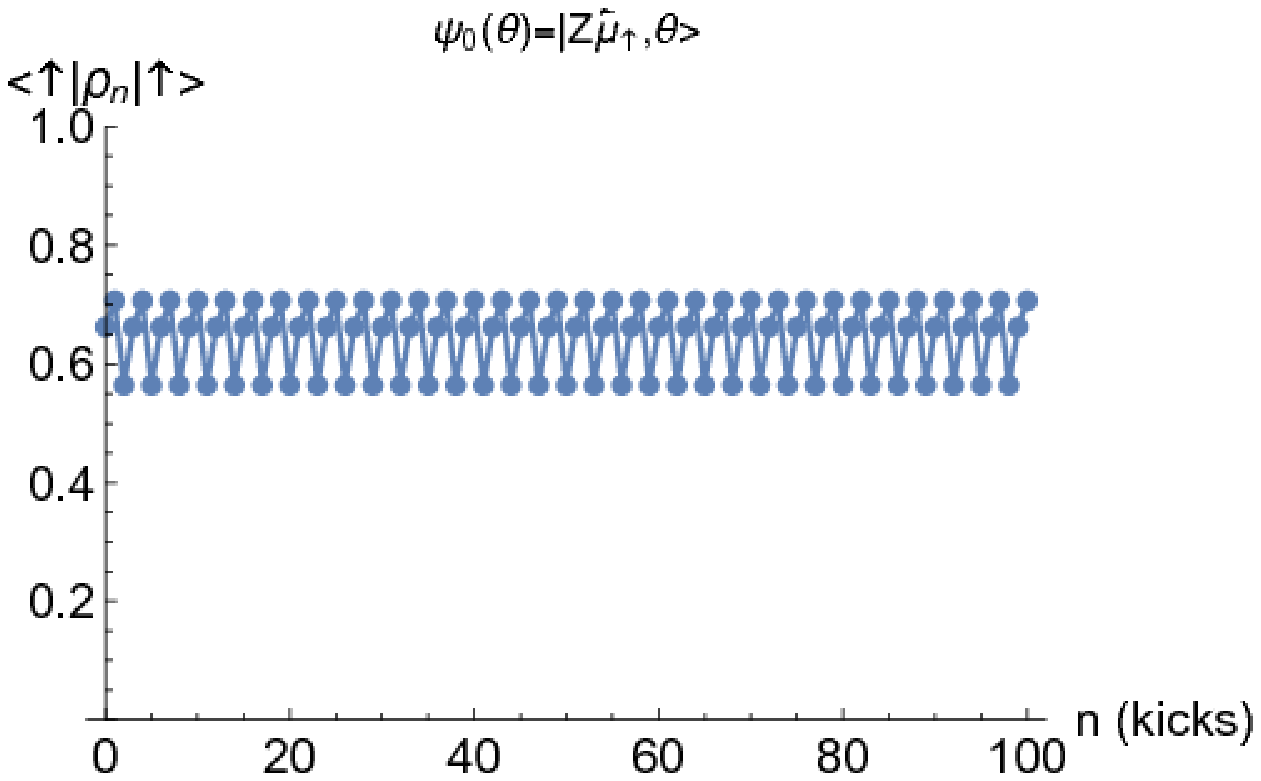} \includegraphics[width=4.2cm]{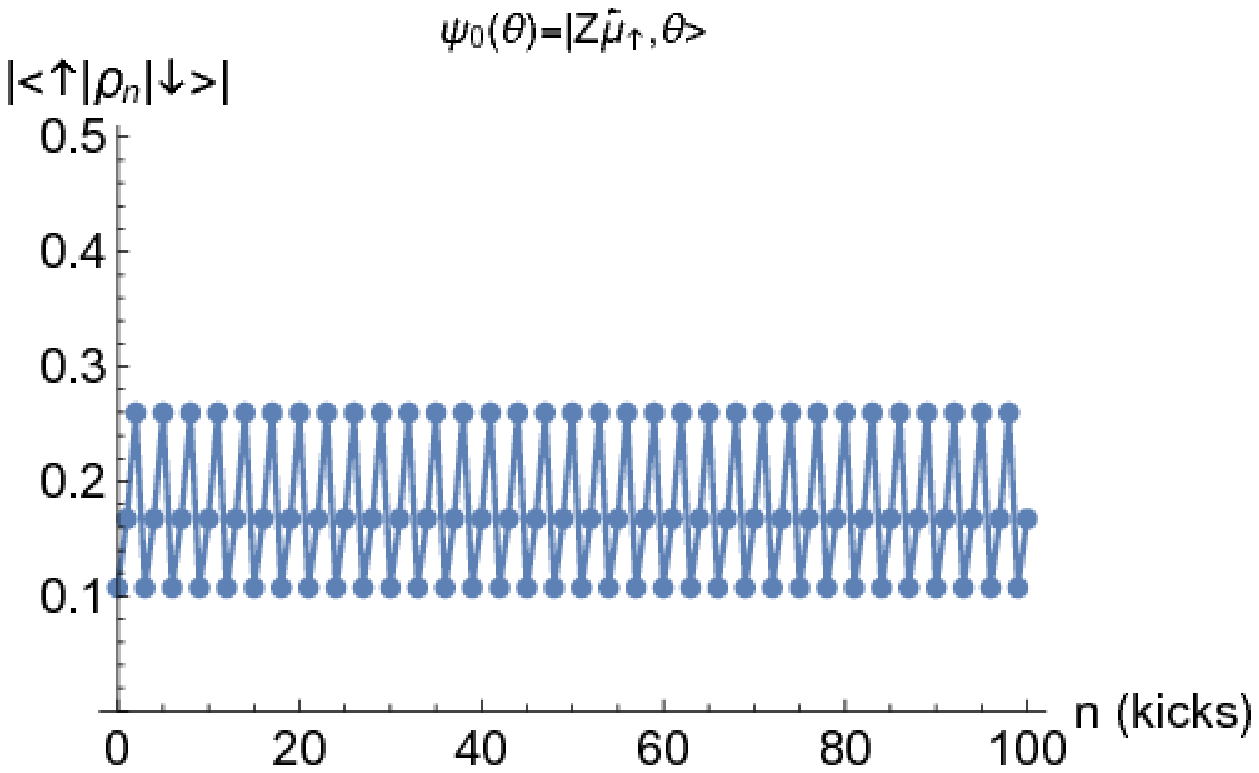} \includegraphics[width=4.2cm]{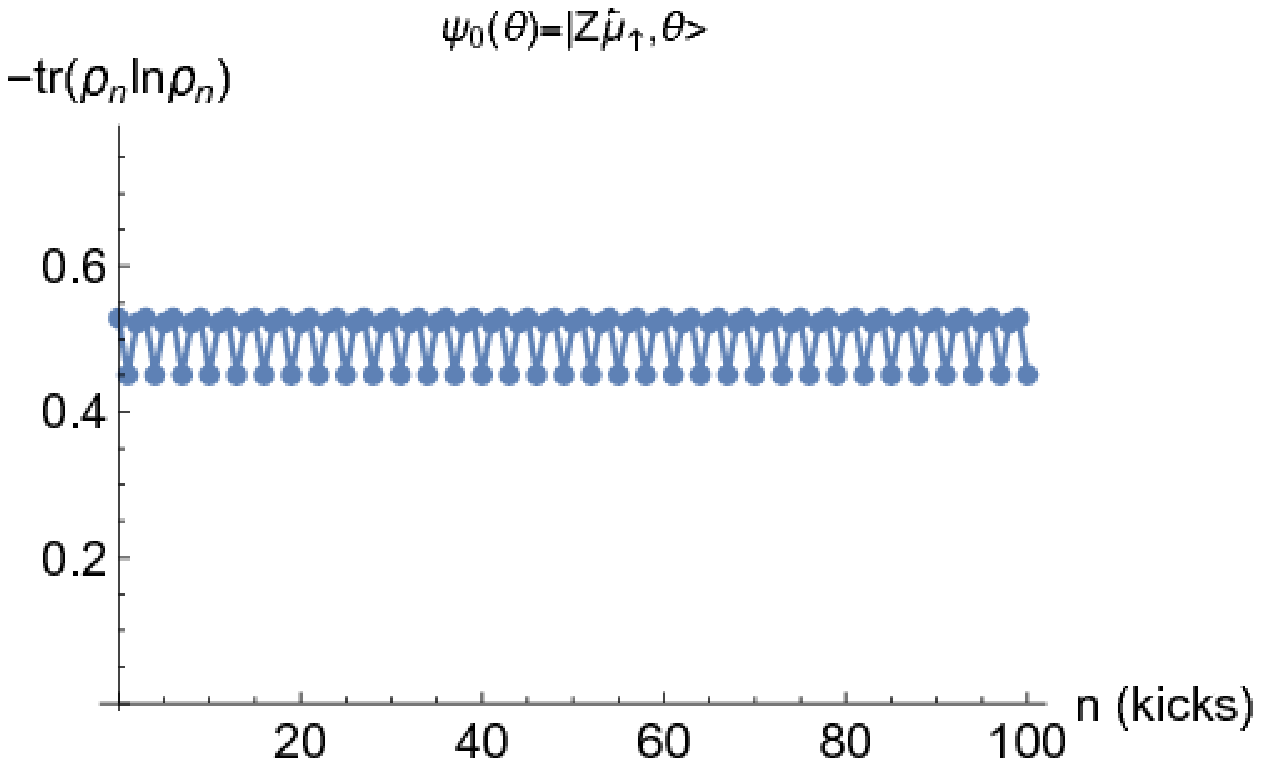}\\
    \includegraphics[width=4.2cm]{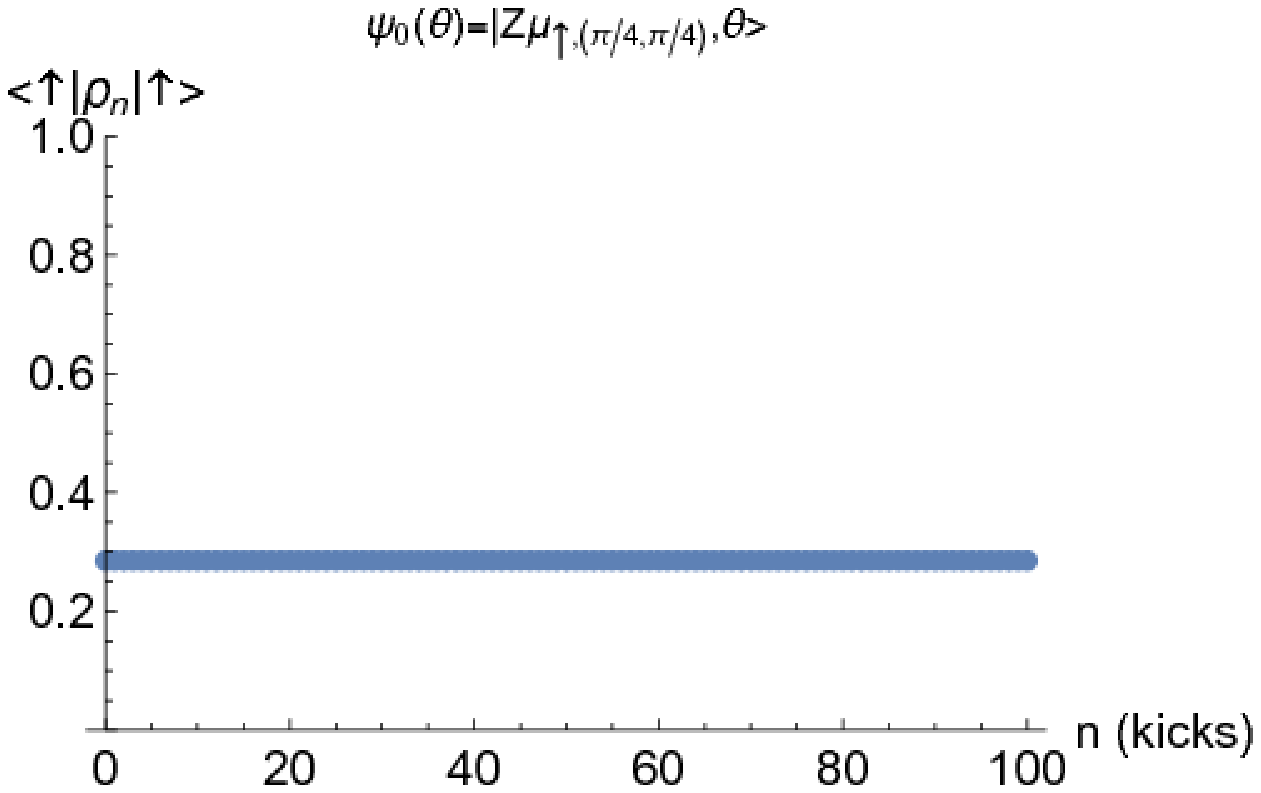} \includegraphics[width=4.2cm]{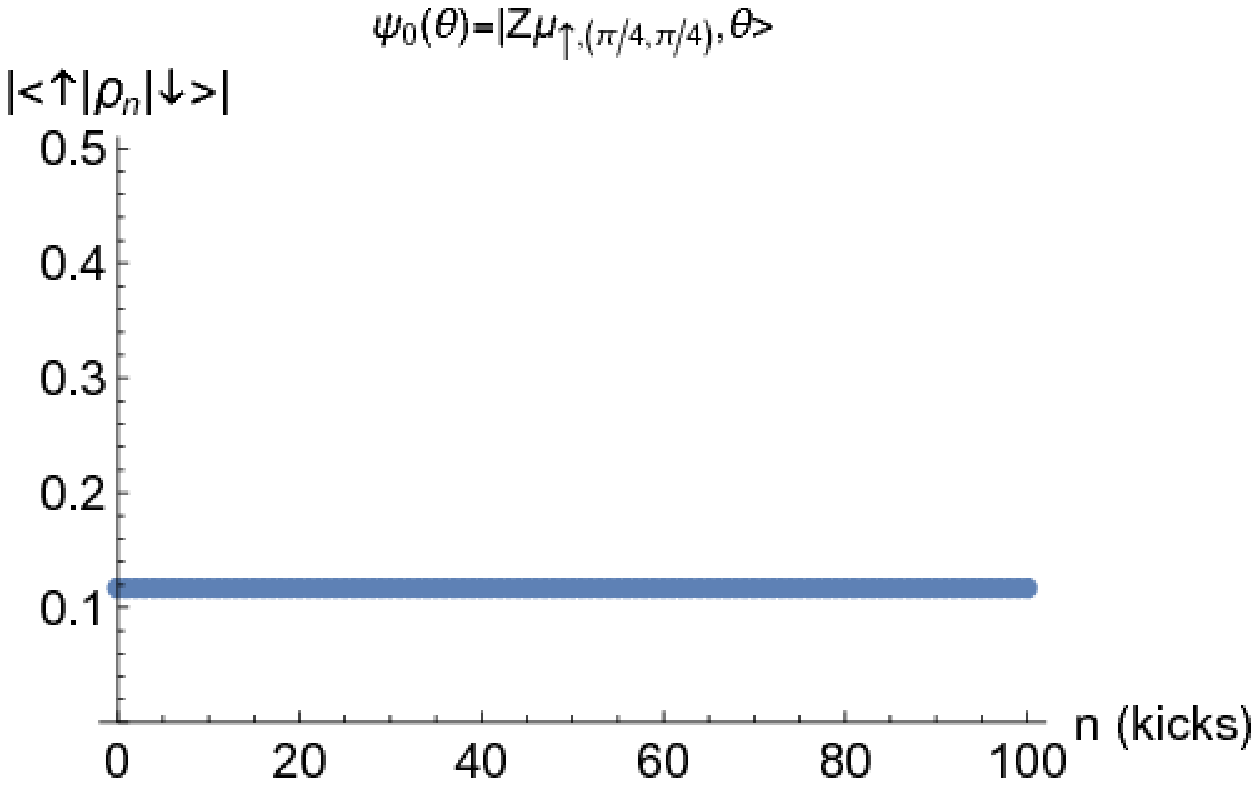} \includegraphics[width=4.2cm]{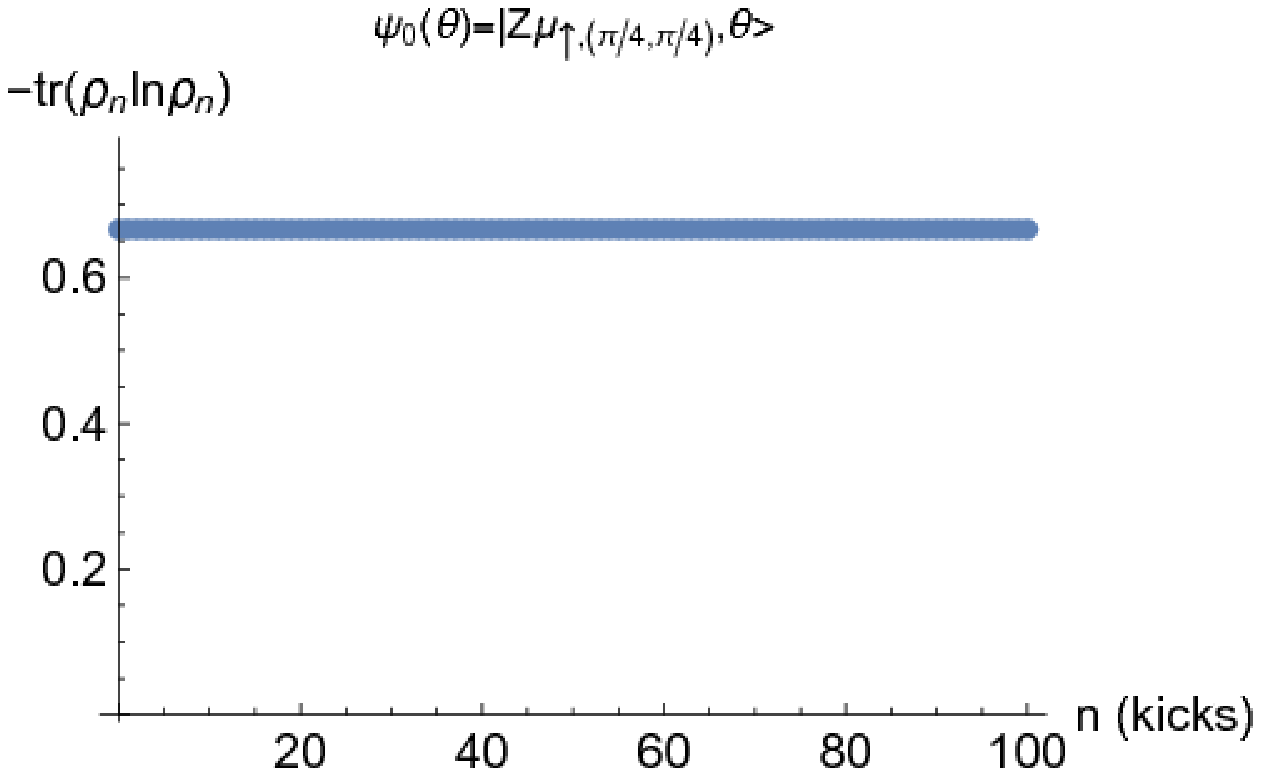}\\
    \caption{\label{dynamicsCAT} Population of the state $|\uparrow \rangle$ (left), coherence (center) and von Neumann entropy (right) of the mixed state for the stroboscopic dynamics of the spin ensemble (with $\frac{\omega_1}{\omega_0} = 2.5$ and $\theta^3 = \frac{\pi}{4}$) where the kick modulation is governed by the cyclic CAT map. The initial condition is for a small uniform dispersion of the first kicks (first line), a large uniform dispersion (second line), the superposition of fundamental quasienergy states (third line), and a single fundamental quasienergy state (last line).}
  \end{center}
\end{figure}
The classical flow being cyclic, it does not generate decoherence for a small initial dispersion of the kicks. The only one decoherence phenomenon occurs for a large initial 
dispersion due to the large dephasing induced in the spin dynamics. As expected, the fundamental quasienergy state is a steady state. The superposition of fundamental quasienergy states is on average stationnary but with small oscillations.\\

It is interesting to consider the structure of $\psi_n(\theta)$ for the initial uniform state $\psi_0(\theta) = \frac{1}{\sqrt 2}(|\uparrow \rangle + |\downarrow \rangle)$, figure \ref{psiCAT}.
\begin{figure}
  \begin{center}
    \includegraphics[width=5.5cm]{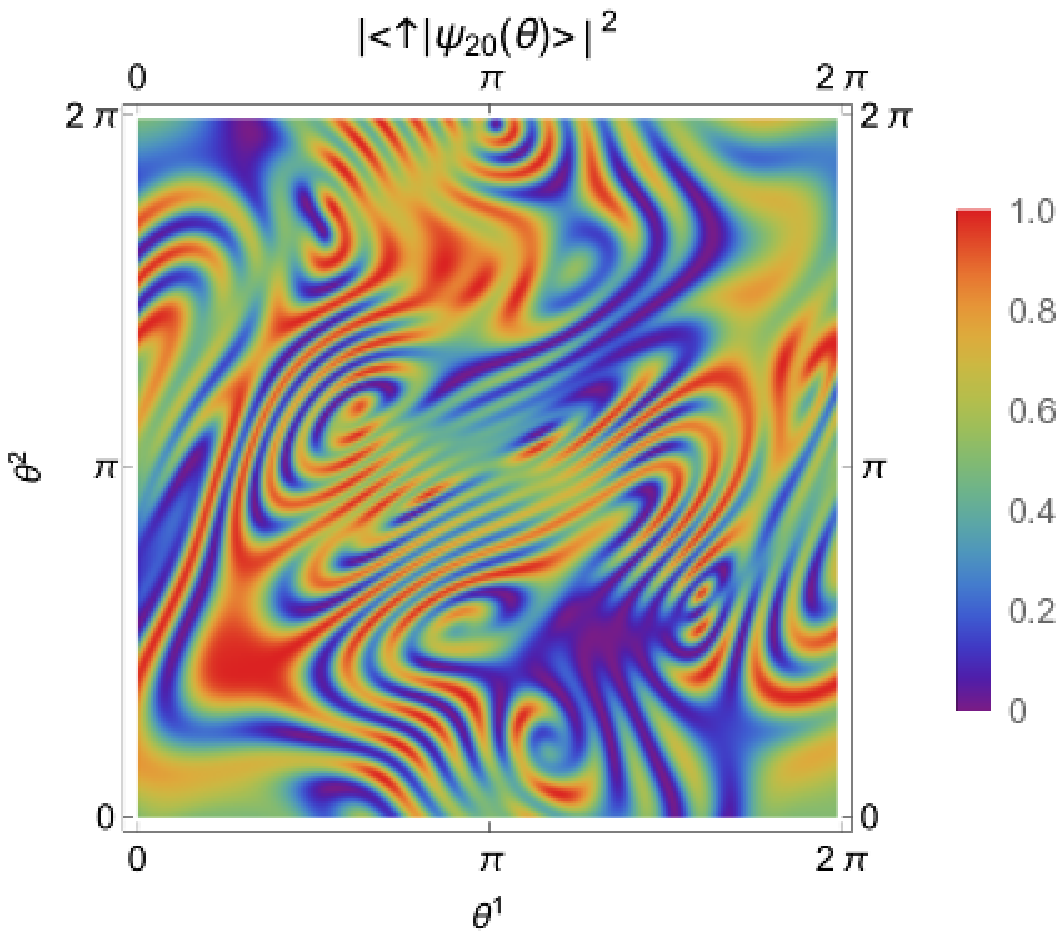} \includegraphics[width=5.5cm]{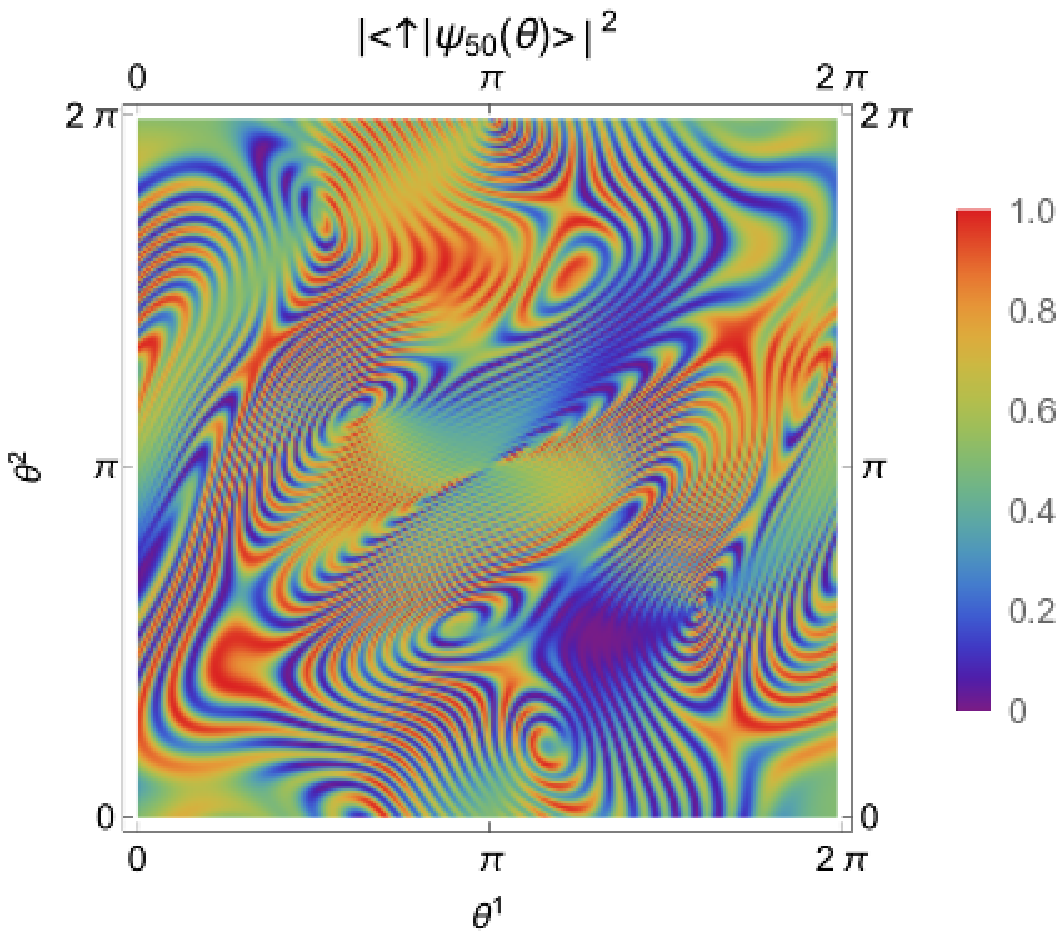}\\
    \includegraphics[width=5.5cm]{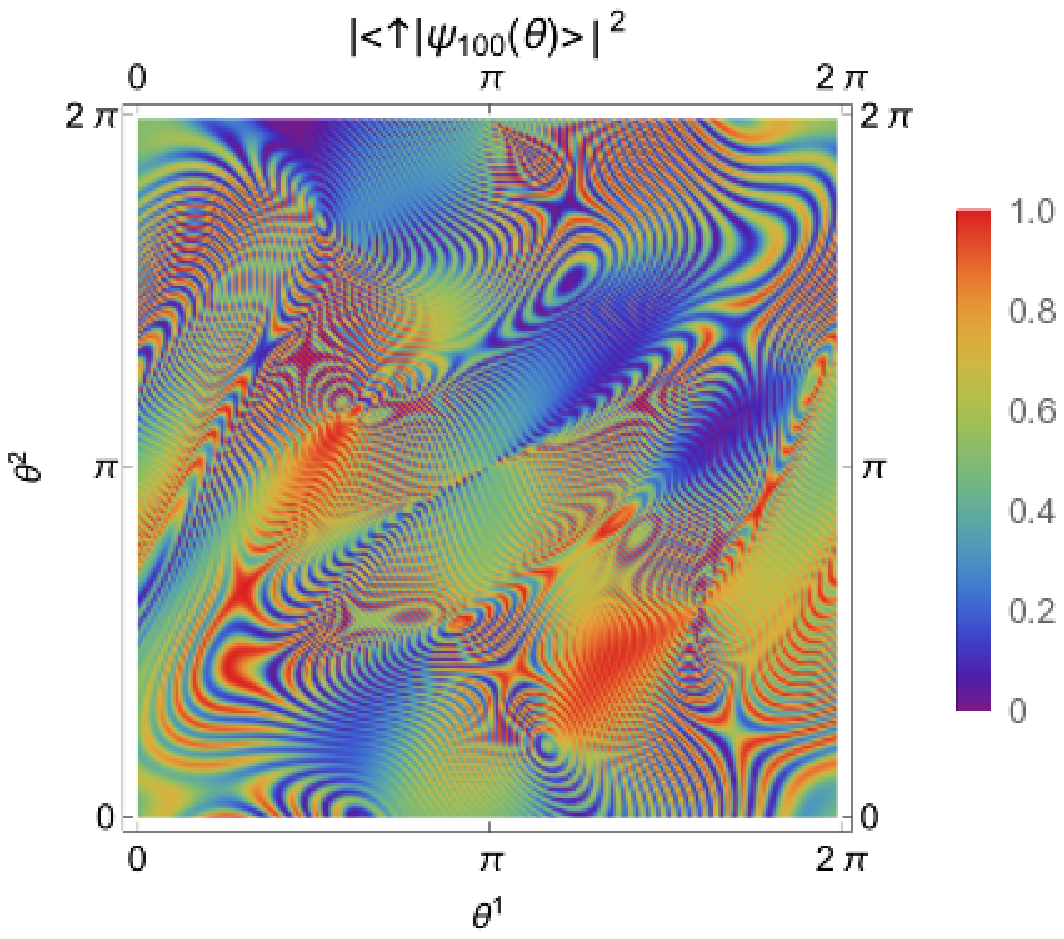} \includegraphics[width=5.5cm]{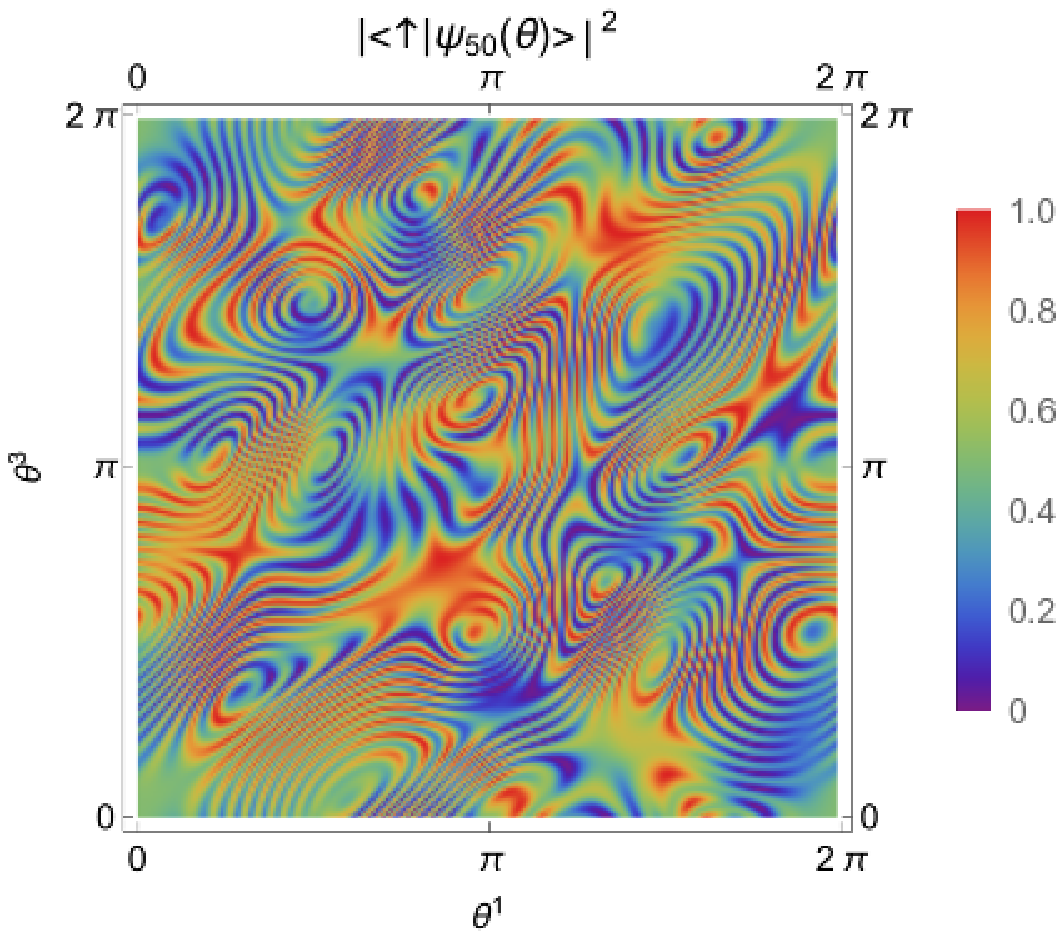}
    \caption{\label{psiCAT} Occupation probability of the state $|\uparrow\rangle$ with respect to $\theta$ for the state $\psi_n$ after $n$ kicks by the cyclic CAT map for the uniform initial condition, with $\frac{\omega_1}{\omega_0} = 2.5$ and $\theta^3 = \frac{\pi}{4}$ (top and bottom left) or $\theta^2=0$ (bottom right).}
  \end{center}
\end{figure}
It is interesting to note that we recover the structures of the SK modes (fig. \ref{SKmodeCAT}).

\subsection{The Arnold's CAT map}
\subsubsection{Quasienergy states :}
Since the Arnold's CAT map is ergodic on the whole of $\mathbb T^2$, the fundamental quasienergies at the fixed point $0$ are stable on the whole of $\mathbb T^2$. We have then
\begin{equation}
  \Sp_{fqe} = \{\tilde \chi_\uparrow, \tilde \chi_\downarrow\}
\end{equation}
with $\tilde \chi_\uparrow = 0$ and $\tilde \chi_\downarrow = \frac{\omega_1}{\omega_0} \pi$. The quasienergy states are obtained by
\begin{equation}
  |Z\mu_\uparrow,\theta \rangle = V_\uparrow(\theta)^{-1} V_\uparrow(\epsilon) |Z\mu_\uparrow,\epsilon \rangle
\end{equation}
with $V_\uparrow(\theta) = \lim_{N \to +\infty} \frac{1}{N} \sum_{n=0}^{N-1} e^{-\imath n \tilde \chi_\uparrow} U(\varphi^{n-1}(\theta))... U(\theta)$ (in practice for $N$ large), and $|Z\mu_\uparrow,\epsilon \rangle$ ($\epsilon$ small) computed from $|Z\mu_{\uparrow/\downarrow},0\rangle$ (eigenvectors of $U(0)$) by using the local expansion formula (see \ref{appendixB1}). A fundamental quasienergy state is represented figure \ref{quasistateArnold}.
\begin{figure}
  \begin{center}
    \includegraphics[width=5.5cm]{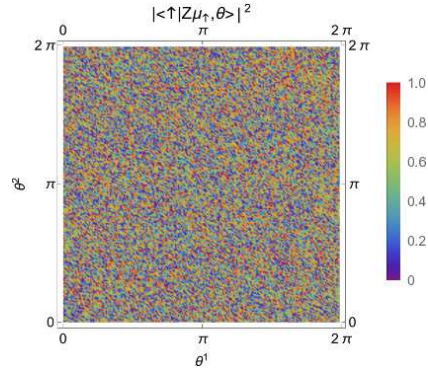}
    \caption{\label{quasistateArnold} Occupation probability of the state $|\uparrow\rangle$ with respect to $\theta$ for the fundamental quasienergy state $|Z\mu_{\uparrow},\theta \rangle$ of the Arnold's CAT map, with $\frac{\omega_1}{\omega_0} = 2.5$ and $\theta^3 = \frac{\pi}{4}$}
  \end{center}
\end{figure}
We see that the fundamental quasienergy state is totally ``uncoherent'' with respect to $\theta$, in accordance with the chaotic behavior of the Arnold's CAT map. This figure recalls the noisy aspect of the orbits of the flow (fig. \ref{classicalflow}).

\subsubsection{Dynamics :}
As for the previous example, we consider three initial conditions:
\begin{itemize}
\item $\psi_0(\theta) = \frac{\mathbb{I}_{D_0}(\theta)}{\mu(D_0)} \frac{1}{\sqrt 2}(|\uparrow \rangle + |\downarrow \rangle)$ corresponding to a highly coherent ensemble of spins with a small uniform dispersion of the first kicks.
\item $\psi_0(\theta) =  \frac{1}{\sqrt 2}(|\uparrow \rangle + |\downarrow \rangle)$ which corresponds to a large uniform dispersion (on the whole of $\mathbb T^2$) of the first kicks.
\item $\psi_0(\theta) = |Z\mu_{\uparrow},\theta \rangle$ a fundamental quasienergy state.
\end{itemize}
The dynamics are represented figure \ref{dynamicsArnold}.
\begin{figure}
  \begin{center}
    \includegraphics[width=4.2cm]{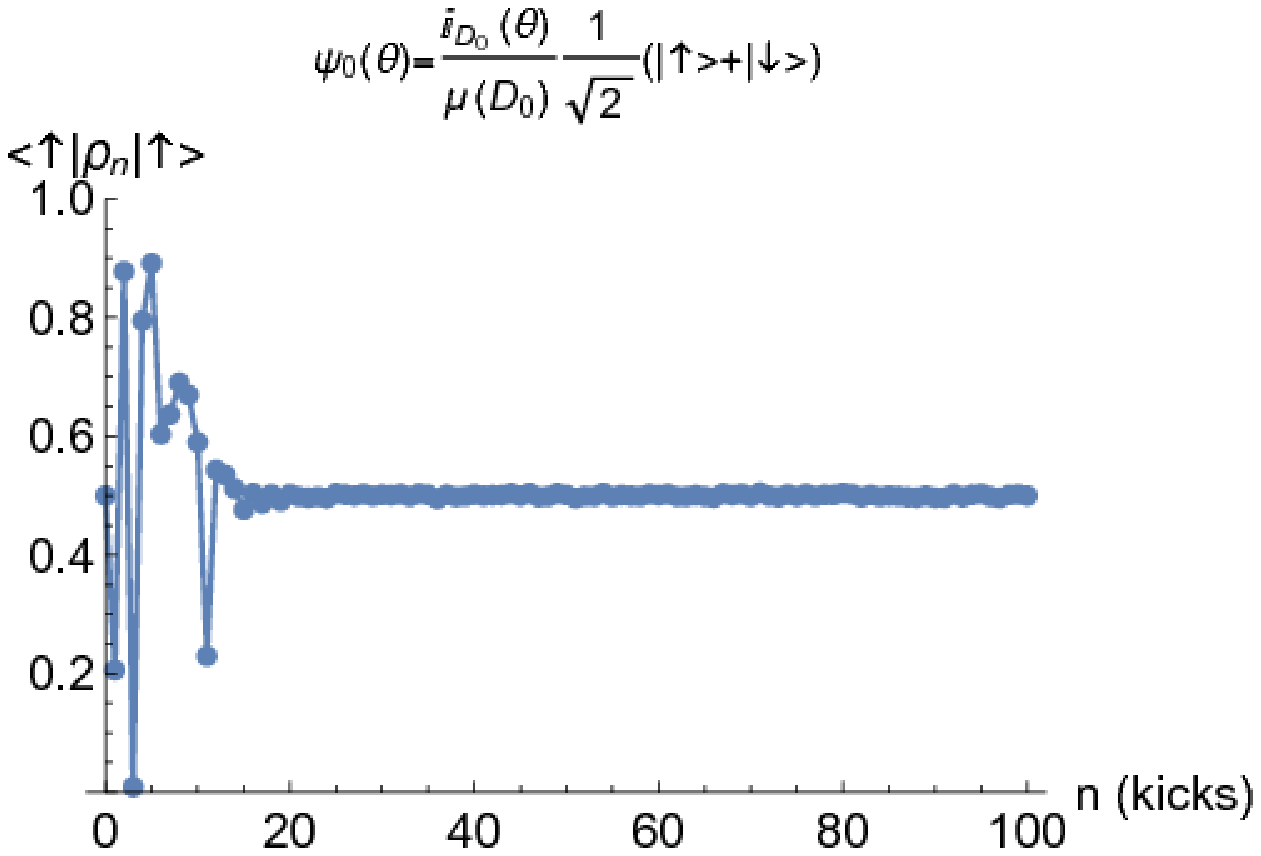} \includegraphics[width=4.2cm]{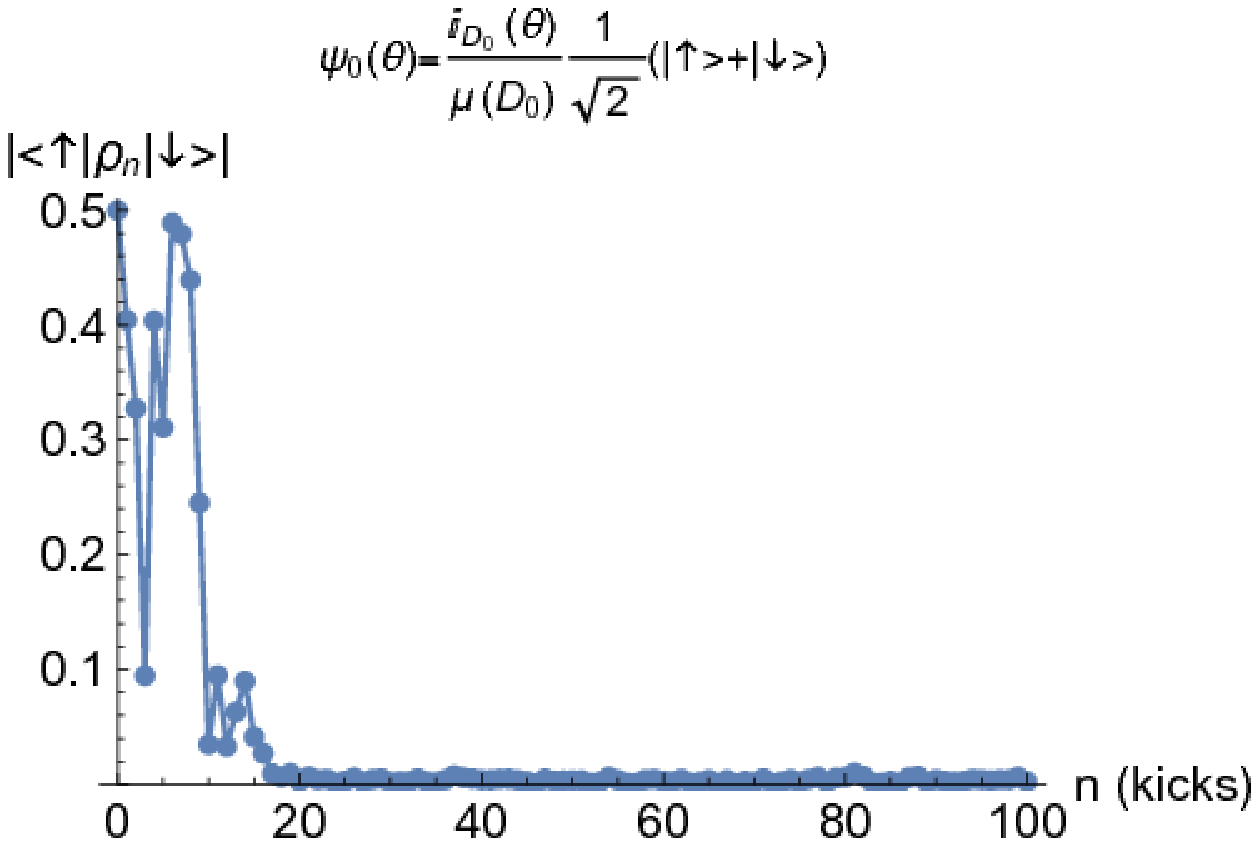} \includegraphics[width=4.2cm]{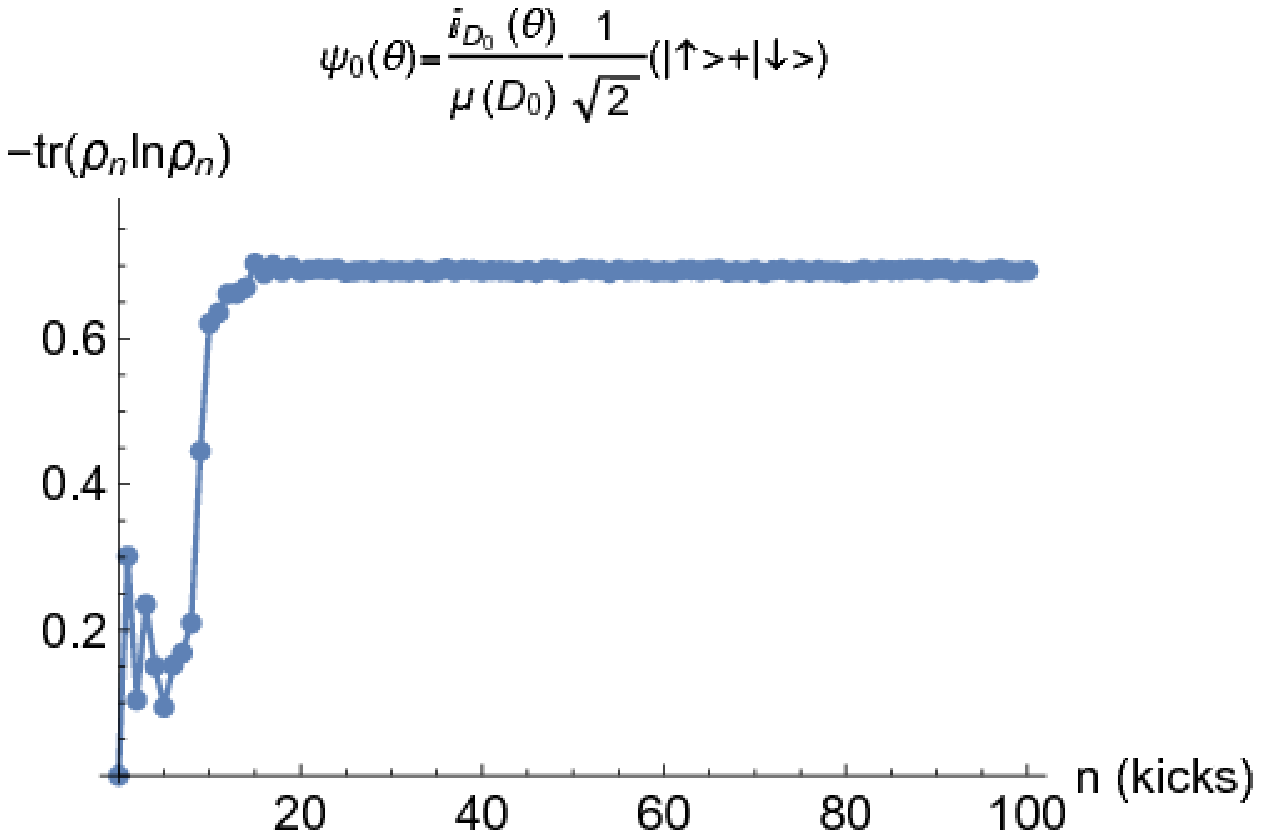}\\
    \includegraphics[width=4.2cm]{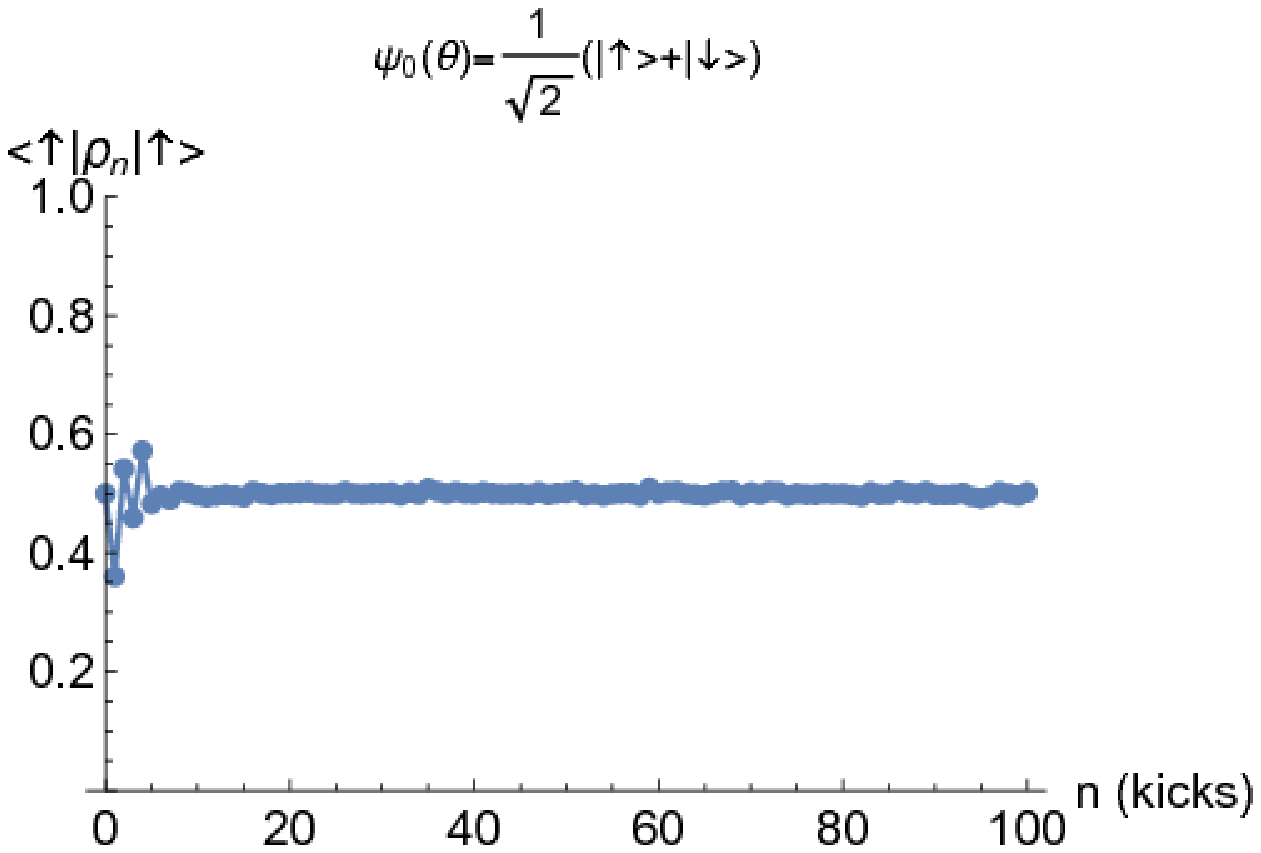} \includegraphics[width=4.2cm]{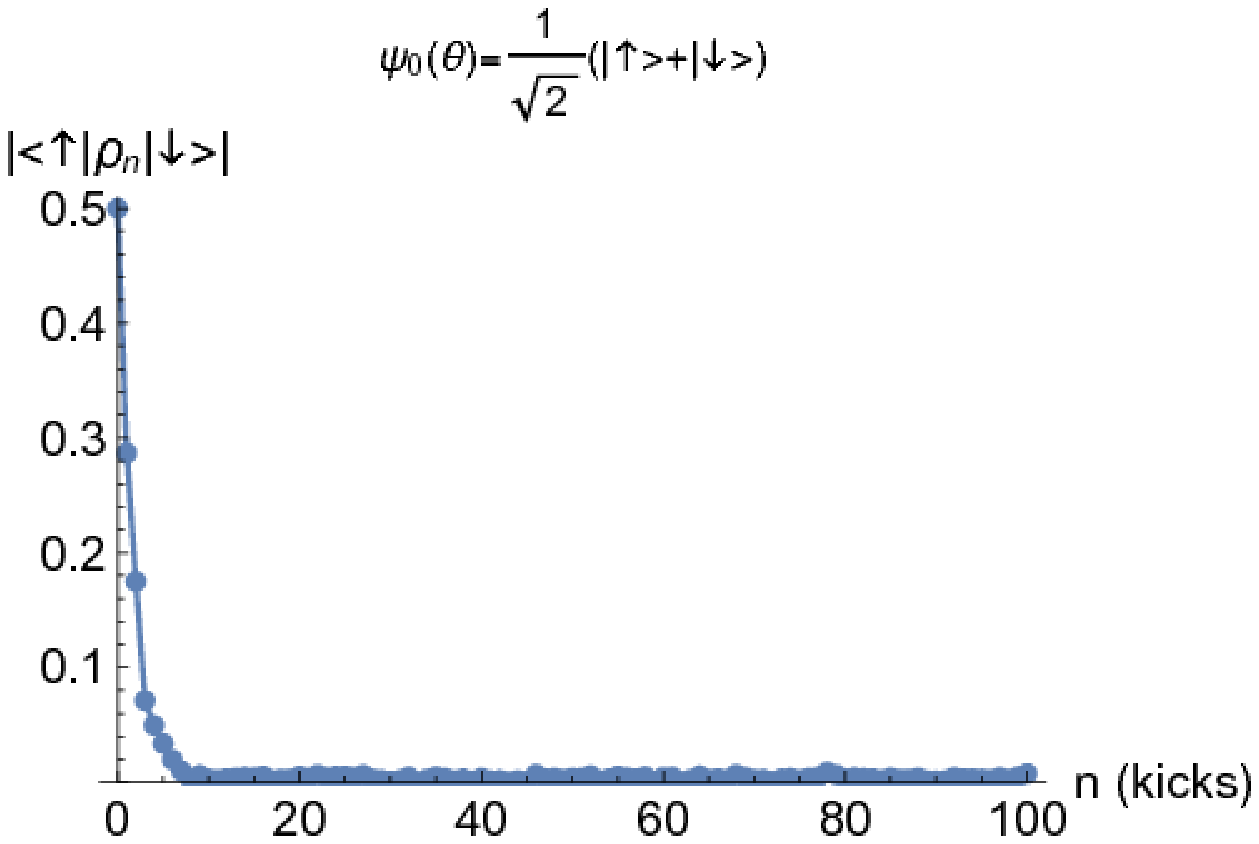} \includegraphics[width=4.2cm]{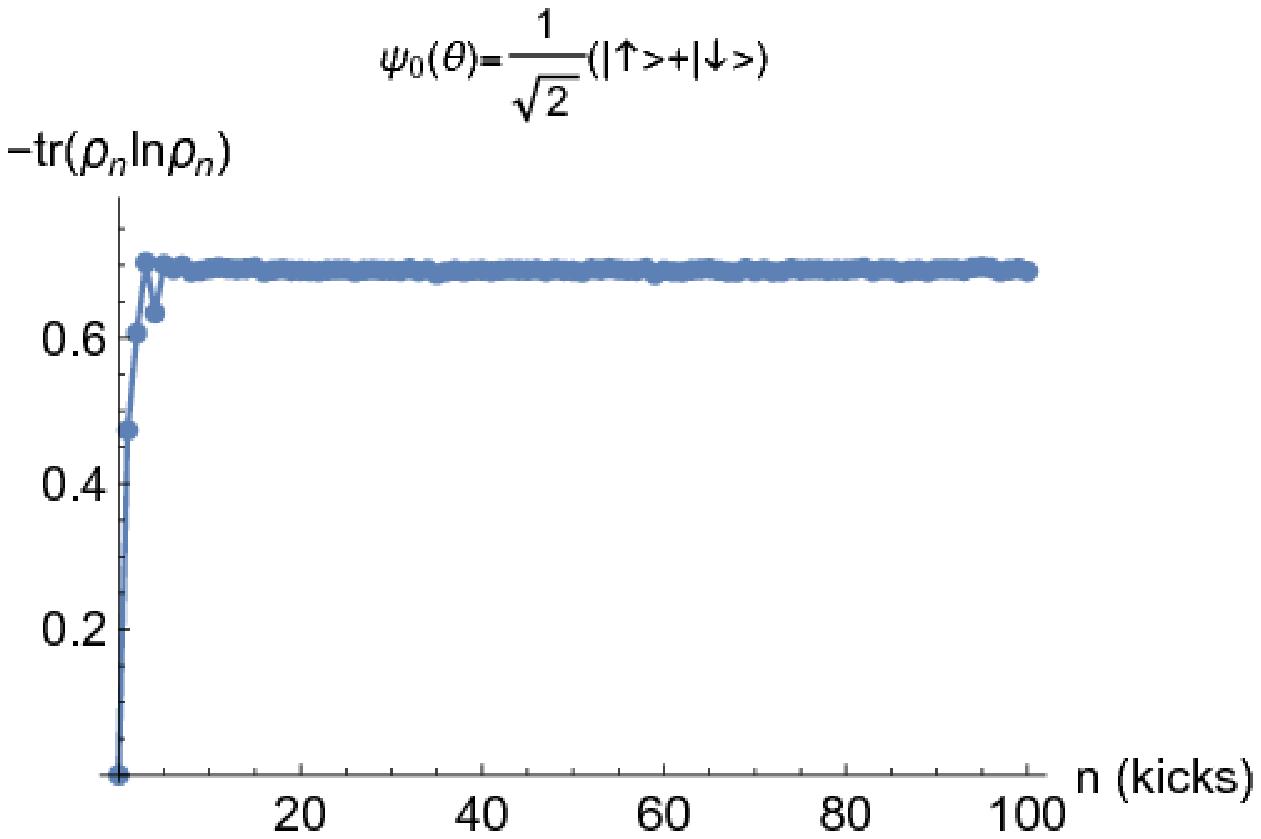}\\
    \includegraphics[width=4.2cm]{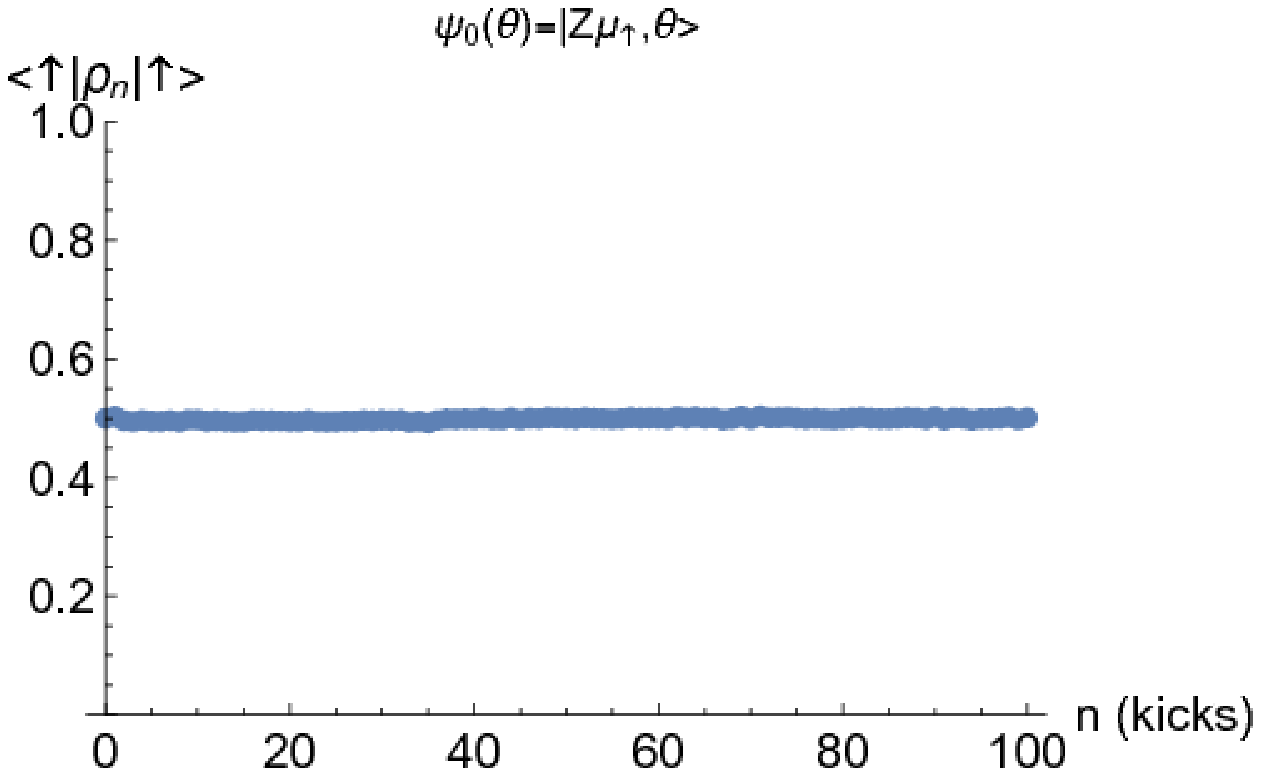} \includegraphics[width=4.2cm]{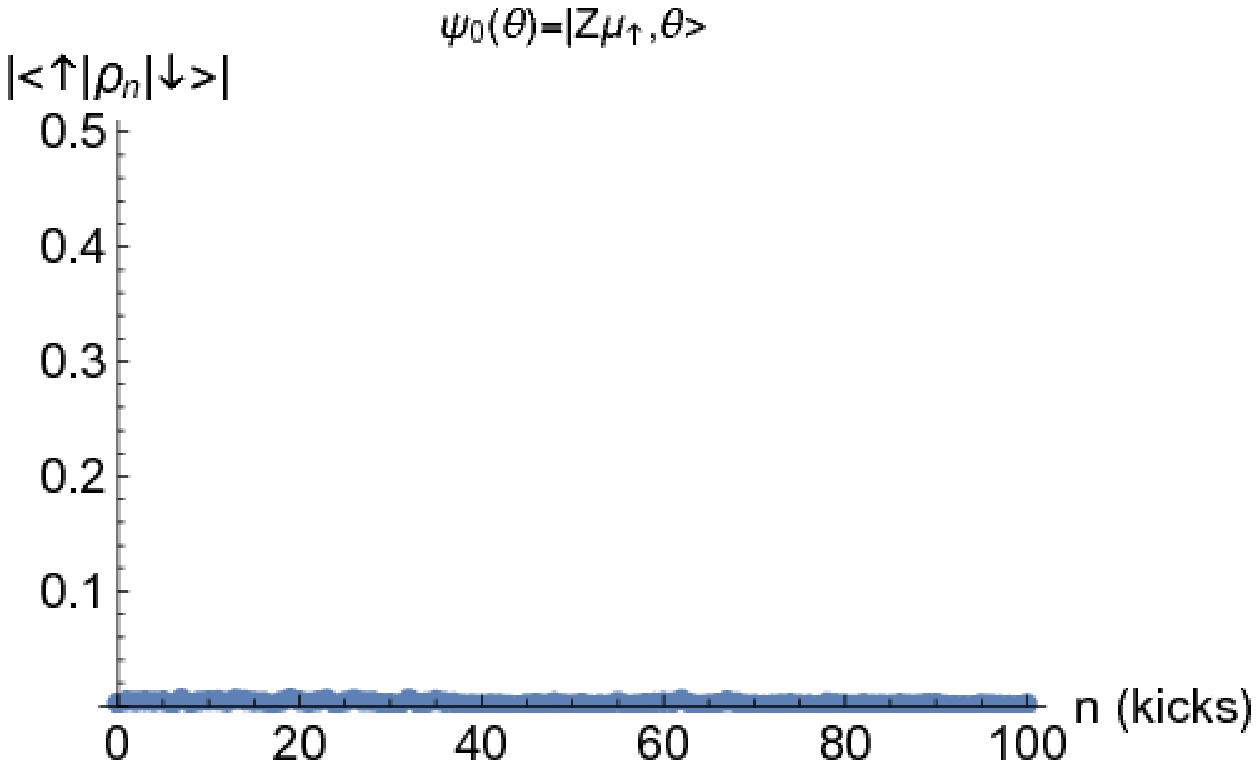} \includegraphics[width=4.2cm]{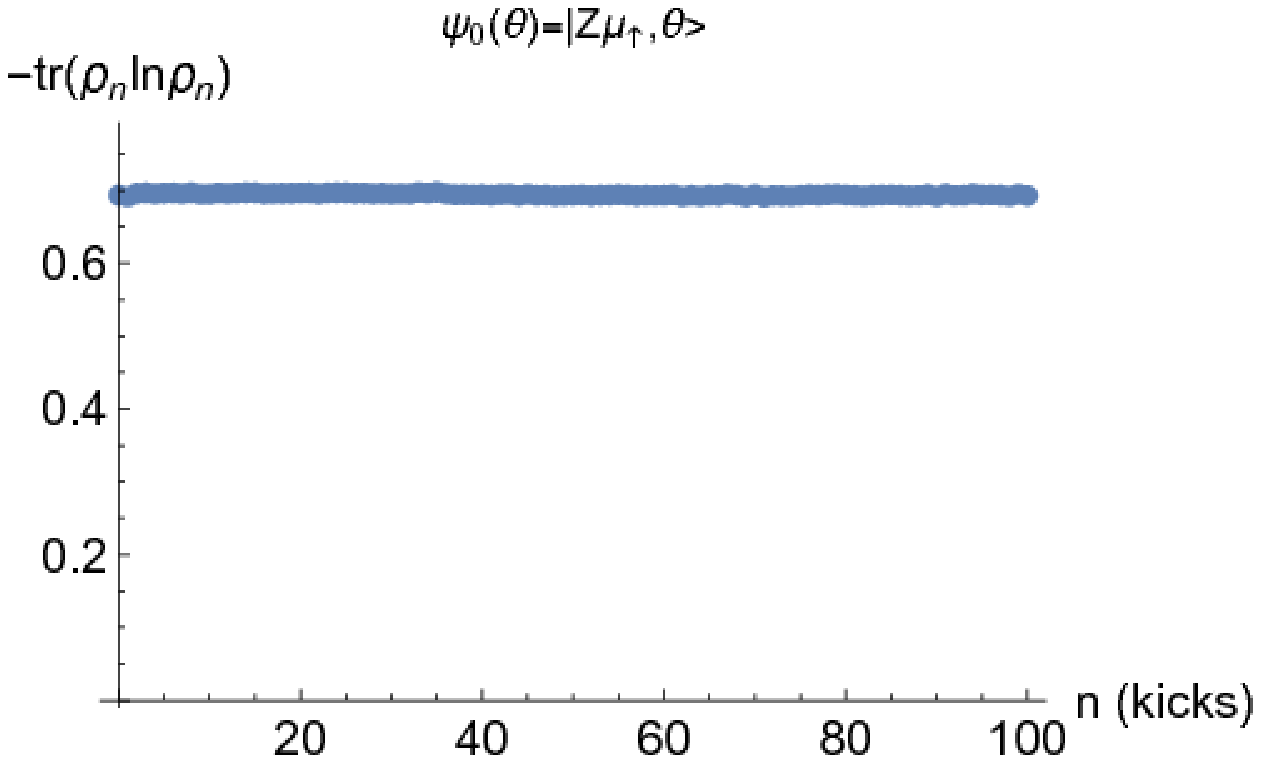}\\
    \caption{\label{dynamicsArnold} Population of the state $|\uparrow \rangle$ (left), coherence (center) and von Neumann entropy (right) of the mixed state for the stroboscopic dynamics of the spin ensemble (with $\frac{\omega_1}{\omega_0} = 2.5$ and $\theta^3 = \frac{\pi}{4}$) where the kick modulation is governed by the Arnold's CAT map. The initial condition is for a small uniform dispersion of the first kicks (first line), a large uniform dispersion (second line) and a fundamental quasienergy state (last line).}
  \end{center}
\end{figure}
We have a decoherence phenomenon for large but also for small initial dispersions of the kicks, in accordance with the chaotic behaviour (and with the sensibility to initial conditions of the flow, for a detailed discussion see \cite{Viennot2,Aubourg}). As expected, the fundamental quasienergy state is a steady state of the quantum system for which the reduced density matrix is the microcanonical density matrix $\left(\small \begin{array}{cc} 1/2 & 0 \\ 0 & 1/2 \end{array} \normalsize \right)$.\\
The representation of the final state for all initial condition is completely similar to figure \ref{quasistateArnold}.

\subsection{The standard map}
\subsubsection{Quasienergy states and SK modes :}
The standard map presents both the behaviours of the two previous examples. We compute the state:
\begin{equation}
  |Z\mu_\uparrow,\theta \rangle = V_\uparrow(\theta)^{-1} V_\uparrow(\epsilon) |Z\mu_\uparrow,\epsilon \rangle
\end{equation}
where $V_\uparrow(\theta) = \frac{1}{N} \sum_{n=0}^{N-1} e^{-\imath n \tilde \chi_\uparrow} U(\varphi^{n-1}(\theta))... U(\theta)$ (with large $N$) and with $\tilde \chi_\uparrow$ and $|Z\mu_\uparrow,0 \rangle$ computed at the fixed point $0$. In the chaotic sea, this state is a fundamental quasienergy state, whereas it is just a SK mode in the islands of stability. It is represented figure \ref{quasistateSTD}.
\begin{figure}
  \begin{center}
    \includegraphics[width=5.5cm]{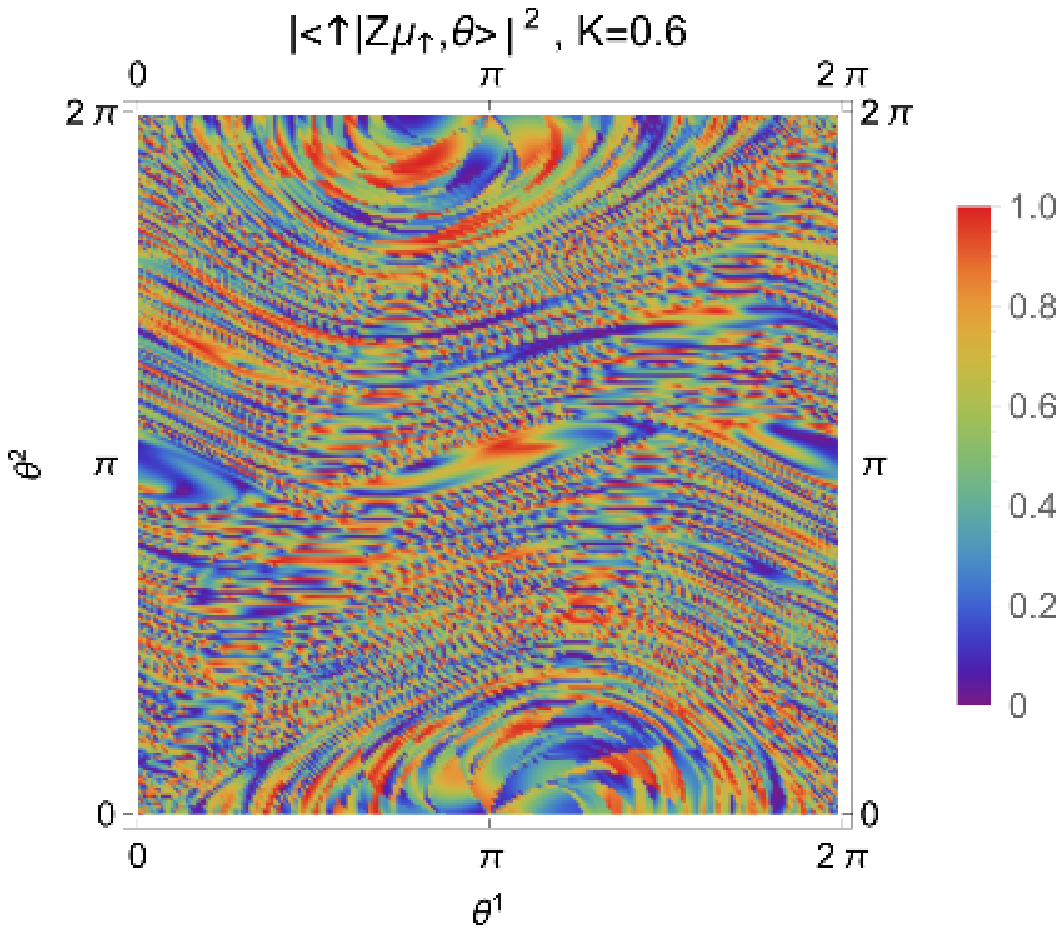} \includegraphics[width=5.5cm]{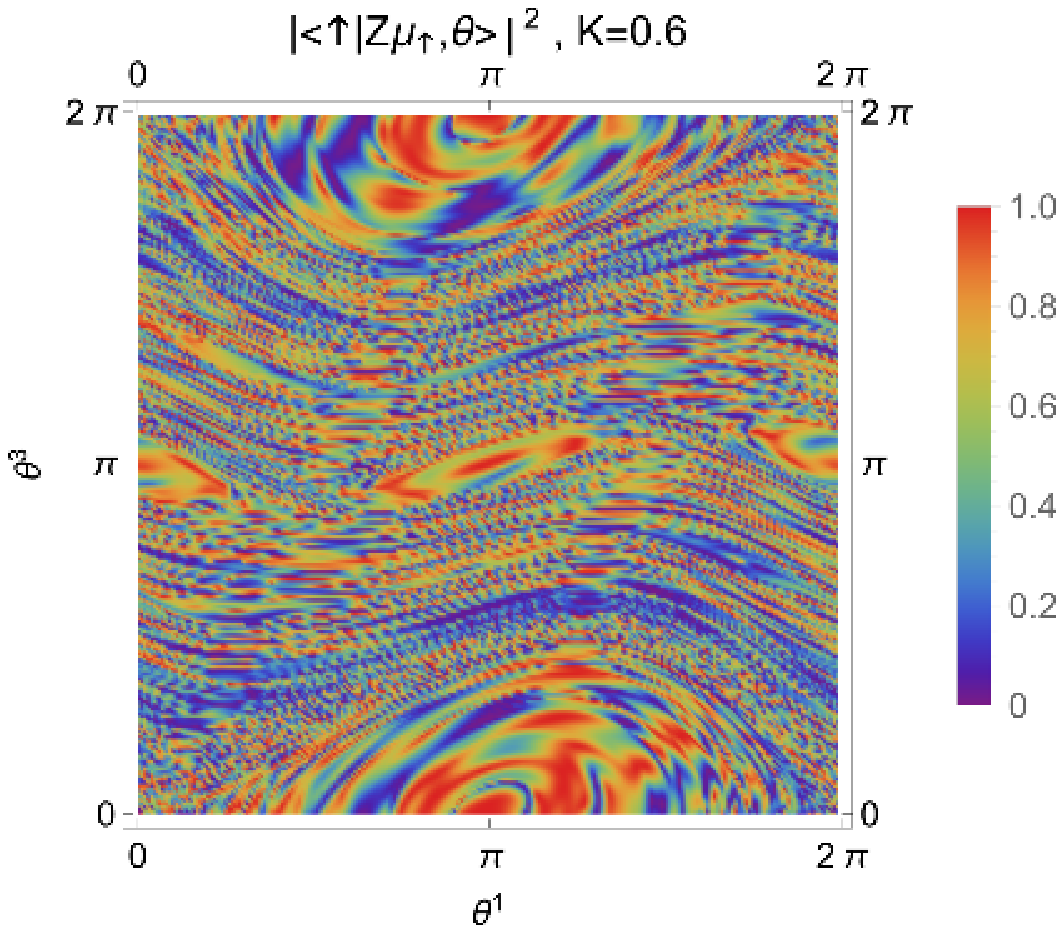}\\
    \includegraphics[width=5.5cm]{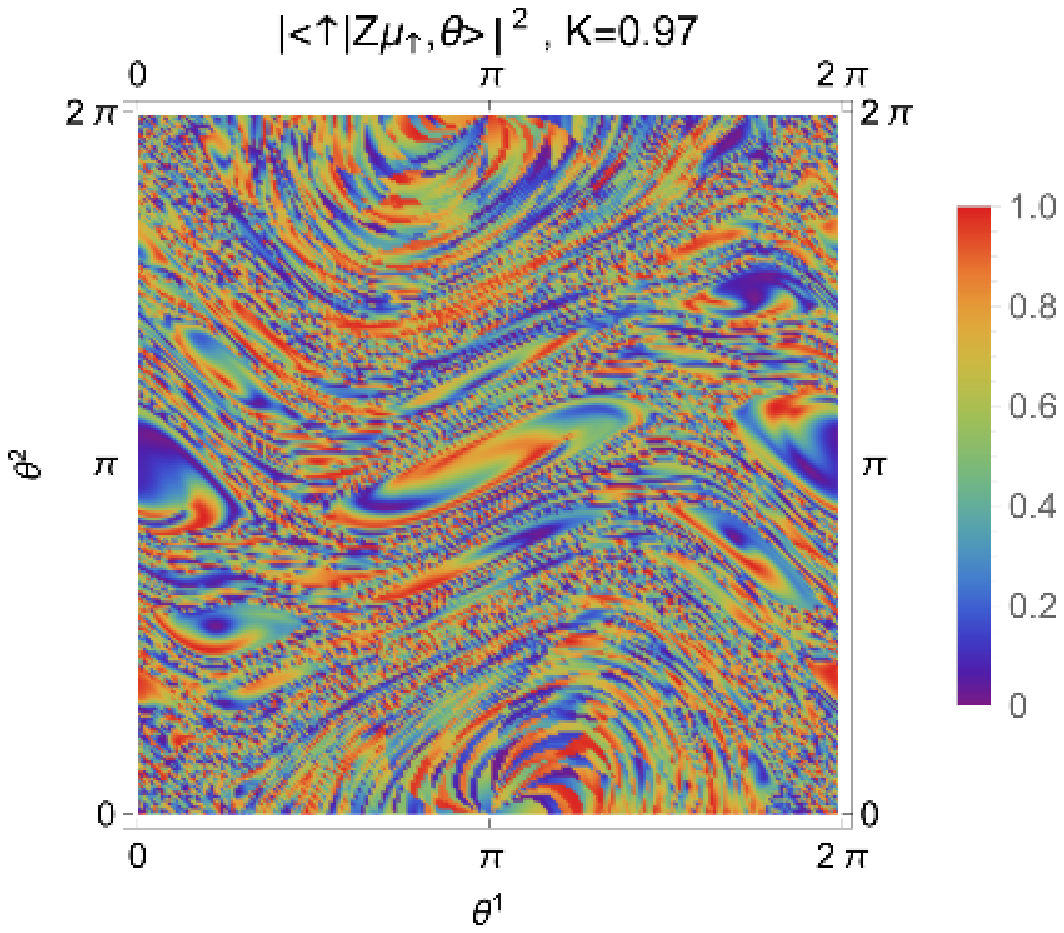} \includegraphics[width=5.5cm]{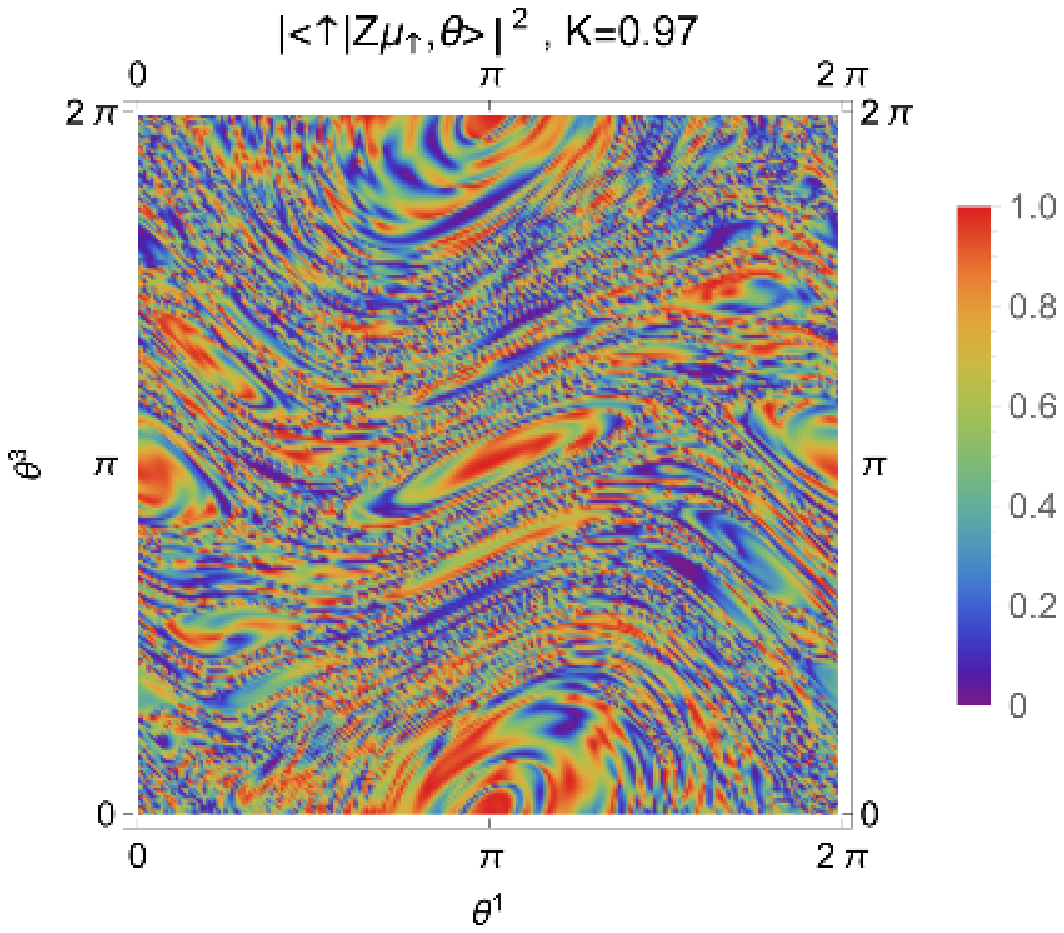}\\
    \includegraphics[width=5.5cm]{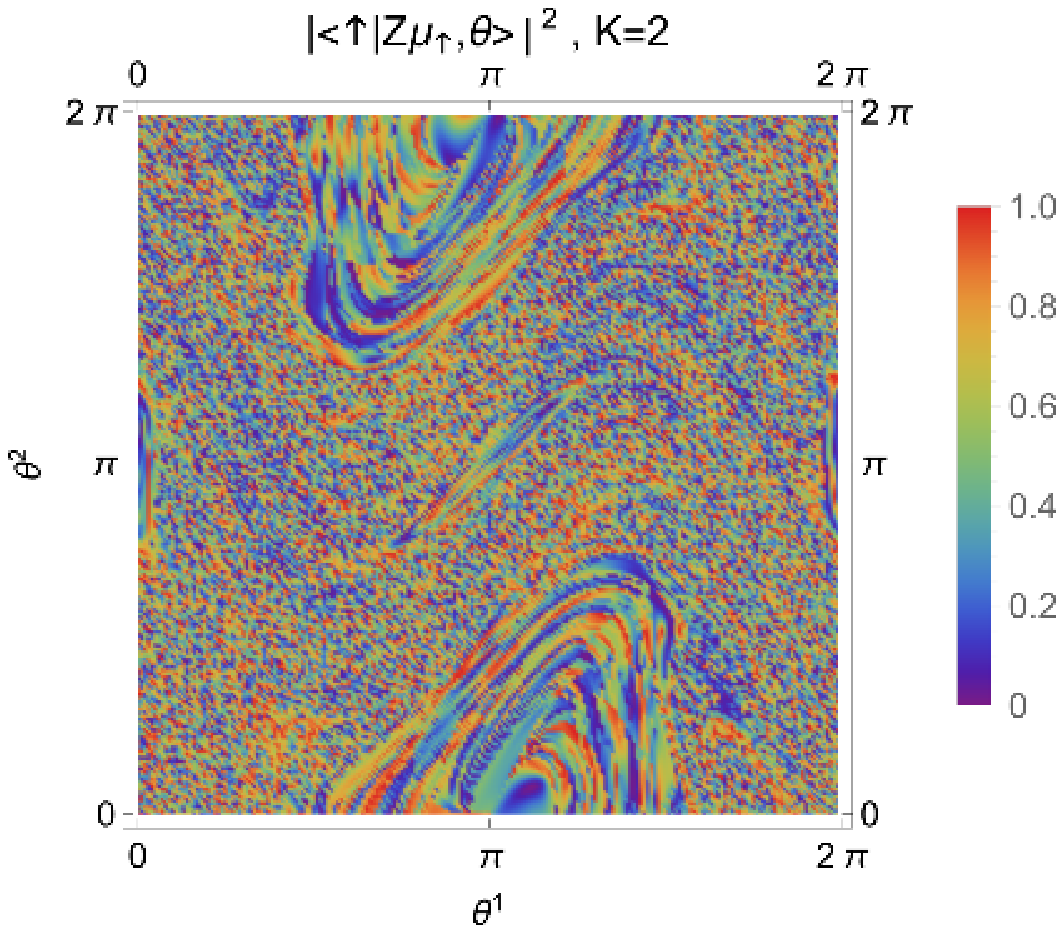} \includegraphics[width=5.5cm]{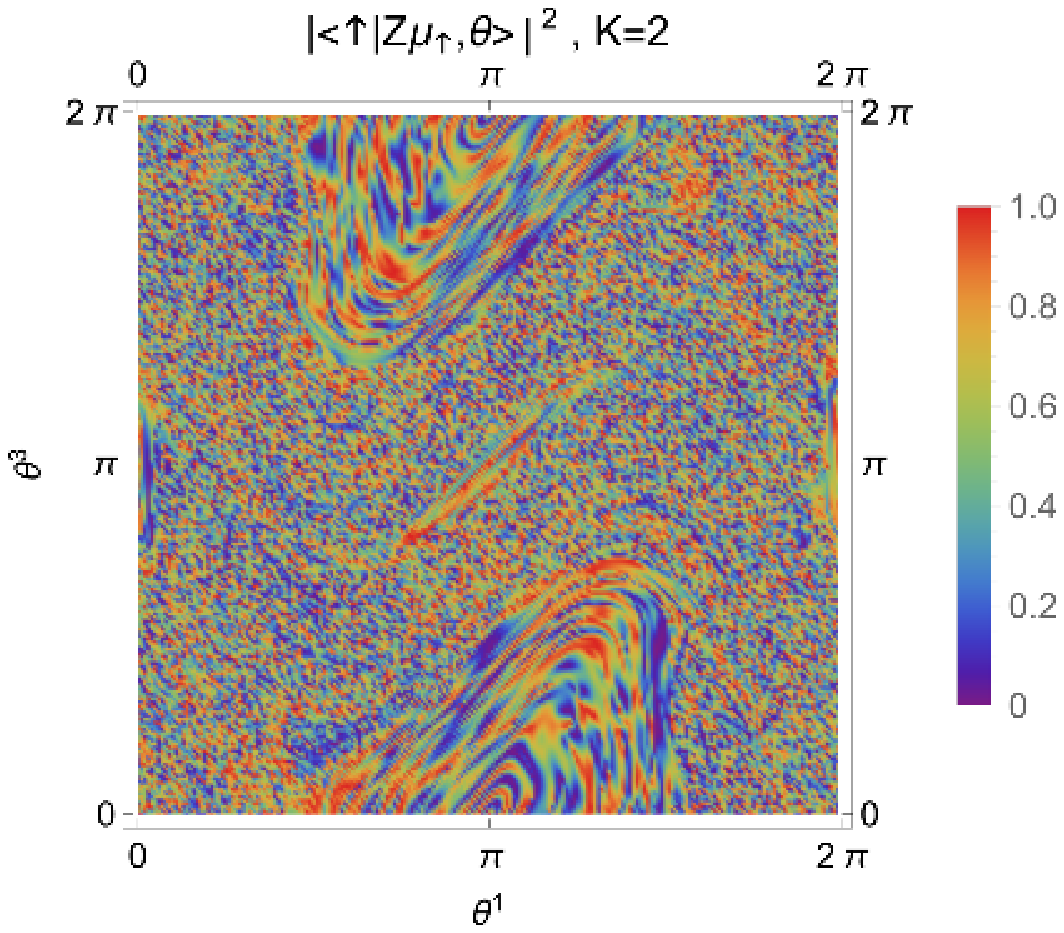}
    \caption{\label{quasistateSTD} Occupation probability of the state $|\uparrow\rangle$ with respect to $\theta$ for the fundamental quasienergy state $|Z\mu_{\uparrow},\theta \rangle$ of the standard map, with $\frac{\omega_1}{\omega_0} = 2.5$ and $\theta^3 = \frac{\pi}{4}$ (left) or $\theta^2=0$ (right), for $K=0.6$ (top), $K=0.97$ (middle) and $K=2$ (bottom).}
  \end{center}
\end{figure}
We see by comparison with the orbits of the classical flow (figure \ref{classicalflow}), that the structure of the islands of stability embedded into the chaotic sea is clearly apparent in the fundamental quasienergy state of the spin ensemble. The chaotic sea appears as an uncoherent region for the probability distribution associated with the quasienergy state, whereas the islands of stability appear as more coherent regions (with some interference structures as for the cyclic map).

\subsubsection{Dynamics :}
As for the previous examples we consider the following initial conditions:
\begin{itemize}
\item $\psi_0(\theta) = \frac{\mathbb{I}_{D_0}(\theta)}{\mu(D_0)} \frac{1}{\sqrt 2}(|\uparrow \rangle + |\downarrow \rangle)$ corresponding to a highly coherent ensemble of spins with a small uniform dispersion of the first kicks.
\item $\psi_0(\theta) =  \frac{1}{\sqrt 2}(|\uparrow \rangle + |\downarrow \rangle)$ which corresponds to a large uniform dispersion (on the whole of $\mathbb T^2$) of the first kicks.
\item $\psi_0(\theta) = |Z\mu_{\uparrow},\theta \rangle$ a fundamental quasienergy state.
\end{itemize}
The dynamics are represented figures \ref{dynamicsSTD06} and \ref{dynamicsSTD2}.
\begin{figure}
  \begin{center}
    \includegraphics[width=4.2cm]{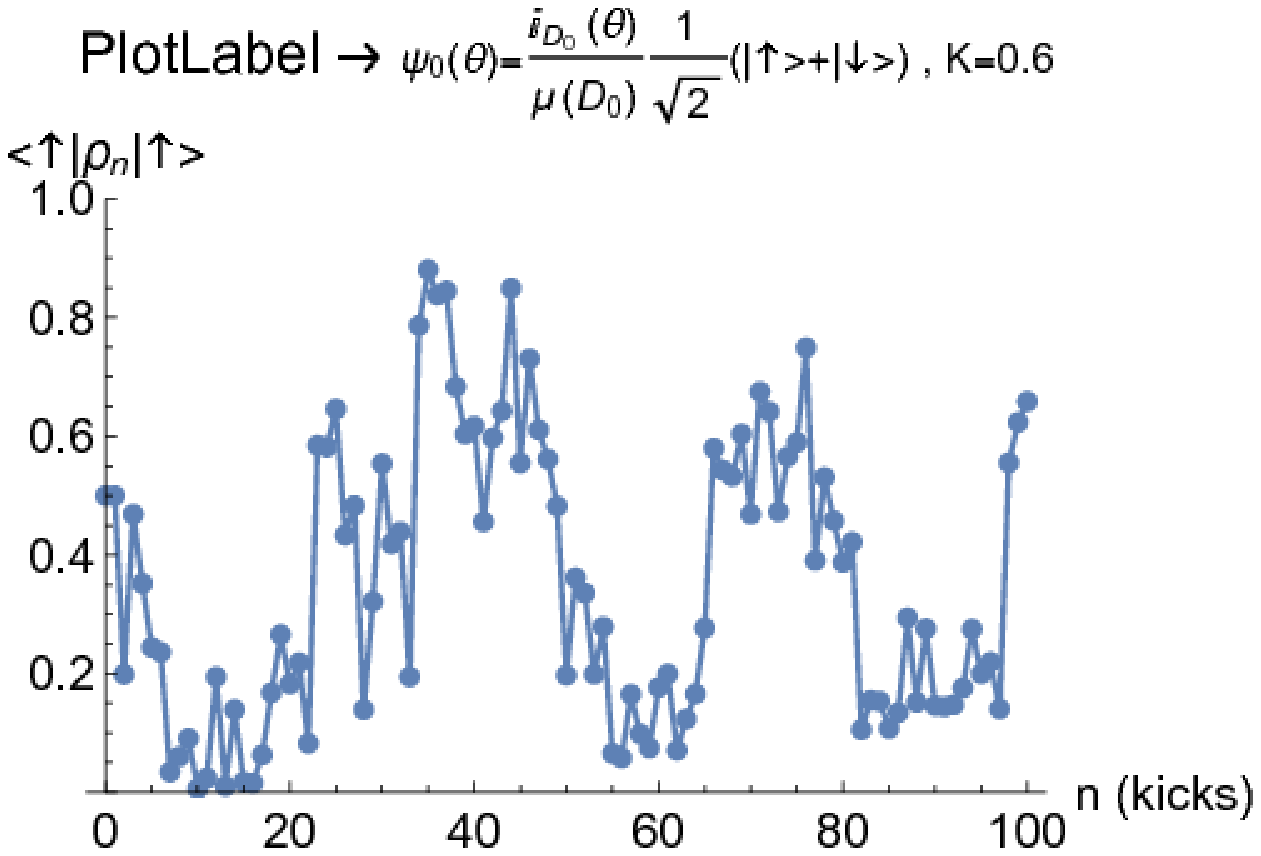} \includegraphics[width=4.2cm]{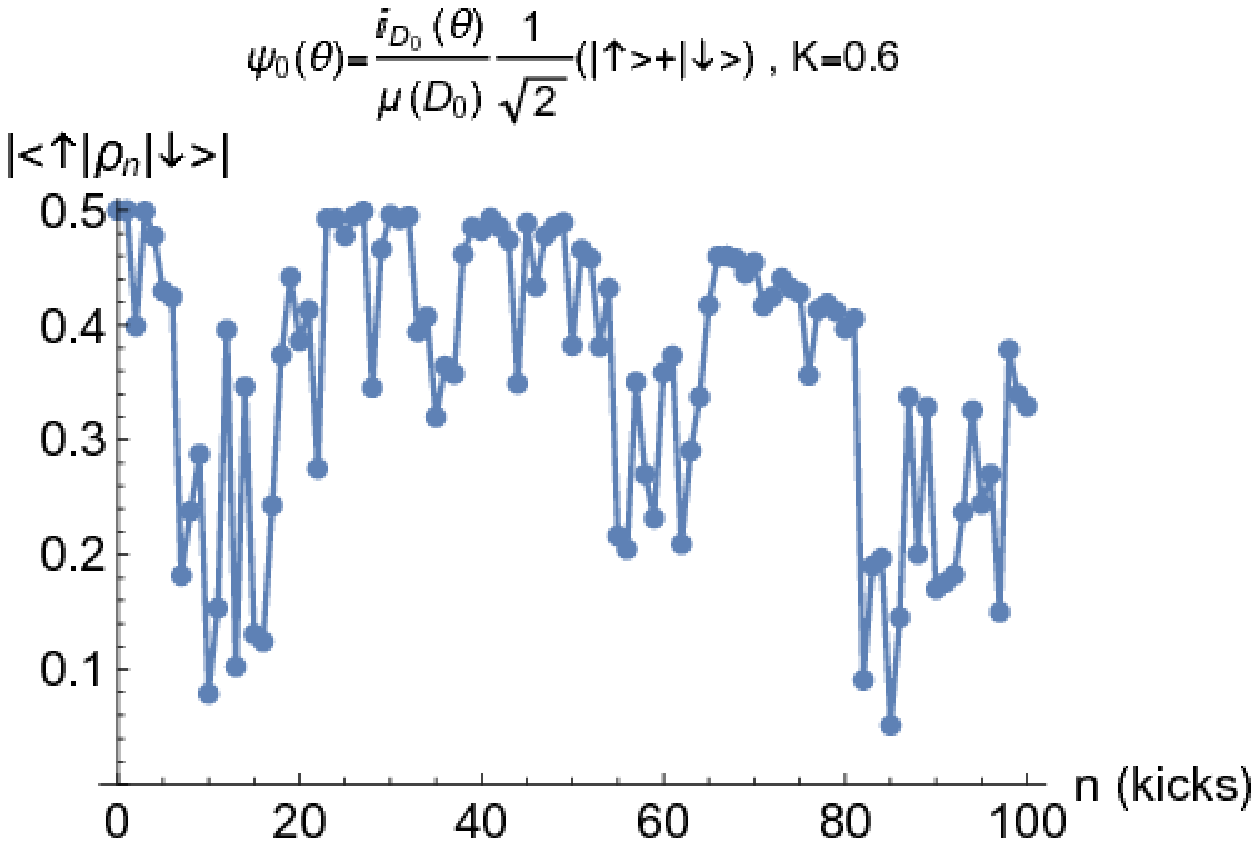} \includegraphics[width=4.2cm]{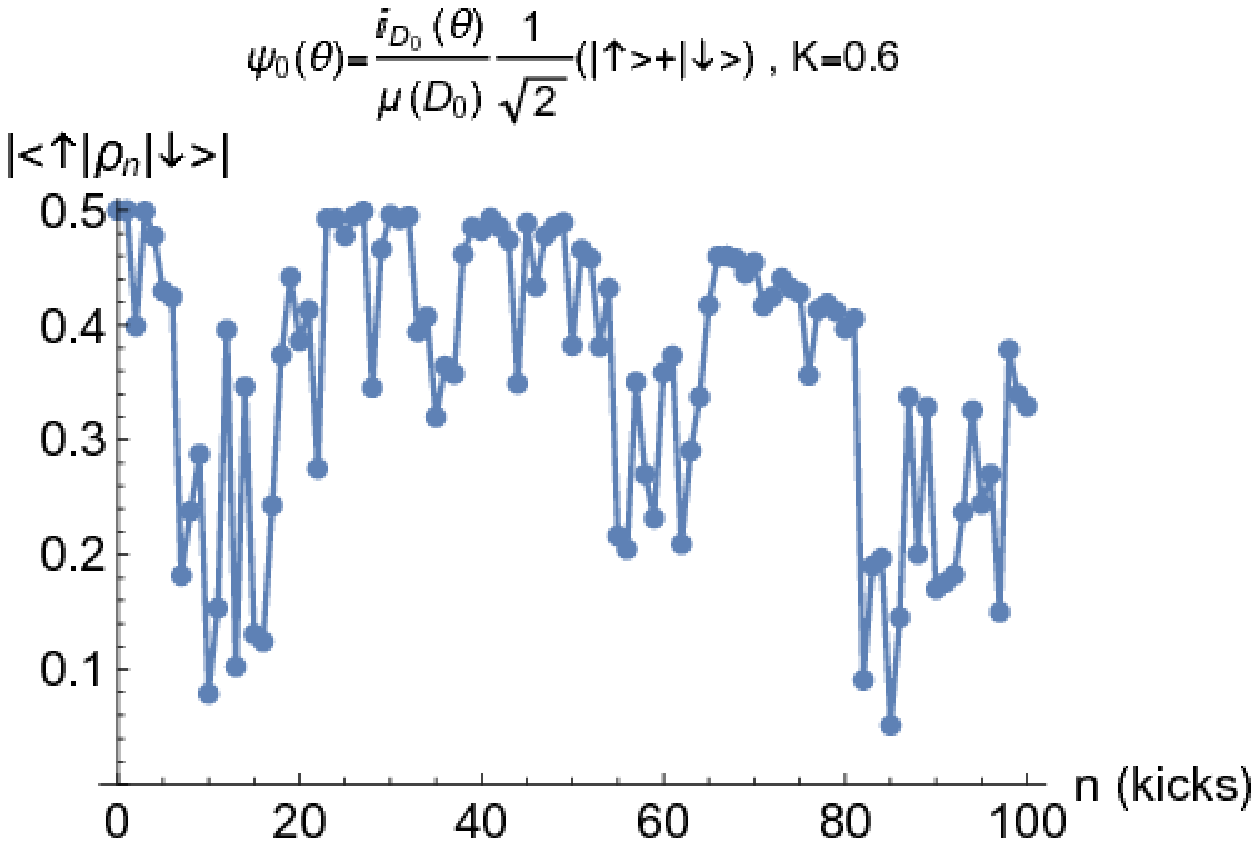}\\
    \includegraphics[width=4.2cm]{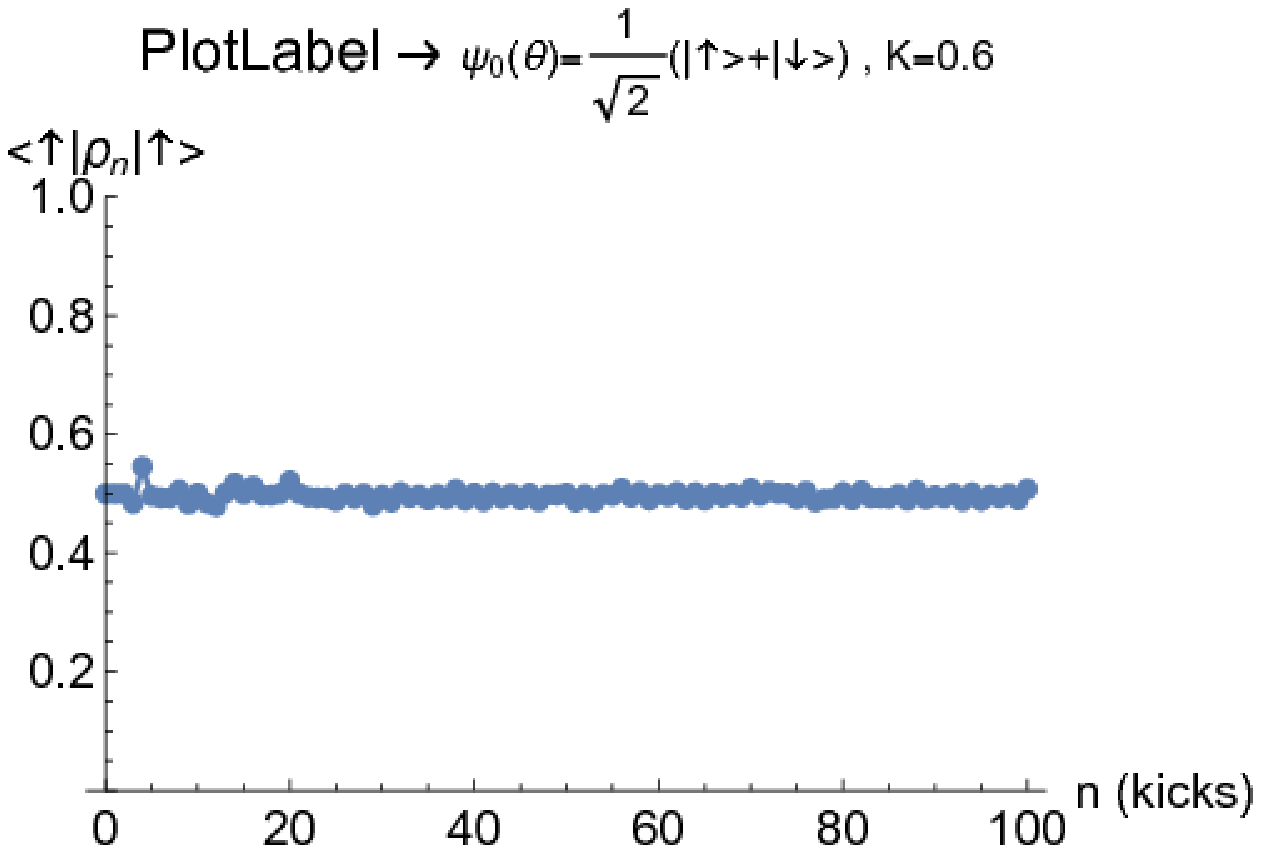} \includegraphics[width=4.2cm]{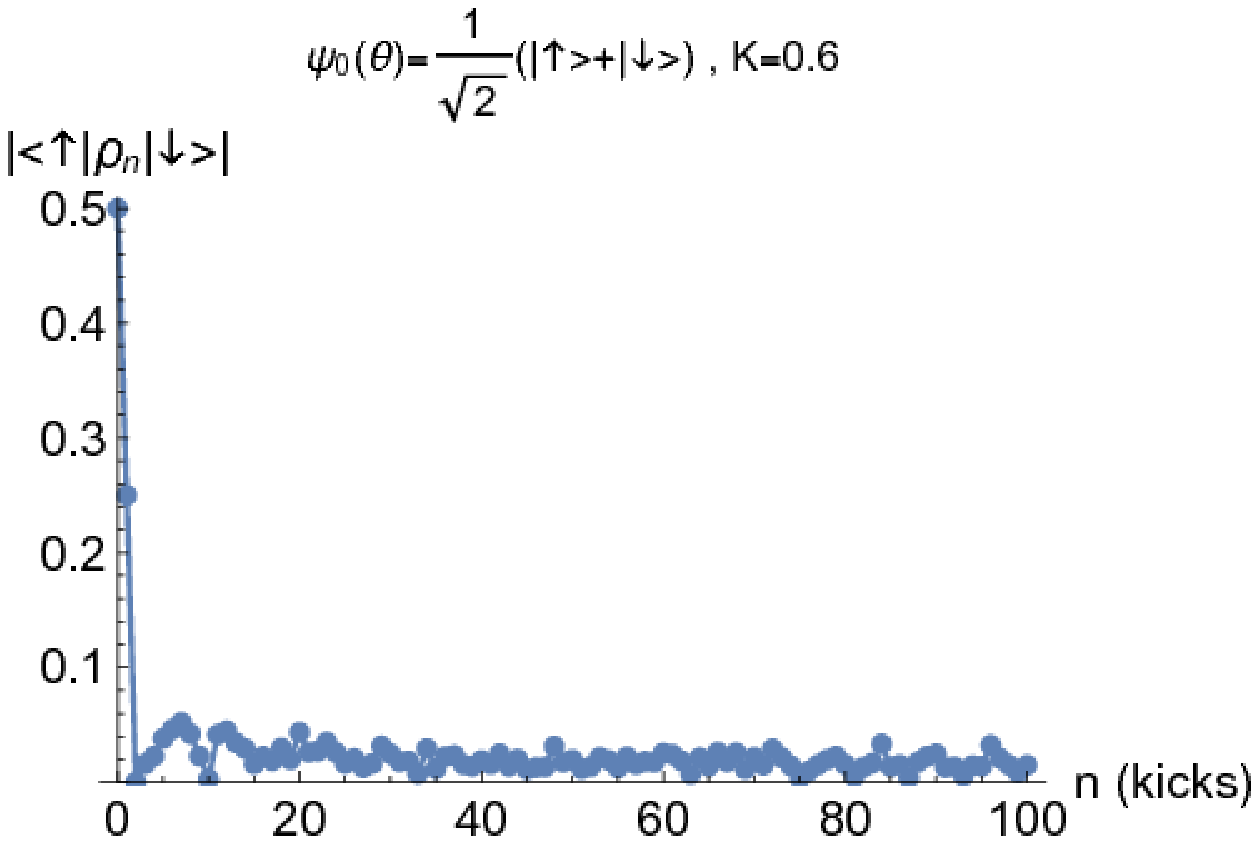} \includegraphics[width=4.2cm]{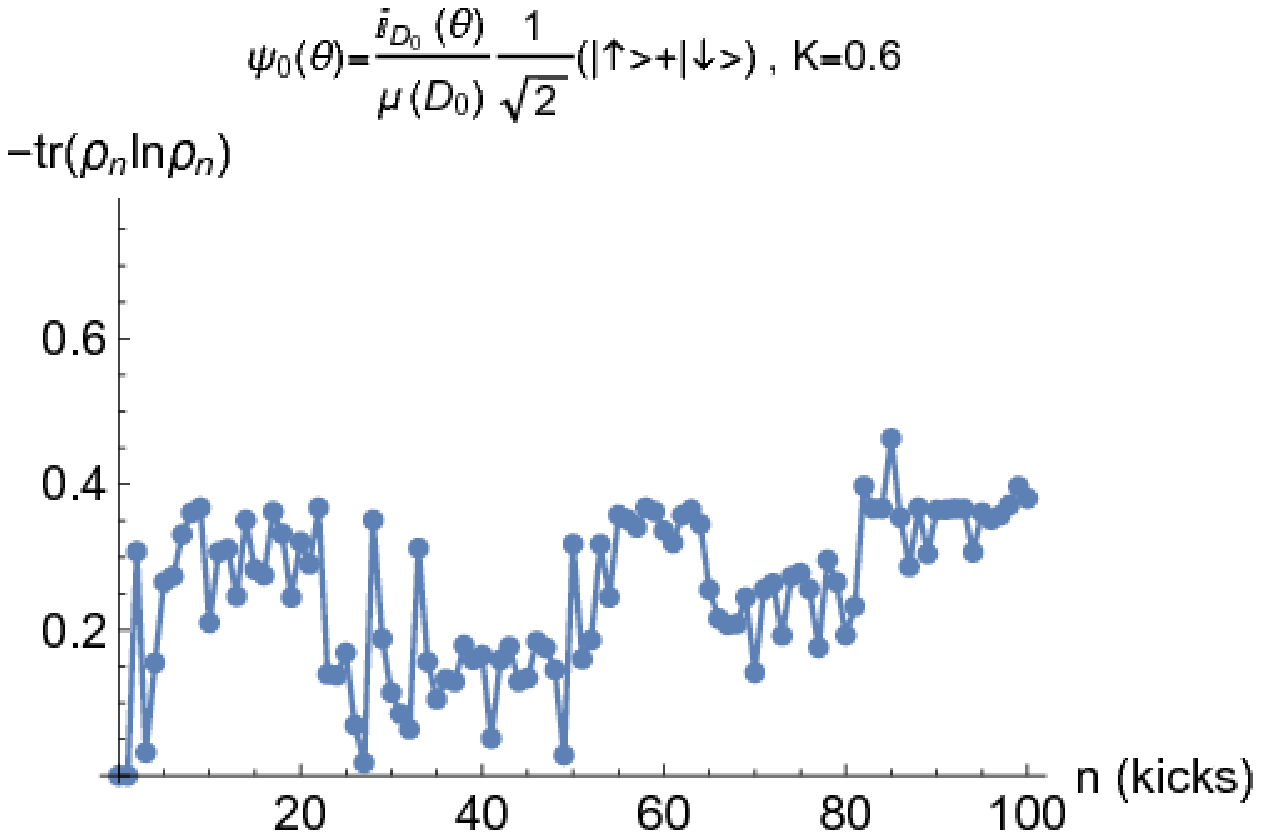}\\
    \includegraphics[width=4.2cm]{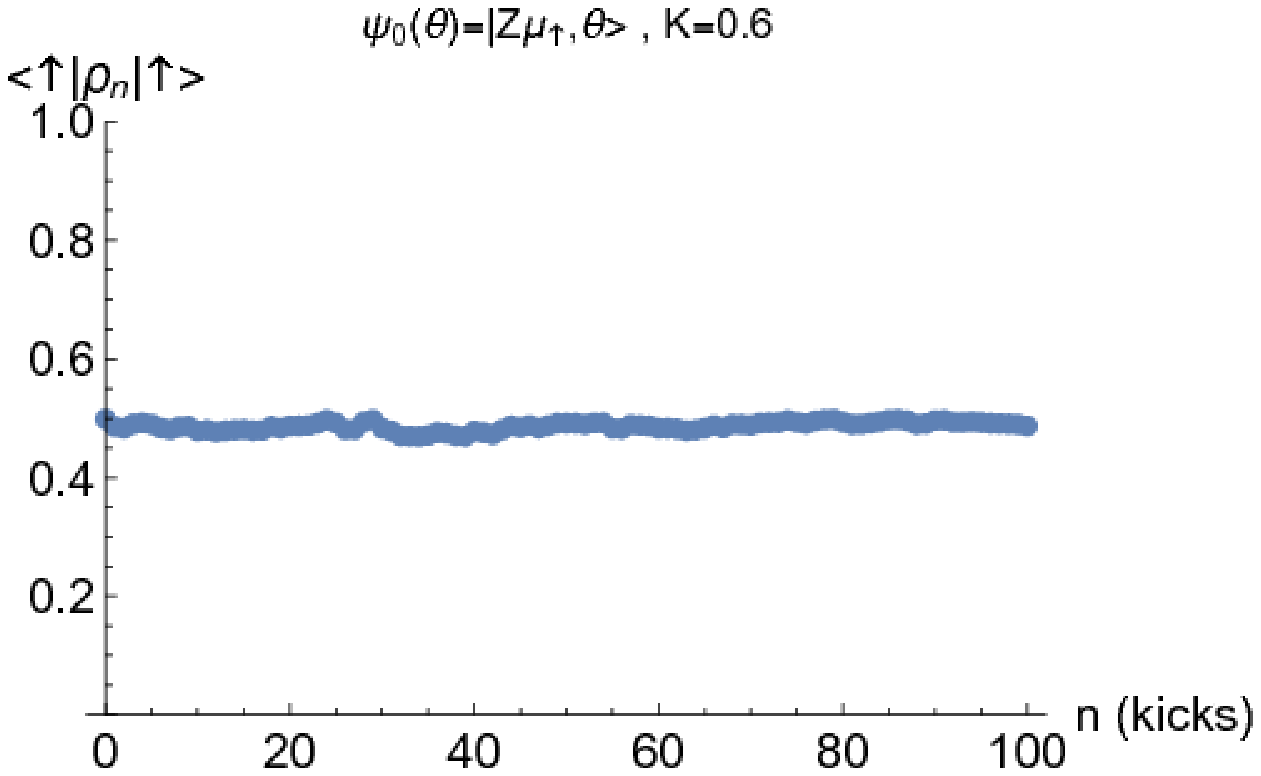} \includegraphics[width=4.2cm]{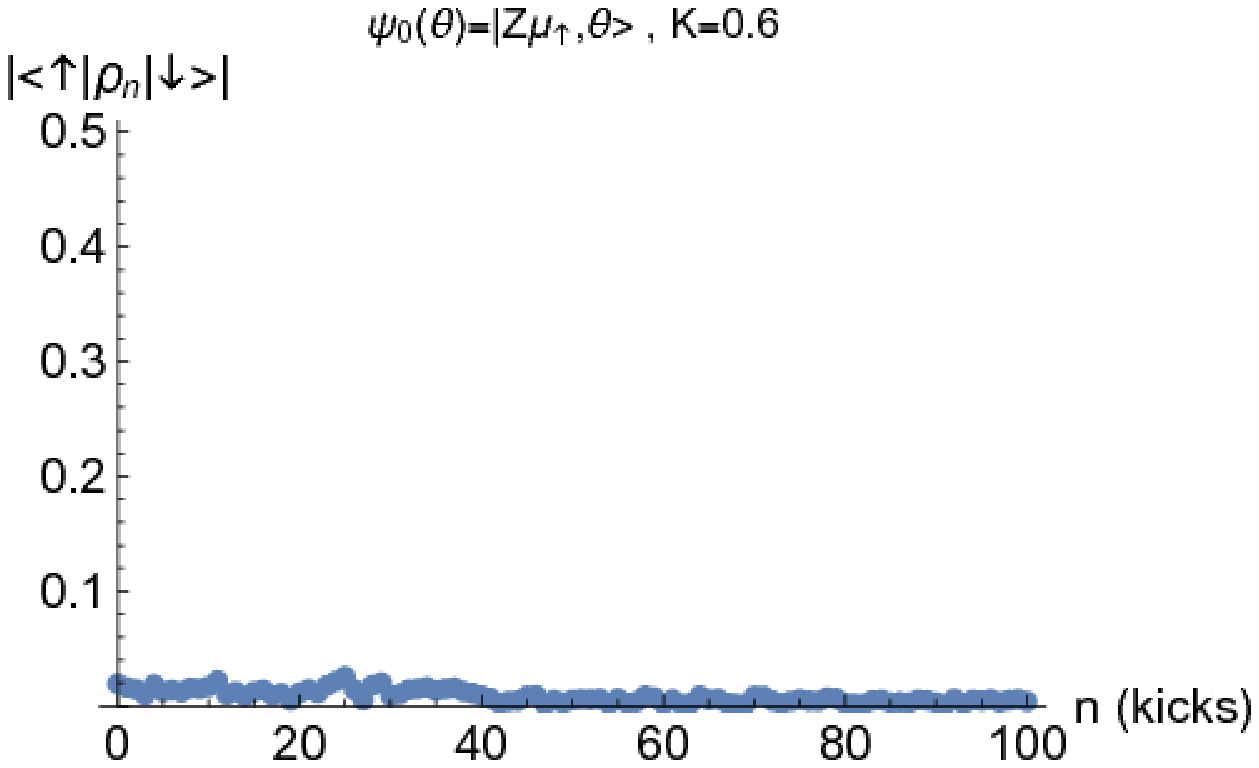} \includegraphics[width=4.2cm]{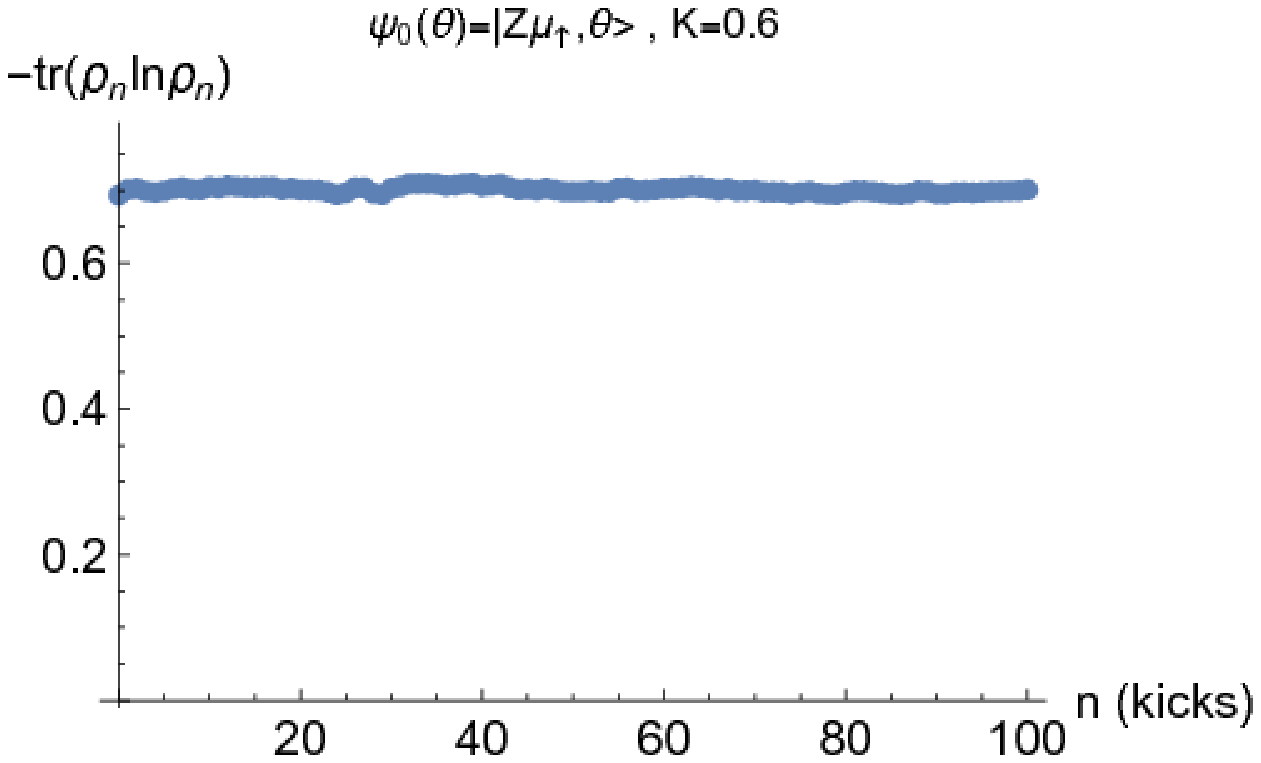}\\
    \caption{\label{dynamicsSTD06} Population of the state $|\uparrow \rangle$ (left), coherence (center) and von Neumann entropy (right) of the mixed state for the stroboscopic dynamics of the spin ensemble (with $\frac{\omega_1}{\omega_0} = 2.5$ and $\theta^2 = 0$) where the kick modulation is governed by the standard map with $K=0.6$ (barely chaotic). The initial condition is for a small uniform dispersion of the first kicks (first line), a large uniform dispersion (second line) and a fundamental quasienergy state (last line).}
  \end{center}
\end{figure}
\begin{figure}
  \begin{center}
    \includegraphics[width=4.2cm]{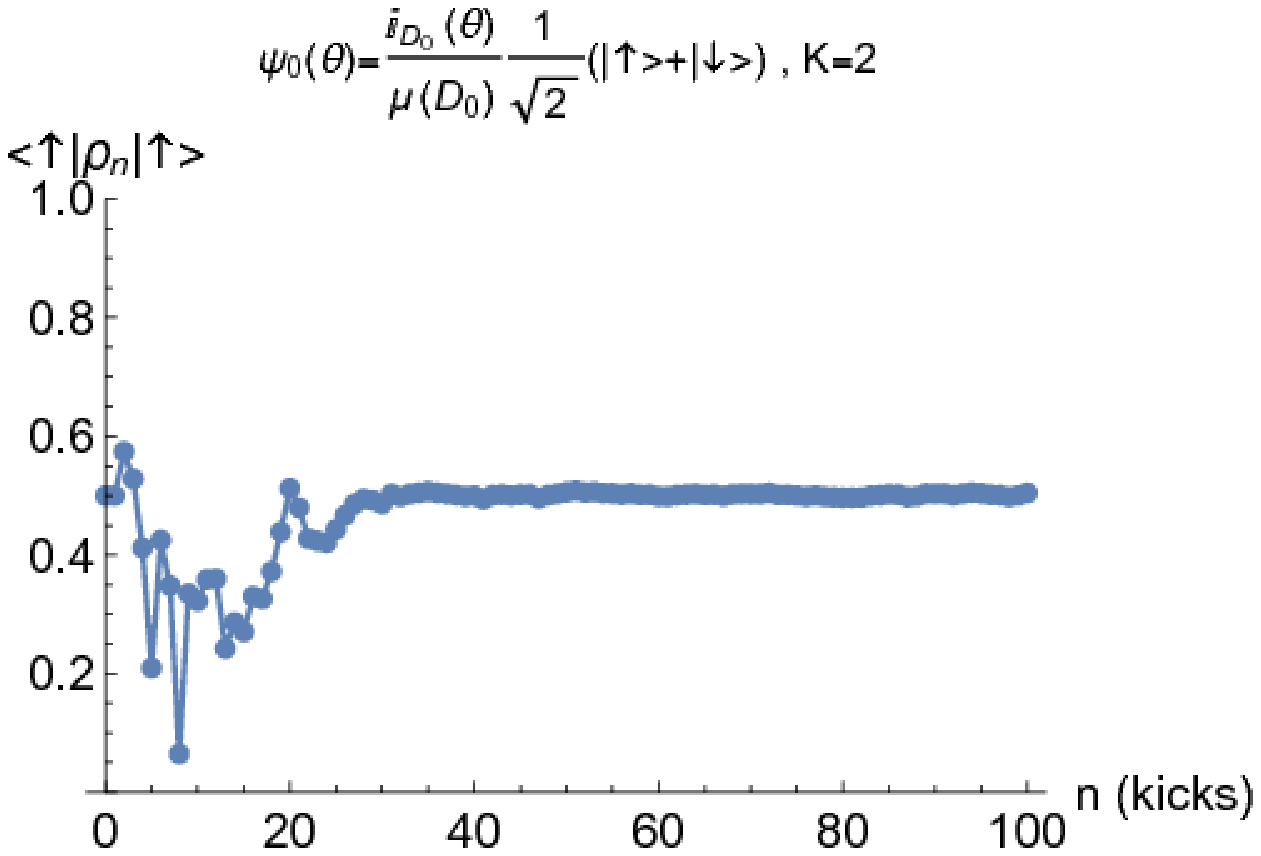} \includegraphics[width=4.2cm]{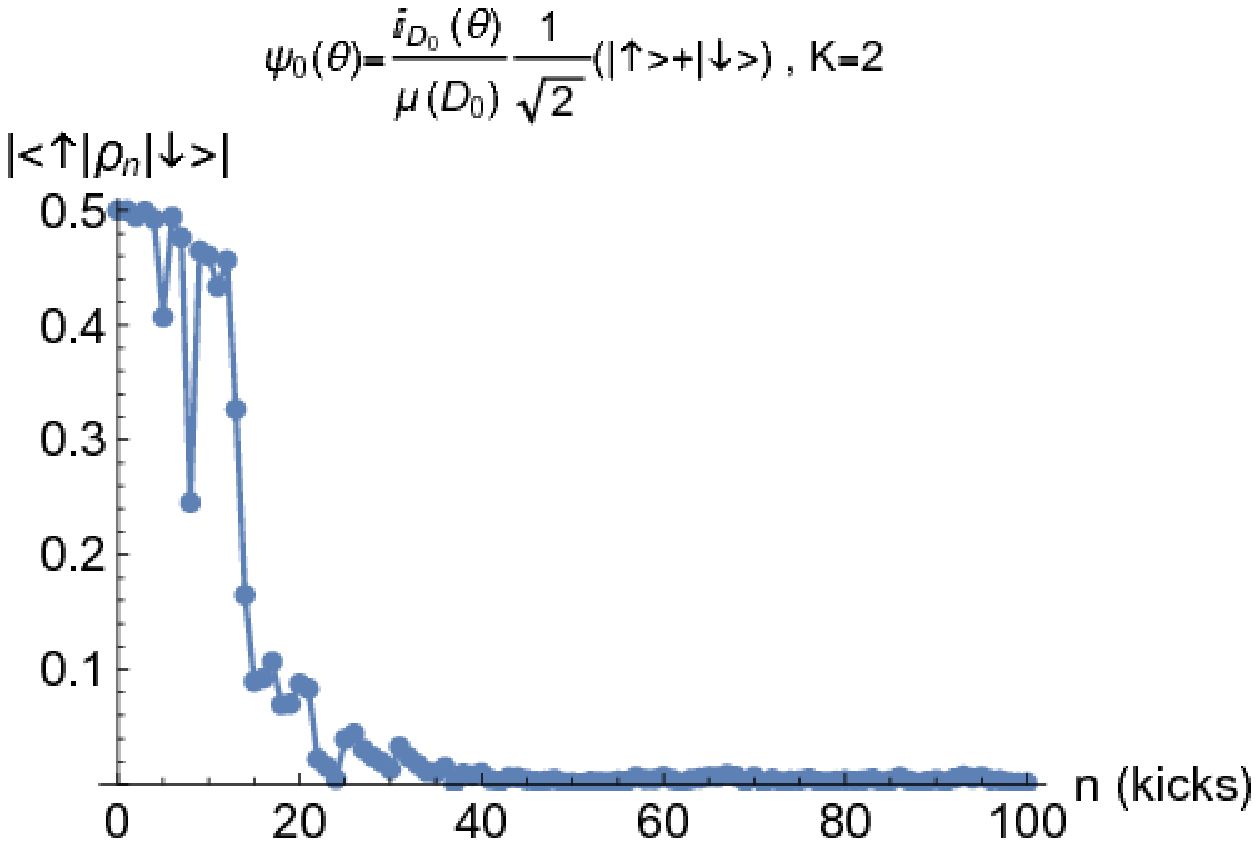} \includegraphics[width=4.2cm]{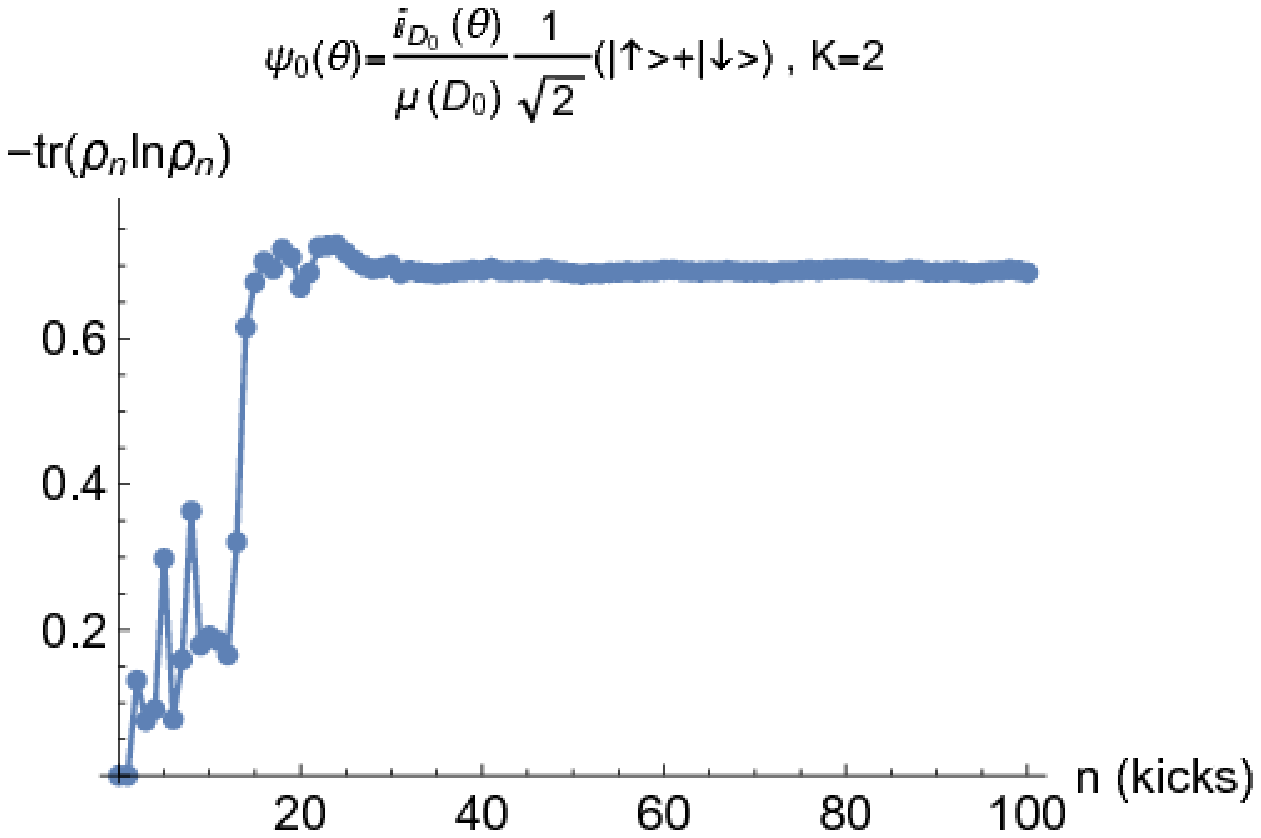}\\
    \includegraphics[width=4.2cm]{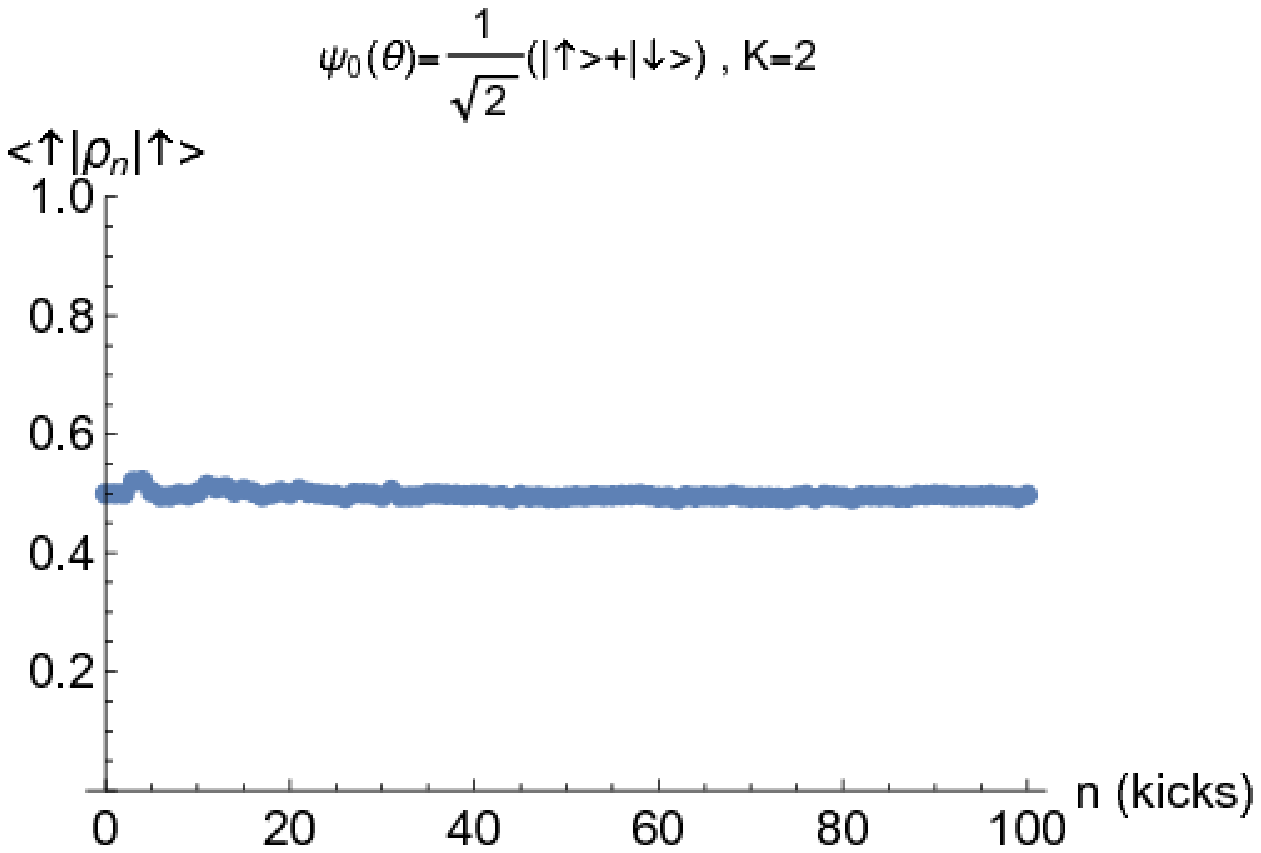} \includegraphics[width=4.2cm]{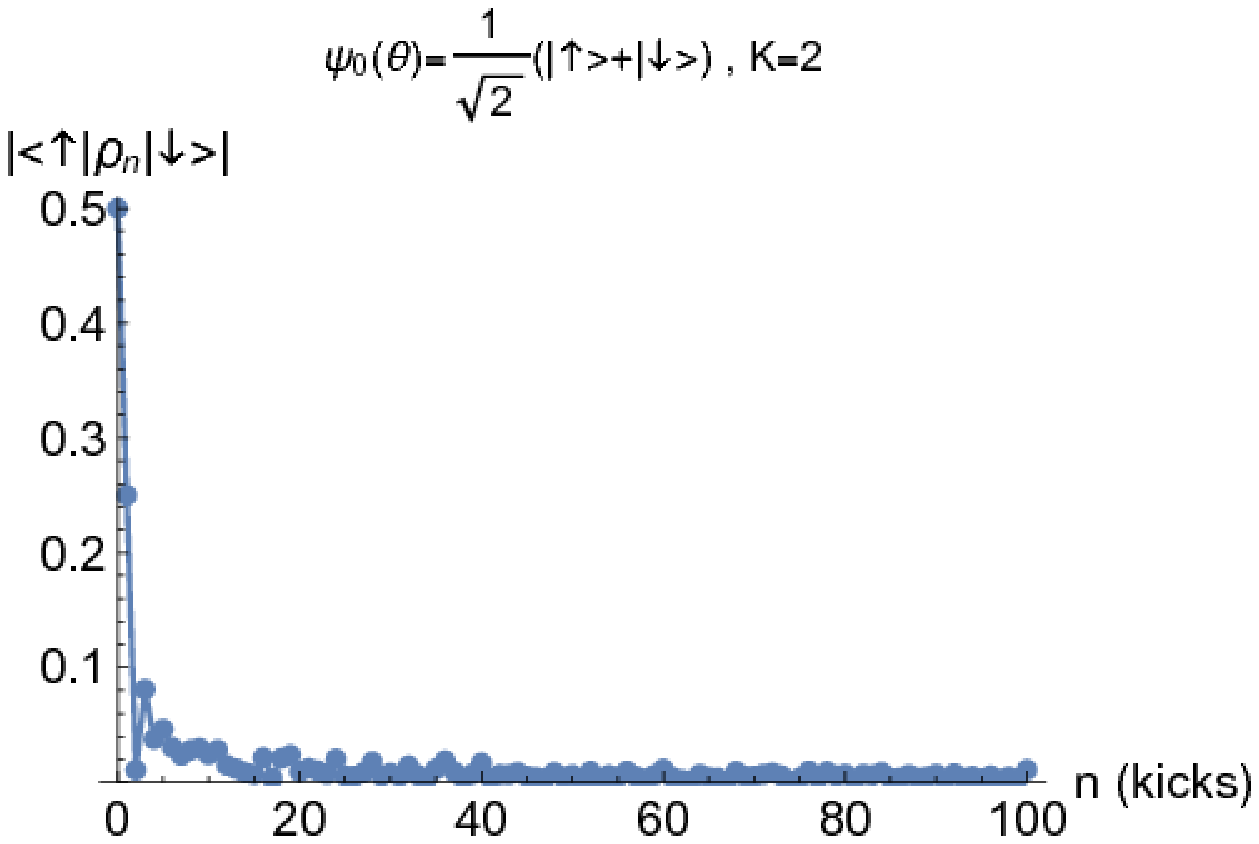} \includegraphics[width=4.2cm]{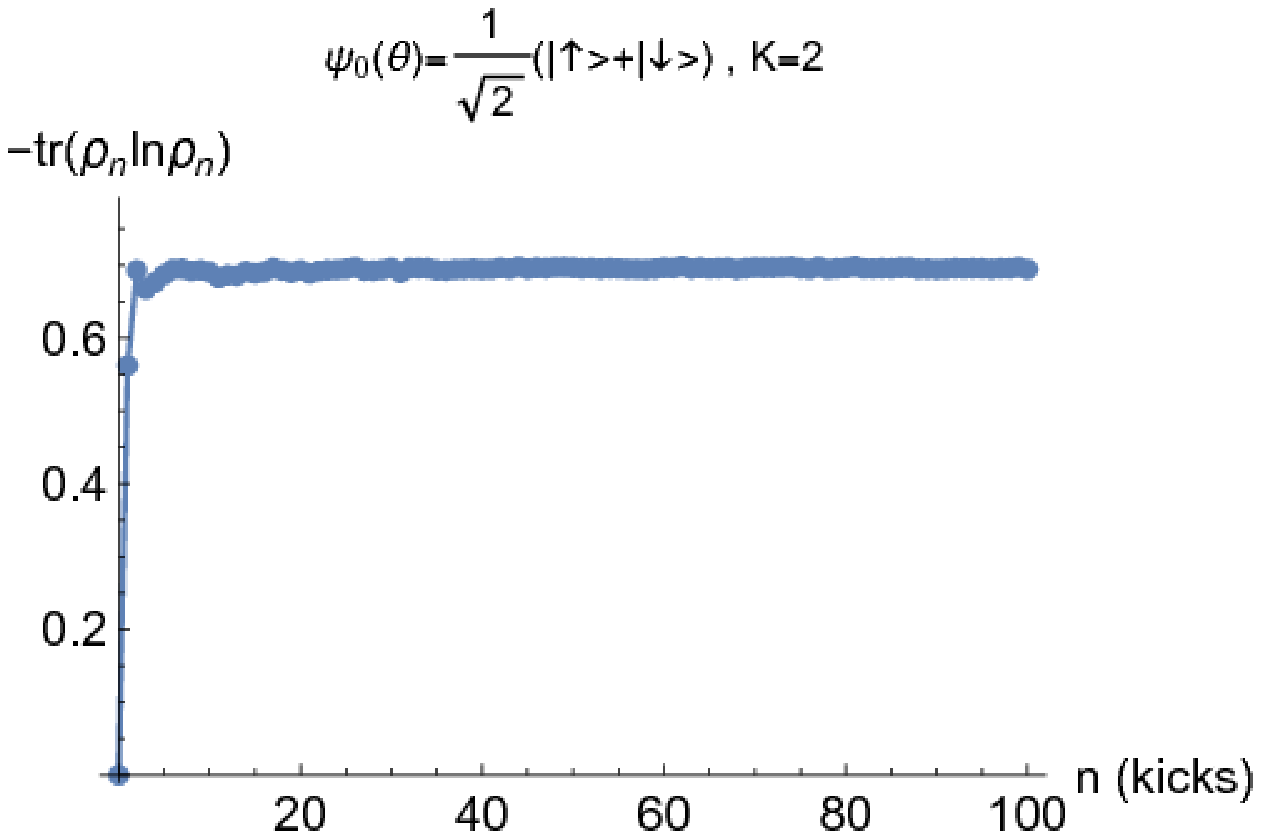}\\
    \includegraphics[width=4.2cm]{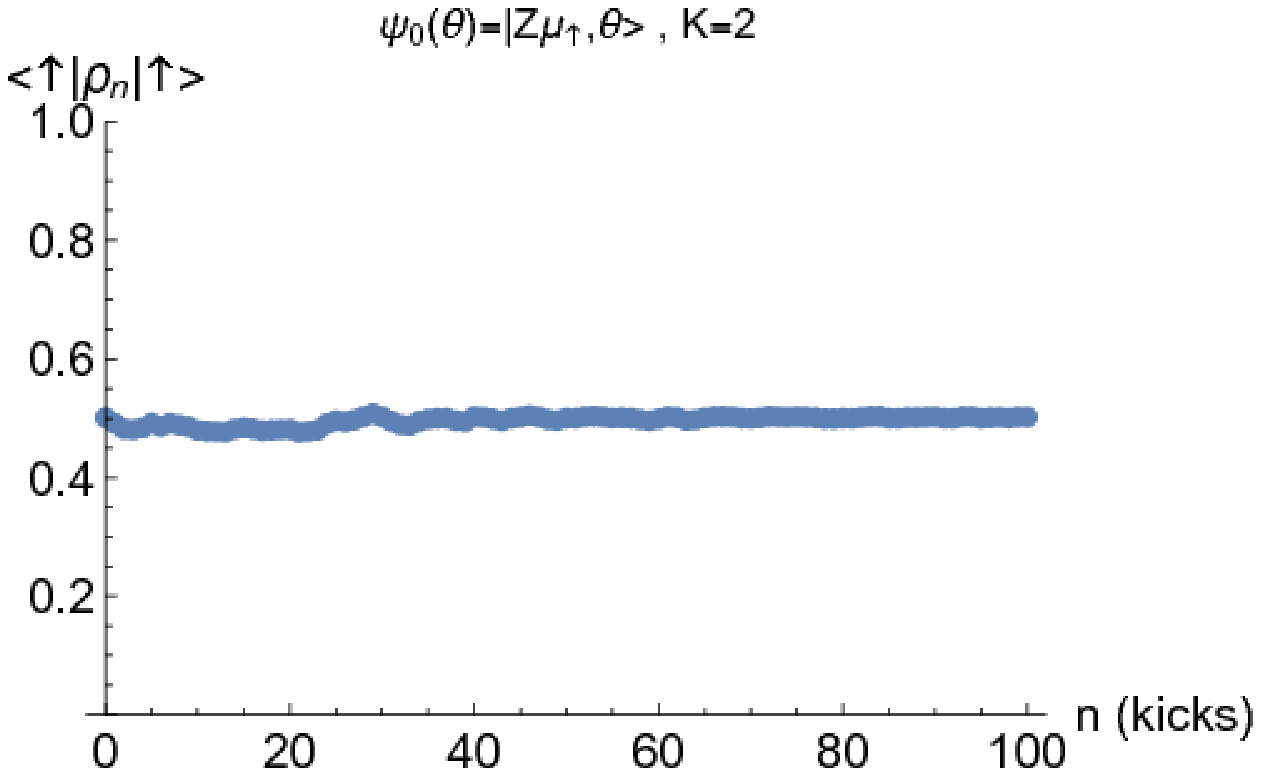} \includegraphics[width=4.2cm]{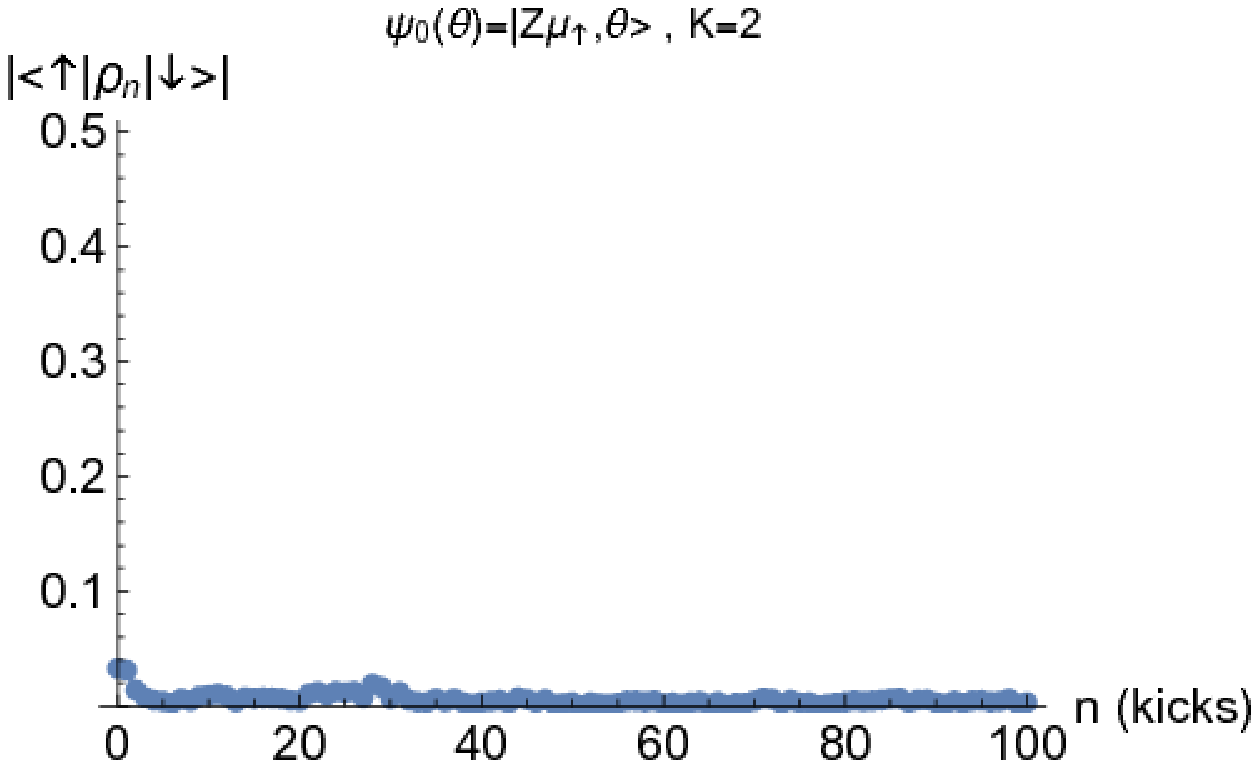} \includegraphics[width=4.2cm]{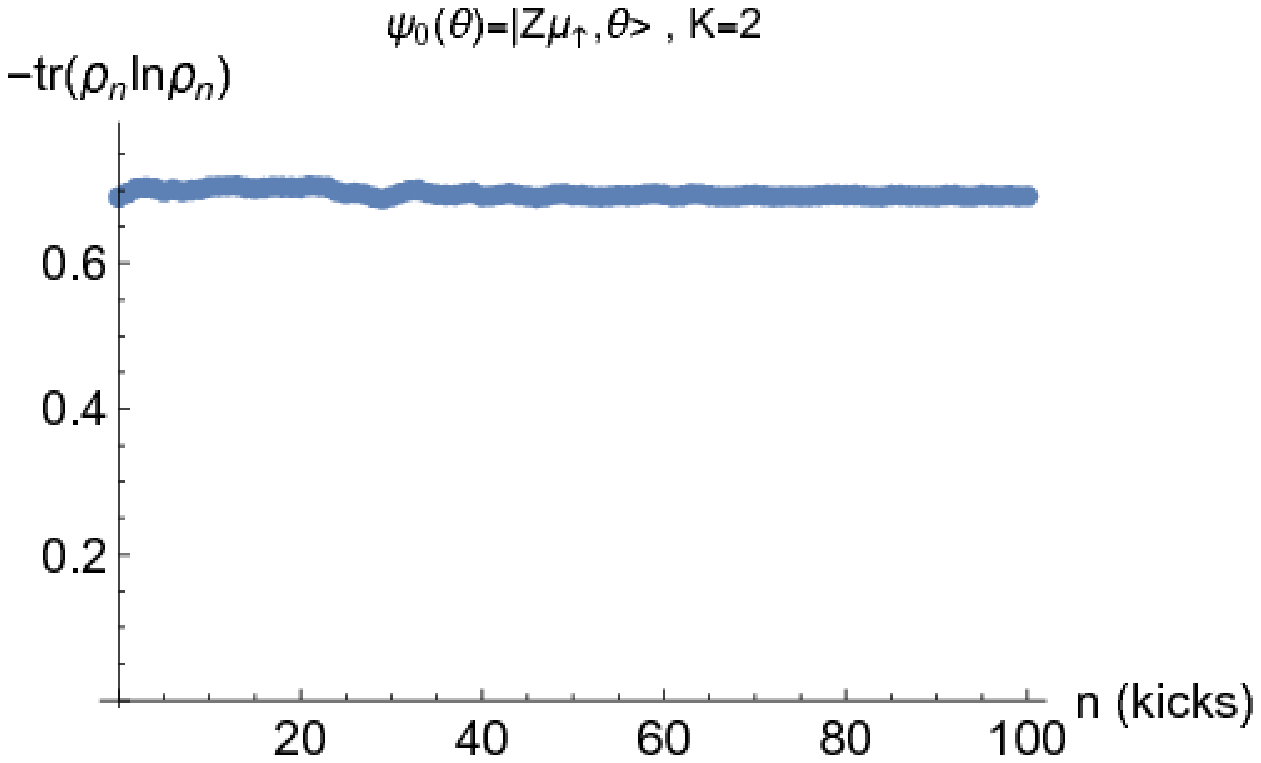}\\
    \caption{\label{dynamicsSTD2} Population of the state $|\uparrow \rangle$ (left), coherence (center) and von Neumann entropy (right) of the mixed state for the stroboscopic dynamics of the spin ensemble (with $\frac{\omega_1}{\omega_0} = 2.5$ and $\theta^2 = 0$) where the kick modulation is governed by the standard map with $K=2$ (highly chaotic). The initial condition is for a small uniform dispersion of the first kicks (first line), a large uniform dispersion (second line) and a fundamental quasienergy state (last line).}
  \end{center}
\end{figure}
For $K=0.6$ (barely chaotic), we do not see decoherence phenomenon for the small initial dispersion because its center is in a region of stability. In contrast, for $K=2$ (strongly chaotic), we have a high decoherence phenomenon. As expected, the fundamental quasienergy state is a steady state.\\
The final states for the initial uniform distribution is represented figure \ref{psiSTD}.
\begin{figure}
  \begin{center}
    \includegraphics[width=5.5cm]{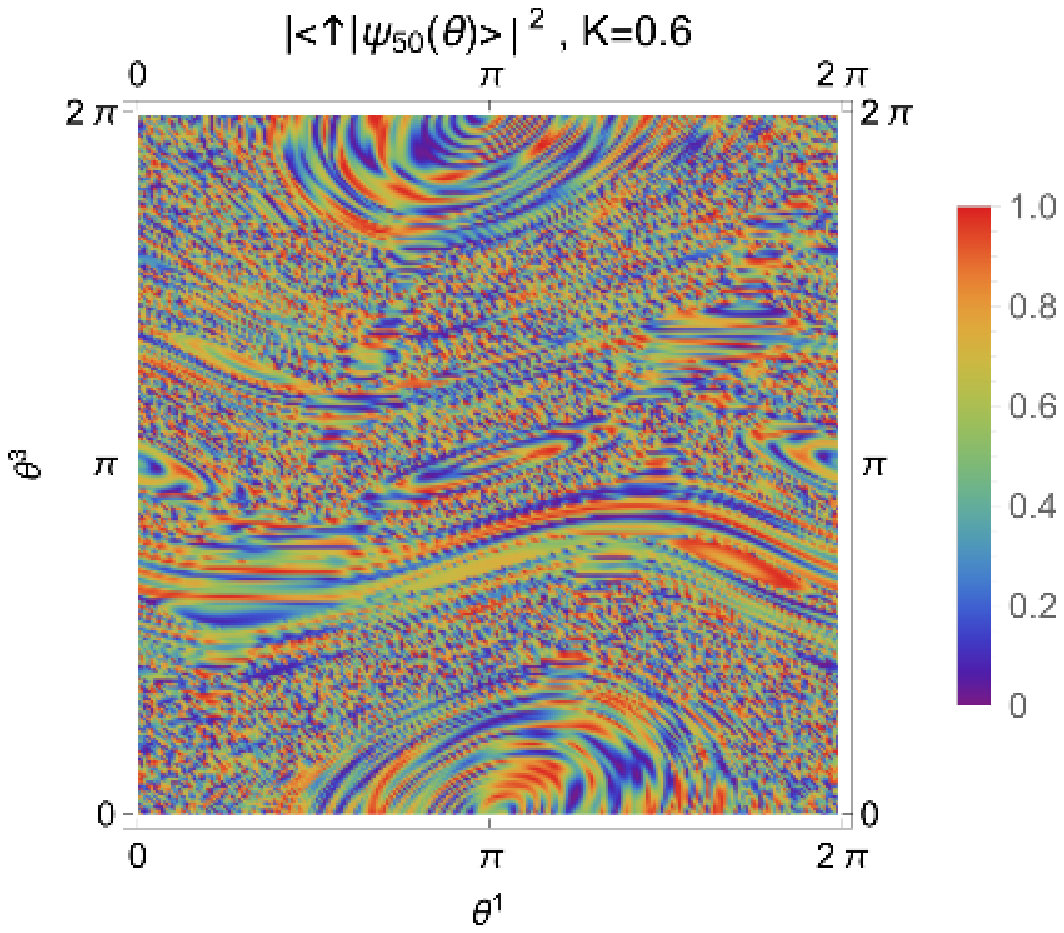} \\
    \includegraphics[width=5.5cm]{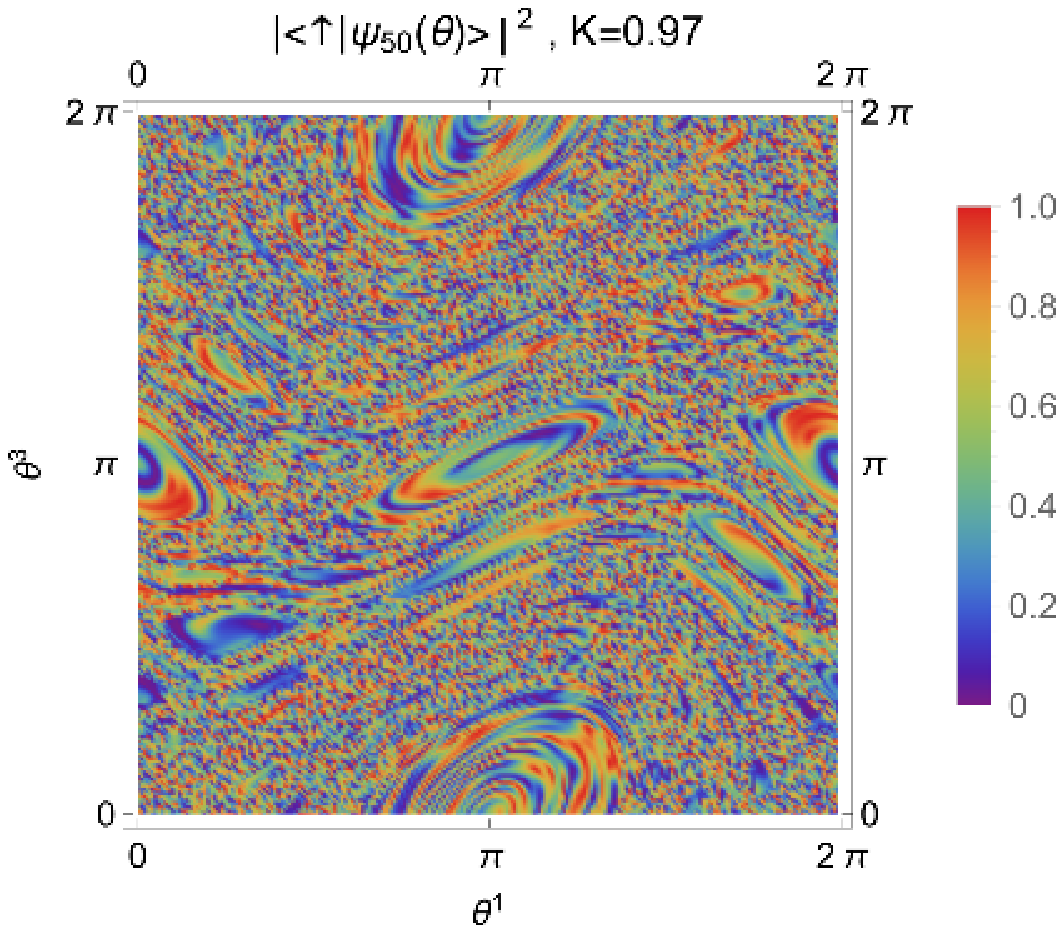}\\
    \includegraphics[width=5.5cm]{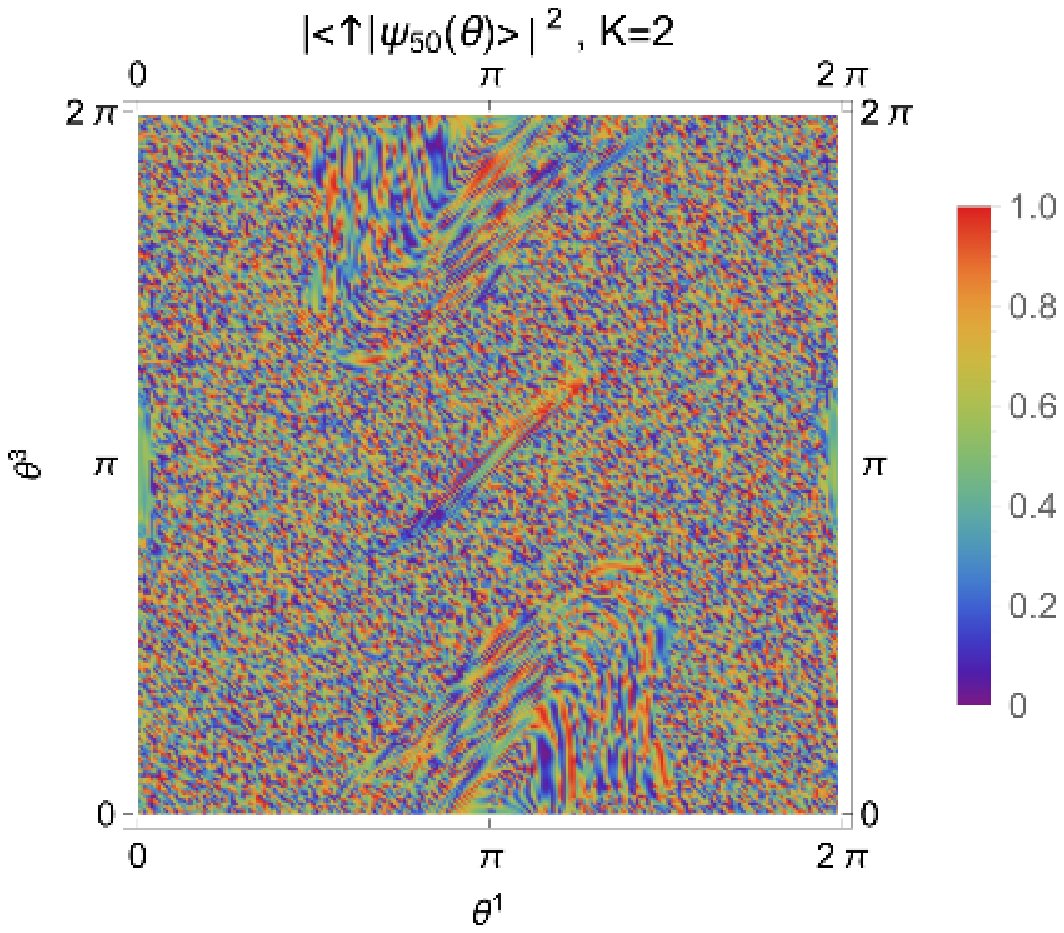}
    \caption{\label{psiSTD} Occupation probability of the state $|\uparrow\rangle$ with respect to $\theta$ for the state $\psi_n$ after $n$ kicks which are modulated following the standard map ($K=0.6, 0.97, 2$) for the uniform initial condition, with $\frac{\omega_1}{\omega_0} = 2.5$ and $\theta^2 = 0$.}
  \end{center}
\end{figure}
We recover the structures of the fundamental quasienergy states.

\section{Discussion and conclusion}
\subsection{Discussion about applications to quantum control and quantum information}
Non-abelian (cyclic) geometric phases are used to develop a geometric method of quantum computation called holonomic quantum computation (HQC) \cite{Lucarelli}. In this approach, the quantum system is supposed to be totally isolated (no decoherence). For more realistic situations where quantum systems are submitted to environmental noises, it is maybe possible to use the ergodic geometric phases associated with SK quasienergies to develop a version of the HQC with decoherence induced by noises modelized by mixing flows. In a same manner, a quantum adiabatic computation algorithm based on Floquet quasienergy states has been proposed in \cite{Tanaka}. It will be interesting to generalize this approach with SK quasienergy states.\\
More precisely, such approaches are developped to perfectly isolated quantum systems. But in the real situations, quantum control or quantum computation are realized on systems submitted to environmental noises responsible to decoherence phenomenons. In some cases, these effects can be modelized by classical random processes as for example in \cite{Yu}. Stochastic noises and chaotic processes are very similar for several properties and the distinguishing is a difficult task (see for example \cite{Ravetti}). We see for example figure \ref{classicalflow} that after enough iterations, the effects of the Arnold's CAT map is very similar to a 2D white noise. Consider the example defined by the evolution operator eq. \ref{Uexemple} where the kick strength $\theta^1$ is the control parameter and where the kick delay $\theta^2$ and the kick direction $\theta^3$ are not controlled but are perturbed by an environmental noise modelized by the Arnold's CAT map. In absence of noise, if we slowly increase $\theta^1$ from $0$ to $2\pi$ for a spin prepared in a Floquet quasienergy state, after the control it is in the other quasienergy state due to a phenomenon called Cheon anholonomy \cite{Miyamoto}. This control can be then assimilated to the realization of the NOT gate in the quasienergy basis. But with the noise modelled by the Arnold's CAT map, the result of the control is totally perturbed. Our approach permits to integrate the effect of the noise in the same formalism by substituting the SK quasienergy states (fig. \ref{quasistateArnold}) to the Floquet quasienergy states. The decoherence effects due to this noise are visible in the density matrix $\rho = \tr_{L^2(\Gamma,d\mu)}|\Psi \rrangle \llangle \Psi|$ as in figure \ref{dynamicsArnold}. This illustrative example is extreme since the noise amplitude is choosen as being very strong (the whole of the delays and the directions is involved by the perturbation), inducing a very rapid decoherence and relaxation to the microcanonical mixed state. The material presented in this paper is just a first step to the applications to quantum control and quantum information with classical noises, since in general the environmental noises are modelled by stochastic processes (as Brownian motion for example) rather than by deterministic chaotic flows. Moreover, in some approaches of the decoherence, the environment is modelled by a quantum bath, and the resulting density matrix obeys to a Lindblad equation. It is known that the Lindblad equation is equivalent to stochastic Schr\"odinger equations (see \cite{Breuer} part 3). These equations are governed by Hamiltonians including stochastic (Wiener or Poisson) processes. Even if the methods used to treat chaotic and stochastic processes are similar \cite{Lasota}, the use of random variables in the classical flow induces some mathematical difficulties which are not the subjet of the present paper. The extension of the present work to stochastic processes could be the subject of futur works.\\
Adiabatic Floquet approach is a tool used to treat the control of quantum systems by laser or magnetic fields (see for example \cite{Leclerc}). In this approach, the fast oscillations of the electromagnetic field is treated by using the Floquet theory, and adiabatic control is realized by slow variations of the other field parameters (amplitude, phase, polarisation direction,...), the adiabatic approach concerning the instantaneous Floquet quasienergy states. But this control theory is limited by the restriction than all control parameters must be slowly modulated. We can extend the possible control field shapes by considering fast evolving control parameters governed by a classical flow with other slow evolving control parameters used for the adiabatic control. The fast oscillations of the field and the fast control parameters can be treated by the SK approach, and an adiabatic approach can be used on the resulting instantaneous SK quasienergy states. The solving of the control problem consists then to find the shape of the control path in the slow parameter space and to fix the classical flow (model, parameters, initial condition). The added adjustable property associated to the classical flow in this adiabatic SK method increases the possibilities of accessible control goals with respect to the usual adiabatic Floquet method. Moreover, parameters defining the classical flow can be also slowly modulated. As exemple, we can consider the model eq. \ref{Uexemple} where the kick strength $\theta^1$ and the kick delay $\theta^2$ is governed by the standard map and where the kick direction $\theta^3$ and the map parameter $K$ are slowly modulated to realize an adiabatic control (slow modulations involving that the significative evolutions of $\theta^3$ and $K$ correspond to several kicks). Evolutions of $K$ permit to change the chaotic sea and the islands of stability in the SK quasienergy states (fig. \ref{quasistateSTD}) used in the adiabatic control. Such applications can be the subject of future works, which need to study an adiabatic theorem adaptated to the SK quasienergy states.

\subsection{Conclusion}
SK quasienergy states can be used to study mixed classical-quantum system, quantum control and quantum information with classical noises. A kicked spin ensemble with kick modulation following a classical flow can be an example of these three cases \cite{Viennot2, Aubourg, Aubourg2}. The fundamental quantum quasienergies are associated with the fixed points, the cyclic points and the ergodic orbits of the classical flow. It is interesting to compare a cyclic CAT map with the chaotic Arnold's CAT map. The Koopman spectrum of the flow of the cyclic map is pure point whereas the fundamental SK spectrum of the spin ensemble driven by this flow has a continuous component. In contrast, the Koopman spectrum of the flow of the Arnold's map is continuous, whereas the fundamental SK spectrum of the driven spin ensemble is pure point. The quasienergy states are the steady states of the driven quantum system which are associated with specific geometric phases if the flow is ergodic. The reduced density matrix of the quantum system evolves to a density matrix of these steady states if the flow is mixing. In the examples, we have seen that the structures appearing in the phase space of the classical flow are transmitted to the quasienergy state of the quantum ensemble as probability distributions. Another specific structures associated with the structure of the Hamiltonian or of the evolution operator appear as well as interferences for the SK modes in the regions of cyclic orbits.\\
In this paper, we have treated only conservative flows. It will be interesting to study the case of dissipative flows, particularly chaotic dissipative flows having a strange attractor. Such systems are more complicated since their Koopman operators are not unitary (and then their SK evolution operators are not unitary and their SK quasienergy spectra will not be real). Another question concerns the interpretation of the SK states as states of an ensemble of copies of one quantum system, as in the example of the spin ensemble treated in this paper. To have a simple interpretation of the SK quasienergy states, we have supposed that no interaction between the spins occurs. It will be interesting to find how modify the SK theory to take into account the interactions between the quantum subsystems driven by the classical flow.

\ack
The authors acknowledge support from I-SITE Bourgogne-Franche-Comt\'e under grants from the I-QUINS project, and support from the R\'egion de Bourgogne-Franche-Comt\'e under grants from the APEX project. Simulations have been executed on computers of the Utinam Institute supported by the R\'egion de Bourgogne-Franche-Comt\'e and the Institut des Sciences de l'Univers (INSU).\\
The authors thank Professor Hans-Rudolf Jauslin for useful discussions.

\appendix
\section{Usefull properties of the Koopman operator} \label{appendixA}
\begin{propo}
\label{algepropspect}
Let $\lambda_1,\lambda_2 \in \Sp(F^\mu\partial_\mu)$ be two eigenvalues associated with $f_{\lambda_1}(\theta)$ and $f_{\lambda_2}(\theta)$, then $\lambda_1+\lambda_2 \in \Sp(F^\mu \partial_\mu)$ with the associated eigenfunction $f_{\lambda_1+\lambda_2}(\theta) = f_{\lambda_1}(\theta) f_{\lambda_2}(\theta)$. Moreover, $\forall r \in \mathbb R$ such that $(f_{\lambda_1})^r \in \mathcal C^1(\Gamma)$, then $r\lambda_1 \in \Sp(F^\mu \partial_\mu)$ with the associated eigenfunction $f_{r\lambda_1}(\theta)= (f_{\lambda_1}(\theta))^r$.
\end{propo}
This result follows directly from the fact that the Koopman generator is a first order derivative. Note that the condition $(f_{\lambda})^r \in \mathcal C^1(\Gamma)$ can drastically reduce the acceptable $r$. For example let $(\mathbb S^1,\varphi^t,\frac{d\theta}{2\pi})$ be the classical dynamical system such that $\varphi^t(\theta) = \theta + \omega t \mod 2\pi$ (with $\omega>0$ constant). The Koopman generator is $\omega \frac{\partial}{\partial \theta}$, and $\Sp(\omega \frac{\partial}{\partial \theta}) = \imath \omega \mathbb Z$ with $f_{\imath n \omega}(\theta) = e^{\imath n \theta}$ ($n\in \mathbb Z$). $(f_{\imath n \omega})^r \in \mathcal C^1(\mathbb S^1)$ (continuous, derivable and $2\pi$-periodic with respect to $\theta$) if only if $r \in \mathbb Z$.

\begin{prop}
\label{mixergo}
Let $(\Gamma,\varphi^t,\mu)$ be a conservative dynamical system.
\begin{itemize}
\item If the dynamical system is mixing then $\forall f,g \in L^2(\Gamma,d\mu)$
\begin{equation}
\lim_{t\to +\infty} \int_\Gamma \overline{g(\theta)} \mathcal T^t f(\theta)d\mu(\theta) = \int_\Gamma \overline{g(\theta)} d\mu(\theta) \int_\Gamma f(\theta) d\mu(\theta)
\end{equation}
\item If the dynamical system is ergodic then $\forall f\in L^2(\Gamma,d\mu)$, for $\mu$-almost all $\theta_0 \in \Gamma$,
\begin{equation}
\lim_{t \to +\infty} \frac{1}{T} \int_0^T \mathcal T^t f(\theta_0) dt = \int_\Gamma f(\theta) d\mu(\theta)
\end{equation}
\end{itemize}
\end{prop}

\begin{proof}
See \cite{Eisner}.
\end{proof}

\begin{prop} \label{nullFmix}
Let $(\Gamma,\varphi^t,\mu)$ be a mixing conservative classical dynamical system such that $\Sp(F^\mu\partial_\mu)\setminus\{0\}$ is continuous. $\forall \lambda \in \Sp(F^\mu\partial_\mu)\setminus\{0\}$ we have
\begin{equation}
\int_\Gamma f_\lambda(\theta) d\mu(\theta) = 0
\end{equation}
\end{prop}

\begin{proof}
Since the dynamical system is mixing and then ergodic, $f_\lambda$ is unimodular (see \cite{Eisner}) and then $f_\lambda \in L^2(\Gamma,d\mu)$). It follows that
\begin{eqnarray}
& & \lim_{t \to +\infty} \int_{\lambda-\delta \lambda}^{\lambda +\delta \lambda} \int_\Gamma \overline{f_\lambda(\theta)} f_\lambda(\varphi^t(\theta)) d\mu(\theta) d\lambda \nonumber \\
& & \qquad = \int_{\lambda-\delta \lambda}^{\lambda +\delta \lambda} |\int_\Gamma f_\lambda(\theta) d\mu(\theta)|^2 d\lambda
\end{eqnarray}
where $\delta \lambda$ is such that $0 \not\in [\lambda-\delta \lambda,\lambda+\delta\lambda]$. But
\begin{eqnarray}
& & \lim_{t \to +\infty} \int_{\lambda-\delta \lambda}^{\lambda +\delta \lambda} \int_\Gamma \overline{f_\lambda(\theta)} f_\lambda(\varphi^t(\theta)) d\mu(\theta) d\lambda \nonumber \\
& & \qquad = \lim_{t \to +\infty} \int_{\lambda-\delta \lambda}^{\lambda +\delta \lambda} e^{\lambda t} \int_\Gamma |f_\lambda(\theta)|^2 d\mu(\theta) d\lambda \\
& & \qquad = \lim_{t \to +\infty} \int_{\lambda-\delta \lambda}^{\lambda +\delta \lambda} e^{\lambda t} d\lambda
\end{eqnarray}
Because $|f_\lambda(\theta)|^2=1$ since it is unimodular. But $\lim_{t \to +\infty} \int_{\lambda-\delta \lambda}^{\lambda +\delta \lambda} e^{\lambda t} d\lambda=0$ ($\lambda \in \imath \mathbb R^*$), it follows that $\int_{\lambda-\delta \lambda}^{\lambda +\delta \lambda} |\int_\Gamma f_\lambda(\theta) d\mu(\theta)|^2 d\lambda = 0 \Rightarrow |\int_\Gamma f_\lambda(\theta) d\mu(\theta)|^2 = 0$.
\end{proof}

\section{Expansion of the quasienergy states}
\subsection{Local expansion} \label{appendixB1}
\begin{prop}
Let $(\Gamma,\mu,\varphi^t,\mathcal H,H)$ be a conservative driven quantum system, $\theta_* \in \Gamma$ be a fixed point of $\varphi^t$, $\{\tilde \chi_i\}_i$ be the fundamental quasienergies associated with $\theta_*$ and $\{|Z\mu_i,\theta\rangle\}_i$ be the associated quasienergy states. Let $(\vec e_a)_a$ be the eigendirections in $\Gamma$ in the neighbourhood of $\theta_*$ and $\{\underline \lambda_a\}_a$ be the associated local Lyapunov eigenvalues (i.e. the eigenvectors and the eigenvalues of the Jacobian matrix of the flow $\partial \varphi^t_{\theta_*}$ supposed here diagonalizable). We have
\begin{equation}
\left. \langle Z\mu_j,\theta|\nabla_{\vec e_a}|Z\mu_i,\theta\rangle \right|_{\theta=\theta_*} = \frac{\langle Z\mu_j,\theta_*| \left. \nabla_{\vec e_a} H \right|_{\theta=\theta_*} | Z\mu_i, \theta_* \rangle}{\tilde \chi_i - \tilde \chi_j + \ihbar \underline \lambda_a}
\end{equation}
\end{prop}

\begin{proof}
$U(t,0;\theta)|Z\mu_i,\theta\rangle = e^{-\ihbar^{-1} \tilde \chi_i t} |Z\mu_i,\varphi^t(\theta)\rangle$. Let $\theta = \theta_* + \delta \theta$. By Taylor expansions we have $\varphi^t(\theta)^\mu = \theta_*^\mu + {(\partial \varphi^t_{\theta_*})^\mu}_\nu \delta \theta^\nu + \mathcal O(\|\delta \theta\|^2)$ and
\begin{eqnarray}
& & \left(U(t,0;\theta_*)+\left.\frac{\partial U}{\partial \theta^\mu} \right|_{\theta_*} \delta \theta^\mu \right) \left(|Z\mu_i,\theta_* \rangle + \left.\frac{\partial}{\partial \theta^\mu}|Z\mu_i\rangle\right|_{\theta_*} \delta \theta^\mu \right) \nonumber \\
& & = e^{-\ihbar^{-1} \tilde \chi_i t}  \left(|Z\mu_i,\theta_* \rangle + \left.\frac{\partial}{\partial \theta^\mu}|Z\mu_i\rangle\right|_{\theta_*} {(\partial \varphi^t_{\theta_*})^\mu}_\nu \delta \theta^\nu \right) + \mathcal O(\|\delta \theta\|^2)
\end{eqnarray}
\begin{eqnarray}
& & U(t,0;\theta_*)\left.\frac{\partial}{\partial \theta^\mu}|Z\mu_i\rangle\right|_{\theta_*} \delta \theta^\mu+\left.\frac{\partial U}{\partial \theta^\mu} \right|_{\theta_*} |Z\mu_i,\theta_* \rangle \delta \theta^\mu  \nonumber \\
& & = e^{-\ihbar^{-1} \tilde \chi_i t}  \left.\frac{\partial}{\partial \theta^\mu}|Z\mu_i\rangle\right|_{\theta_*} {(\partial \varphi^t_{\theta_*})^\mu}_\nu \delta \theta^\nu 
\end{eqnarray}
\begin{eqnarray}
& & e^{-\ihbar^{-1} \tilde \chi_j t} \left. \langle Z\mu_j|\frac{\partial}{\partial \theta^\mu}|Z\mu_i\rangle\right|_{\theta_*} + \langle Z\mu_j,\theta_*|\left.\frac{\partial U}{\partial \theta^\mu} \right|_{\theta_*} |Z\mu_i,\theta_* \rangle \nonumber \\
& & = e^{-\ihbar^{-1} \tilde \chi_i t}  \left. \langle Z\mu_j|\frac{\partial}{\partial \theta^\nu}|Z\mu_i\rangle\right|_{\theta_*} {(\partial \varphi^t_{\theta_*})^\nu}_\mu
\end{eqnarray}
\begin{eqnarray}
& & \langle Z\mu_j,\theta_*|\left.\frac{\partial U}{\partial \theta^\mu} \right|_{\theta_*} |Z\mu_i,\theta_* \rangle \nonumber \\
& & \label{dU} = \left(e^{-\ihbar^{-1} \tilde \chi_i t} {(\partial \varphi^t_{\theta_*})^\nu}_\mu - e^{-\ihbar^{-1} \tilde \chi_j t} {\delta^\nu}_\mu \right) \left. \langle Z\mu_j|\frac{\partial}{\partial \theta^\nu}|Z\mu_i\rangle\right|_{\theta_*}
\end{eqnarray}
\begin{eqnarray}
& & \langle Z\mu_j,\theta_*|\left.\ihbar \frac{\partial^2 U}{\partial t \partial \theta^\mu} \right|_{\theta_*} |Z\mu_i,\theta_* \rangle \nonumber \\
& & = \left(\tilde \chi_i e^{-\ihbar^{-1} \tilde \chi_i t} {(\partial \varphi^t_{\theta_*})^\nu}_\mu + \ihbar e^{-\ihbar^{-1} \tilde \chi_i t} {(\partial \dot \varphi^t_{\theta_*})^\nu}_\mu - \tilde \chi_j e^{-\ihbar^{-1} \tilde \chi_j t} {\delta^\nu}_\mu \right) \nonumber \\
& & \label{d2U} \quad \times \left. \langle Z\mu_j|\frac{\partial}{\partial \theta^\nu}|Z\mu_i\rangle\right|_{\theta_*}
\end{eqnarray}
But
\begin{eqnarray}
\ihbar \frac{\partial^2 U}{\partial t \partial \theta^\mu} & = & \frac{\partial}{\partial \theta^\mu} \left(H(\varphi^t(\theta))U(t,0;\theta) \right) \\
& = & \left. \frac{\partial H}{\partial \theta^\nu} \right|_{\varphi^t(\theta)} {(\partial\varphi^t_\theta)^\nu}_\mu U(t,0;\theta)+ H(\varphi^t(\theta)) \frac{\partial U}{\partial \theta^\mu}
\end{eqnarray}
By using equation \ref{dU}
\begin{eqnarray}
& & \langle Z\mu_j,\theta_*|\left.\ihbar \frac{\partial^2 U}{\partial t \partial \theta^\mu} \right|_{\theta_*} |Z\mu_i,\theta_* \rangle \nonumber \\
& & = e^{-\ihbar^{-1} \tilde \chi_i t} \langle Z\mu_j,\theta_*|\left. \frac{\partial H}{\partial \theta^\nu} \right|_{\theta_*}|Z\mu_i,\theta_*\rangle {(\partial\varphi^t_\theta)^\nu}_\mu \nonumber\\
& & + \tilde \chi_j \left(e^{-\ihbar^{-1} \tilde \chi_i t} {(\partial \varphi^t_{\theta_*})^\nu}_\mu - e^{-\ihbar^{-1} \tilde \chi_j t} {\delta^\nu}_\mu \right) \left. \langle Z\mu_j|\frac{\partial}{\partial \theta^\nu}|Z\mu_i\rangle\right|_{\theta_*}
\end{eqnarray}
By comparison with equation \ref{d2U} we have
\begin{eqnarray}
&& \langle Z\mu_j,\theta_*|\left. \frac{\partial H}{\partial \theta^\nu} \right|_{\theta_*}|Z\mu_i,\theta_*\rangle {(\partial\varphi^t_\theta)^\nu}_\mu \nonumber\\
& & = \left((\tilde \chi_i - \tilde \chi_j) {(\partial\varphi^t_\theta)^\nu}_\mu + \ihbar  {(\partial \dot \varphi^t_{\theta_*})^\nu}_\mu \right) \left. \langle Z\mu_j|\frac{\partial}{\partial \theta^\nu}|Z\mu_i\rangle\right|_{\theta_*}
\end{eqnarray}
$\dot \varphi^t(\theta) = F(\varphi^t(\theta)) \Rightarrow {(\partial \dot \varphi^t_\theta)^\nu}_\mu = {(\partial F_{\varphi^t(\theta)})^\nu}_\rho {(\partial \varphi^t)^\rho}_\mu$ it follows that
\begin{eqnarray}
&& \langle Z\mu_j,\theta_*|\left. \frac{\partial H}{\partial \theta^\nu} \right|_{\theta_*}|Z\mu_i,\theta_*\rangle {(\partial\varphi^t_\theta)^\nu}_\mu \nonumber\\
& & = \left((\tilde \chi_i - \tilde \chi_j) {\delta^\nu}_\rho + \ihbar {(\partial F_{\theta_*})^\nu}_\rho \right) {(\partial\varphi^t_\theta)^\rho}_\mu \left. \langle Z\mu_j|\frac{\partial}{\partial \theta^\nu}|Z\mu_i\rangle\right|_{\theta_*}
\end{eqnarray}
By definition ${(\partial F_{\theta_*})^\nu}_\rho e^\rho_a = \underline \lambda_a e^\nu_a$ and ${(\partial \varphi^t_{\theta_*})^\rho}_\mu e^\mu_a = e^{\underline \lambda_a t} e^\rho_a$.
\begin{eqnarray}
&& \langle Z\mu_j,\theta_*|\left. \frac{\partial H}{\partial \theta^\nu} \right|_{\theta_*}|Z\mu_i,\theta_*\rangle e^{\underline\lambda_a t} e^\nu_a \nonumber\\
& & = \left(\tilde \chi_i - \tilde \chi_j+ \ihbar \underline\lambda_a \right) e^{\underline\lambda_a t} e^\nu_a \left. \langle Z\mu_j|\frac{\partial}{\partial \theta^\nu}|Z\mu_i\rangle\right|_{\theta_*}
\end{eqnarray}
\begin{eqnarray}
& & \langle Z\mu_j,\theta_*|\left. \nabla_{\vec e_a} H \right|_{\theta_*}|Z\mu_i,\theta_*\rangle  \nonumber \\
& & = \left(\tilde \chi_i - \tilde \chi_j+ \ihbar \underline\lambda_a \right) \left. \langle Z\mu_j|\nabla_{\vec e_a}|Z\mu_i\rangle\right|_{\theta_*}
\end{eqnarray}
\end{proof}

For a stroboscopic driven quantum system we have
\begin{equation}
  \left. \langle Z\mu_j,\theta|\nabla_{\vec e_a}|Z\mu_i,\theta\rangle\right|_{\theta=\theta_*} = \frac{\langle Z\mu_j,\theta_*|\left. \nabla_{\vec e_a} U \right|_{\theta=\theta_*}|Z\mu_i,\theta_* \rangle}{e^{-\imath (\tilde \chi_i+\imath \underline \lambda_a)} - e^{-\imath \tilde \chi_j}}
\end{equation}
for a fixed point $\theta_*$. Moreover, we can also consider a $p$-cyclic point $\theta_*$ and we have
\begin{equation}
  \left. \langle Z\mu_j,\theta|\nabla_{\vec e_a}|Z\mu_i,\theta\rangle\right|_{\theta=\theta_*} = \frac{\langle Z\mu_j,\theta_*|\left. \nabla_{\vec e_a} V_p \right|_{\theta=\theta_*}|Z\mu_i,\theta_* \rangle}{e^{-\imath p(\tilde \chi_i+\imath \underline \lambda_a)} - e^{-\imath p \tilde \chi_j}}
\end{equation}
with $V_p(\theta) = U(\varphi^{p-1}(\theta))...U(\theta)$.\\

Let $\vartheta^a = e^a_\mu \theta^\mu$ be the eigencoordinates in the neighbourhood of $\theta_*$. By using this property we can write
\begin{eqnarray}
|Z\mu_i,\theta \rangle & = & |Z\mu_i,\theta_*\rangle + \sum_a \left. \nabla_{\vec e_a}|Z\mu_i\rangle \right|_{\theta_*} (\vartheta^a-\vartheta^a_*) + \mathcal O(\|\theta-\theta_*\|^2) \\
& = & |Z\mu_i,\theta_*\rangle \nonumber \\
& & \quad + \sum_a\sum_j \frac{\langle Z\mu_j,\theta_*|\left.\nabla_{\vec e_a}H \right|_{\theta_*}|Z\mu_i,\theta_*\rangle}{\tilde \chi_i - \tilde \chi_j + \ihbar \underline\lambda_a} |Z\mu_j,\theta_*\rangle (\vartheta^a-\vartheta^a_*) \nonumber \\
& & \quad + \mathcal O(\|\theta-\theta_*\|^2)
\end{eqnarray}
because $(|Z\mu_j,\theta_*\rangle)_j$ is a basis of $\mathcal H$ (it is the set of the eigenvectors of $H(\theta_*)$). We see that $\frac{\langle Z\mu_j,\theta_*|\left.\nabla_{\vec e_a}H \right|_{\theta_*}|Z\mu_i,\theta_*\rangle}{\tilde \chi_i - \tilde \chi_j + \ihbar \underline \lambda_a}(\vartheta^a-\vartheta^a_*) $ measures the propensity of the classical dynamical system to induce a transition from $|Z\mu_i,\theta_*\rangle$ to $|Z\mu_j,\theta_*\rangle$ in the neighbourhood of $\theta_*$. We see also a phenomenon of resonance if $\R(\underline\lambda_a)\simeq 0$  and $\tilde \chi_i-\tilde \chi_j \simeq \hbar \I(\underline\lambda_a)$. $\I(\underline\lambda_a)$ is the frequency of the rotation of the flow around $\theta_*$ and $|\R(\underline\lambda_a)|$ is the inverse of the ``life duration'' of the flow around $\theta_*$ (if $\R(\underline\lambda_a)<0$, $1/|\R(\underline\lambda_a)|$ is the characteristic duration of the fall on $\theta_*$ and if $\R(\underline\lambda_a)>0$ it is the characteristic duration of the escape from the neighbourhood of $\theta_*$). If the flow rotates around $\theta_*$ with a frequency tuned with the quantum transition frequency $\frac{\tilde \chi_i-\tilde\chi_j}{\hbar}$, it induces a strong transition $|Z\mu_i,\theta_*\rangle \to |Z\mu_j,\theta_*\rangle$ (if the duration of the rotation is sufficiently large i.e. $\R(\underline\lambda_a)\simeq 0$) as the same thing than the oscillations of an electromagnetic field with the same tuned frequency.

\subsection{Perturbative expansion}
\begin{prop}
  Let $(\Gamma,\mu,\varphi^t,\mathcal H,H)$ be a conservative driven quantum system, $\theta_* \in \Gamma$ be a fixed or cyclic point of $\varphi^t$, $\{\tilde \chi_i\}_i$ be the fundamental quasienergies associated with $\theta_*$ and $\{|Z\mu_i,\theta\rangle\}_i$ be the associated quasienergy states. Let $\Sp(F^\mu\partial_\mu) \ni \lambda \to H_\lambda \in \mathcal L(\mathcal H)$ be such that $H(\theta) = H(\theta_*) + \sum_{\lambda\not=0} f_\lambda^\alpha(\theta) H_{\lambda \alpha}$ (with $f_\lambda^\alpha(\theta_*) = 0$). $\alpha$ runs on the degeneracy of $\lambda$, the summation on $\alpha$ is implicit.  We suppose that $\exists \epsilon>0$ such that $|\langle Z\mu_j,\theta_*|H_{\lambda \alpha}|Z\mu_i,\theta_* \rangle| < \epsilon$, $\forall \lambda \not=0$. We have then
\begin{eqnarray}
|Z\mu_i,\theta \rangle & = & |Z\mu_i,\theta_* \rangle \nonumber \\
& & \quad + \sum_{\lambda \not= 0} \sum_j f_\lambda^\alpha(\theta) \frac{\langle Z\mu_j,\theta_*|H_{\lambda \alpha}|Z\mu_i,\theta_*\rangle}{\tilde \chi_i - \tilde \chi_j + \ihbar \lambda} |Z\mu_j,\theta_*\rangle \nonumber \\
& & \quad + \mathcal O(\epsilon^2) 
\end{eqnarray}
\end{prop}

\begin{proof}
We set
\begin{equation}
|Z\mu_i,\theta \rangle = |Z\mu_i,\theta_*\rangle + \sum_{\lambda \not=0} f_\lambda^\alpha(\theta)|\zeta_{i\lambda\alpha} \rangle
\end{equation}
with $|\zeta_{i\lambda\alpha}\rangle \in \mathcal H$. The eigenequation $(-\ihbar F^\mu \partial_\mu + H(\theta))|Z\mu_i,\theta \rangle = \tilde \chi_i |Z\mu_i,\theta \rangle$ becomes
\begin{eqnarray}
& &\left(-\ihbar F^\mu \partial_\mu + H(\theta_*) +\sum_{\lambda\not=0} H_{\lambda \alpha} f_\lambda^\alpha(\theta)\right)\left(|Z\mu_i,\theta_*\rangle+ \sum_{\lambda \not=0} f_\lambda^\alpha(\theta)|\zeta_{i\lambda\alpha} \rangle \right) \nonumber \\
& & \qquad = \tilde \chi_i \left(|Z\mu_i,\theta_*\rangle + \sum_{\lambda \not=0} f_\lambda^\alpha(\theta)|\zeta_{i\lambda\alpha} \rangle \right)
\end{eqnarray}
\begin{eqnarray}
& & -\ihbar \sum_{\lambda \not=0} \lambda f_\lambda^\alpha(\theta) |\zeta_{i\lambda\alpha} \rangle  + \sum_{\lambda\not=0} f_\lambda^\alpha(\theta)H(\theta_*)| \zeta_{i\lambda\alpha} \rangle  \nonumber \\
& & + \sum_{\lambda\not=0} f_\lambda^\alpha(\theta)H_{\lambda\alpha} |Z\mu_i,\theta_*\rangle + \sum_{\lambda,\nu\not=0} f_{\lambda+\nu}^{\alpha+\beta}(\theta) H_{\lambda\alpha} |\zeta_{i\nu\beta} \rangle \nonumber \\
& & \qquad = \tilde \chi_i \sum_{\lambda \not=0} f_\lambda^\alpha(\theta)|\zeta_{i\lambda\alpha} \rangle
\end{eqnarray}
\begin{eqnarray}
  & & \sum_{\lambda \not=0} f_\lambda^\alpha(\theta) \left(-\ihbar \lambda |\zeta_{i\lambda\alpha} \rangle + H(\theta_*)| \zeta_{i\lambda\alpha} \rangle + H_{\lambda\alpha} |Z\mu_i,\theta_*\rangle \right. \nonumber \\
  & & \qquad \left. + \sum_{\nu\not=0,\lambda}H_{\lambda-\nu,\alpha-\beta} |\zeta_{i\nu\beta} \rangle \right) \nonumber \\
& & \qquad \qquad = \tilde \chi_i \sum_{\lambda \not=0} f_\lambda^\alpha(\theta)|\zeta_{i\lambda\alpha} \rangle
\end{eqnarray}
It follows that
\begin{equation}
\left(H(\theta_*)-\tilde \chi_i -\ihbar \lambda \right) | \zeta_{i\lambda\alpha} \rangle + H_{\lambda\alpha} |Z\mu_i,\theta_*\rangle +  \sum_{\nu\not=0,\lambda}H_{\lambda-\nu,\alpha-\beta} |\zeta_{i\nu\beta} \rangle = 0
\end{equation}
Since $(|Z\mu_j,\theta_*\rangle)_j$ is a basis of $\mathcal H$ we set $|\zeta_{i\lambda\alpha} \rangle = \sum_j \zeta_{ij\lambda\alpha} |Z\mu_j,\theta_* \rangle$ and then
\begin{eqnarray}
& & \sum_j \zeta_{ij\lambda\alpha} \left(\tilde \chi_j-\tilde \chi_i -\ihbar \lambda \right) | Z\mu_j,\theta_* \rangle \nonumber \\
& & \qquad + H_{\lambda\alpha} |Z\mu_i,\theta_*\rangle +  \sum_{\nu\not=0,\lambda} \sum_j \zeta_{ij\nu\beta} H_{\lambda-\nu,\alpha-\beta} |Z\mu_j,\theta_* \rangle = 0
\end{eqnarray}
By projection of this equation on $\langle Z\mu_j,\theta_*|$ we find
\begin{eqnarray}
& & \zeta_{ij\lambda\alpha} \left(\tilde \chi_j-\tilde \chi_i -\ihbar \lambda \right) + \langle Z\mu_j,\theta_*|H_{\lambda\alpha}|Z\mu_i,\theta_*\rangle \nonumber \\
& & \qquad + \sum_{\nu\not=0,\lambda} \sum_k \zeta_{ik\nu\beta} \langle Z\mu_j,\theta_*|H_{\lambda-\nu,\alpha-\beta} |Z\mu_k,\theta_* \rangle = 0
\end{eqnarray}
It follows
\begin{eqnarray}
\zeta_{ij\lambda\alpha} & = & \frac{\langle Z\mu_j,\theta_*|H_{\lambda\alpha}|Z\mu_i,\theta_*\rangle}{\tilde \chi_i-\tilde \chi_j+\ihbar \lambda} \nonumber \\
& & \quad  + \sum_{\nu\not=0,\lambda} \sum_k \zeta_{ik\nu\beta} \frac{\langle Z\mu_j,\theta_*|H_{\lambda-\nu,\alpha-\beta} |Z\mu_k,\theta_* \rangle}{\tilde \chi_i-\tilde \chi_j+\ihbar \lambda} \\
& = & \frac{\langle Z\mu_j,\theta_*|H_{\lambda\alpha}|Z\mu_i,\theta_*\rangle}{\tilde \chi_i-\tilde \chi_j+\ihbar \lambda} \nonumber \\
& &   + \sum_{\nu\not=0,\lambda} \sum_k \frac{\langle Z\mu_k,\theta_*|H_{\nu\beta}|Z\mu_i,\theta_*\rangle}{\tilde \chi_i-\tilde \chi_k+\ihbar \nu} \frac{\langle Z\mu_j,\theta_*|H_{\lambda-\nu,\alpha-\beta} |Z\mu_k,\theta_* \rangle}{\tilde \chi_i-\tilde \chi_j+\ihbar \lambda} \nonumber \\
& &  + \mathcal O(\epsilon^3)
\end{eqnarray}
\end{proof}
Moreover we have (see \cite{Budisic})
\begin{eqnarray}
  & & f_\lambda^\alpha(\theta) \langle Z\mu_j,\theta_*|H_{\lambda\alpha}|Z\mu_i,\theta_*\rangle \nonumber \\
  & & \qquad = \lim_{T \to +\infty} \frac{1}{T} \int_0^T e^{-\lambda t} \langle Z\mu_j,\theta_*|H(\varphi^t(\theta))|Z\mu_i,\theta_*\rangle dt
\end{eqnarray}
permitting to find the element of the decomposition $f_\lambda^\alpha H_{\lambda \alpha}$.\\

For a stroboscopic driven quantum system, we have for a fixed point $\theta_*$:
\begin{eqnarray}
  |Z\mu_i,\theta \rangle & = & |Z\mu_i,\theta_* \rangle \nonumber \\
  & & \qquad + \sum_{\lambda\not=0} \sum_j f_\lambda^\alpha(\theta) \frac{\langle Z\mu_j,\theta_*|U_{\lambda \alpha}|Z\mu_i,\theta_*\rangle}{e^{-\imath (\tilde \chi_i+\imath \lambda)} - e^{-\imath \tilde \chi_j}}|Z\mu_j,\theta_* \rangle \nonumber \\
  & & \qquad + \mathcal O(\epsilon^2)
\end{eqnarray}
with $U(\theta) = U(\theta_*) + \sum_{\lambda\not=0} f_\lambda^\alpha(\theta) U_{\lambda \alpha}$ and $|\langle Z\mu_j,\theta_*|U_{\lambda \alpha}|Z\mu_i,\theta_* \rangle|< \epsilon$; $f_\lambda^\alpha(\theta) U_{\lambda \alpha} = \lim_{N\to+\infty} \frac{1}{N} \sum_{n=0}^{N-1} e^{-\lambda n} U(\varphi^n(\theta))$.\\

We have the same comments that for the local expansion, with a resonance phenomenon if $\tilde \chi_i - \tilde \chi_j \simeq \hbar \I(\lambda)$ ($\R(\lambda)=0$ since we consider a conservative system). If the Koopman operator presents an absolutely continuous spectrum then resonances are strongly likely.

\end{document}